\documentclass[journal,final]{IEEEtran}
\IEEEoverridecommandlockouts

\usepackage{amsmath,amssymb,amsthm}
\usepackage{mathtools}
\usepackage{bm}
\usepackage{combelow}
\usepackage{xcolor}
\usepackage[normalem]{ulem}
\usepackage{multirow}
\usepackage{bbold}
\usepackage{dsfont}
\usepackage{bigdelim}
\usepackage{tikz}

\newtheorem{theorem}{Theorem}
\newtheorem{proposition}{Proposition}
\newtheorem{definition}{Definition}
\newtheorem{remark}{Remark}
\newtheorem{corollary}{Corollary}
\newtheorem{problem}{Problem}

\newtheorem{lemma}{Lemma}
\newtheorem{example}{Example}
\newtheorem{fact}{Fact}

\DeclareMathOperator{\spec}{spec}
\DeclareMathOperator{\rank}{rank}

\DeclareMathOperator{\diag}{diag}
\DeclareMathOperator{\im}{Im}
\newcommand{\R}{\mathbb{R}}

\begin{document}

\title{Machines that Predict Trajectories from Templates
}

\author{C. De Persis and P. Tesi 
\thanks{C. De Persis is with ENTEG, 
University of Groningen, 9747 AG Groningen, The Netherlands. 
Email: {\tt\small c.de.persis@rug.nl}. \\
P. Tesi is with DINFO, University of Florence, 50139 Firenze, Italy. 
E-mail: {\tt\small pietro.tesi@unifi.it}.}
}

\maketitle

\begin{abstract}
We study trajectory prediction from libraries of stored output templates.
Given the past of an unknown trajectory, the goal is to predict its future
without identifying the state-space model that generated it. We show that
libraries of trajectories generated by one or more dynamical systems define
behavioral spaces that can be used as prediction mechanisms. For linear
systems, we characterize exact prediction in terms of continuation maps,
behavioral containment, and spectral conditions on output-visible
eigenvalues. We also analyze robustness to noisy observations and noisy
libraries, derive error bounds for out-of-library trajectories, and show
how interconnection constraints can compose template libraries into new
behavioral spaces with emergent modes. Finally, we extend the framework to
nonlinear systems whose output trajectories are contained in, or immersed
into, finite-dimensional linear behaviors. These results provide a theory
of template-based prediction machines capable of generalizing beyond the
stored trajectories and, in some cases, beyond the systems that generated
them.
\end{abstract}


\section{Introduction}

\textbf{Solving complex tasks using templates.}
Using stored examples, or templates, to solve complex tasks is an
intuitive and powerful idea that appears in many areas of science and
engineering. In classification and pattern recognition \cite{Vapnik1995,Burges1998}, new objects are
assigned labels by comparison with stored examples from known categories.
In instance-based regression \cite{Aha1991}, predictions are obtained by exploiting
similarities with previously observed samples. The appeal of this
paradigm is that it can support fast inference and decision making: rather
than constructing a full model of the phenomenon from scratch, one reuses
representative examples stored in memory.

In systems and control, related ideas appear in several contexts. For
example, in fault detection, online measurements can be compared with
stored nominal trajectories to detect abnormal behavior \cite{Gao2015}. More generally,
one may ask whether stored dynamical behaviors can be used to infer
stability properties, predict future trajectories, or supervise control
systems without first identifying a complete parametric model. Despite the
importance of these tasks, a general theory of control, estimation and prediction from trajectory
templates is still largely undeveloped.

\textbf{Predicting dynamical behaviors from templates.}
In this paper, we study the problem of predicting the behavior of an
unknown dynamical system from templates. We assume that the portion
of an output trajectory is observed, and we aim to predict its future
evolution using a stored library of sample trajectories. Concretely, we
assume access to collections of finite length output windows generated by a
finite number of dynamical systems, and we address the task of predicting
the future of a new trajectory from its observed past without knowing
which system is generating it. 

This viewpoint casts trajectory prediction as a \emph{dictionary-based}
inference problem. The library provides a repertoire of representative
behaviors, and prediction amounts to completing a partially observed
output sequence using this repertoire, rather than estimating the
underlying state-space matrices. The use of multiple template systems is
not only a way to describe multi-mode behavior. It also provides a
mechanism for approximating \cite{korda2018linear,Klus-Koopman2020,mezic2026koopman}, and in special cases exactly representing \cite{fliess2005finite,Iacob-Koopman2025},
nonlinear behaviors: a finite collection of linear templates can
approximate, or even span, the output trajectories of a nonlinear system
when these trajectories are close to, or contained in, a finite-dimensional
linear behavior. The idea may be similarly relevant to the prediction of linear time-varying systems \cite{sasfi2025great}.

Predicting the behavior of dynamical systems is a central problem across
engineering and applied science. Accurate forecasts enable trajectory
tracking and anticipative control, support anomaly detection and
monitoring, and play a key role in predictive maintenance by revealing
early signs of performance degradation or impending faults 
\cite{turan2021data,shi2022data,eun2023data,disaro2024equivalence}. They also
form the basis of predictive control, where future system evolutions
are repeatedly predicted over a receding horizon to compute control
actions \cite{deepc,BerberichMPC2022,hewing2020learning,krishnan2021direct}. Beyond these
engineering motivations, trajectory prediction from partial observations
can also be viewed as a form of sequence completion \cite{forgione2023from}, a computational task
that is central in modern sequence models, 
either  based on transformers \cite{abella2024asymptotic}, \cite{altafini2026multistability} or 
selective state space models
\cite{honarpisheh2026generalization}, \cite{casti2025selection}, here specialized to the
structure of dynamical systems and physical time series.

\textbf{From stored trajectories to generative behavioral spaces.}
The template-based prediction paradigm has a distinctive dynamical
feature. A stored trajectory is not merely an isolated example. Because it
is generated by an underlying dynamical system, it encodes information
about prototypical evolutions such as exponential modes, oscillations,
decay rates, and transient patterns. Thus, a trajectory library may define
a \emph{generative space of admissible behaviors}: once a new past trajectory is
observed, it is compared not only with individual examples, but with the
space of possible continuations generated by the library.

This observation is consistent with behavioral data-driven control. In
particular, the fundamental lemma of Willems \cite{willems2005} and its variants \cite{henk-fl-multiple-traj,alsalti2023design,Mesbahi-ExtWillems2021,shang2024willems,lazar2025product} show that,
for a single sufficiently excited LTI system, measured trajectories can
span the behavior of the system and can therefore be used as a surrogate
for a parametric model in simulation, prediction, and control
\cite{Markov08,de2019formulas,ARC-survey,Ajeleye-abstr2023,Dorfler-bridging2023}. The present work takes this viewpoint one step further.
We study how libraries collected from
template systems can be \emph{assembled and used} for prediction, and we analyze
when such libraries \emph{generalize} beyond the systems that generated the
stored data.

\textbf{Main contributions.}
This paper develops a theory of prediction machines built from
trajectory templates (see Fig.~\ref{fig:out_of_library_templates} for an illustration of such template-based machines).
Given a past output window, these machines produce
a future output window by exploiting the behavioral information encoded
in the library, without first identifying the state-space model that
generated the observed trajectory. The main contributions are:

\textbf{(i) Prediction in single- and multi-system libraries.}
Sections~\ref{sec:singlemode_continuation} and~\ref{sec:multimode}
develop the basic prediction mechanism. We first recall the continuation
map associated with a single linear behavior (Lemma \ref{lem:key}) and then show that an aggregate
library formed by several template systems can predict trajectories
generated by any system in the library, without identifying which system
is active (Theorem \ref{thm:mixed-prediction}). This provides the foundational mechanism by which past
measurements are mapped to future outputs through a template library.

\textbf{(ii) Robustness to noisy data and libraries.}
Section~\ref{sec:mix-noise} studies how prediction degrades in the
presence of noisy observations and noisy libraries. The results quantify
the effect of measurement noise, library perturbations, and regularized
least-squares predictors (Theorem \ref{thm:mix-noisyH-decomp}).

\textbf{(iii) Out-of-library generalization.}
Section~\ref{sec:out-of-library} addresses the central question of how a
trajectory library can predict beyond the nominal template systems. We
derive an error bound in terms of the distance between the target
trajectory and the behavioral space generated by the library
(Proposition~\ref{prop:outlib-decomp}), and we provide conditions under
which exact prediction is recovered
(Theorem~\ref{thm:mix-ool-exact}). For linear systems, these conditions take a spectral form, requiring the
output-visible eigenvalues of the target system to be covered by those
represented in the library
(Proposition~\ref{prop:mix-sentinels}).

\textbf{(iv) Compositional generalization via interconnections.}
Section~\ref{subsec:feedback_composed_libraries} shows that the
generalization capability of a library is not limited to the modes already
present in its atomic components. By imposing compatibility constraints
associated with prescribed interconnection topologies, template
trajectories can be composed into new behavioral spaces describing
interconnected systems. This produces emergent modes outside the original
atomic libraries and enables the prediction of distinctly new trajectories
(Theorem~\ref{thm:exact_prediction_interconnection_library} and
Proposition~\ref{prop:interconnection_spectral}).

\textbf{(v) Nonlinear trajectories from linear templates.} Finally,
Section~\ref{sec:output-containment-immersion} extends the theory to
nonlinear systems whose output trajectories are contained in, or immersed
into, finite-dimensional linear behaviors \cite{fliess2005finite}. This provides a mechanism by
which linear template libraries can predict nonlinear trajectories,
provided that the relevant output behavior is captured by the ambient
linear representation (Proposition \ref{prop:nonlinear-output-containment}).


\section{Trajectory prediction problem}
\label{sec:multimode_intro}

\subsection{Framework and problem definition}

Let $\mathcal{I}:=\{1,\dots,Q\}$ be a finite index set, and for each $i\in\mathcal{I}$ consider
a discrete-time LTI system\footnote{
While we focus on autonomous systems to keep the core ideas clear, 
we envision that the framework can be expanded to include inputs. }
\begin{equation}\label{eq:multimode_sys}
x_{t+1} = A_i x_t , \quad 
y_t = C_i x_t
\end{equation}
where $t \in \mathbb N$,  $x_t\in\mathbb{R}^{n_i}$ is the state and $y_t\in\mathbb{R}^p$ is the vector of observations.
(The state dimension may change over $\mathcal{I}$.)

Assume that for each system $i\in\mathcal{I}$ we have collected $N_i$ sample trajectories of
length $T$, and stacked them as columns into a \emph{template matrix}
\begin{equation}\label{eq:mode_library}
H_Y^{(i)} :=
\begin{bmatrix}
Y^{(i,1)} & \cdots & Y^{(i,N_i)}
\end{bmatrix}, 
\quad
Y^{(i,j)} :=
\begin{bmatrix}
y_0^{(i,j)}\\
y_1^{(i,j)}\\
\vdots\\
y_{T-1}^{(i,j)}
\end{bmatrix}\in\mathbb{R}^{Tp}.
\end{equation}
In what follows, we interpret $H_Y^{(i)}$ as a library of
length-$T$ output windows, and we use it as a dictionary of typical behaviors for  system $i$.
For convenience, we also define the aggregated library
or aggregated template matrix
across all systems,
\begin{equation}\label{eq:global_library}
H_Y :=
\begin{bmatrix}
H_Y^{(1)} & H_Y^{(2)} & \cdots & H_Y^{(Q)}
\end{bmatrix}\in\mathbb{R}^{Tp\times N},
\,\,
N:=\sum_{i=1}^Q N_i.
\end{equation}

\begin{problem} \label{prob:def}
Let $Y^\star\in\mathbb{R}^{Tp}$ be any trajectory of length $T$ generated by some system whose output has the same dimension as the systems in $\mathcal{I}$. Partition $Y^\star$ as
\begin{equation}\label{eq:new_traj_split_intro}
Y^\star=
\begin{bmatrix}
Y_{\mathrm{past}}\\[1mm]
Y_{\mathrm{fut}}
\end{bmatrix}, \quad
Y_{\mathrm{past}}\in\mathbb{R}^{rp},\quad
Y_{\mathrm{fut}}\in\mathbb{R}^{(T-r)p},
\end{equation}
where $r<T$. Assuming that only the past window $Y_{\mathrm{past}}$ is observed, the goal is to
predict the corresponding future window $Y_{\mathrm{fut}}$ using the library
$\{H_Y^{(i)}\}_{i=1}^Q$ (or equivalently the aggregated matrix $H_Y$). \qed  
\end{problem}

\textbf{Sequence completion} The problem at hand is prototypical in engineering problems where we must complete
a trajectory from partial observations. A classical example is \emph{trajectory prediction}:
given the recent history of a system output, one seeks to predict the upcoming evolution
to anticipate control actions or detect anomalies.
At the same time, the problem can be read through the lens of \emph{sequence completion}:
we are given a prefix $Y_{\mathrm{past}}$ and we wish to infer a plausible continuation
$Y_{\mathrm{fut}}$.

\textbf{Library of behaviors} In this perspective, the Hankel or template matrix $H_Y$ plays the role of a \emph{library} (or
dictionary) of representative behaviors: each column is a sample length-$T$ behavior, and
prediction amounts to completing a partially observed sequence by combining (in a suitable
sense) the sample behaviors in the dictionary. This connection provides an intuitive
interpretation of data-driven prediction: as the problem is that of predicting trajectories that may not have been generated by any of the systems in the library, one leverages a
catalog of empirical trajectories and uses it directly for completion rather than estimating a specific $(A_i,C_i)$.

Problem \ref{prob:def} does not assume  that the system that generates the data belongs to the family $\mathcal{I}$ of nominal systems. However, in the first part of the paper, to lay the foundation for a solution in the general case, we do consider the simpler problem in which the trajectory to be completed is generated by one of the template systems. Later in Section \ref{sec:out-of-library} we will extend the study to the case where the system generating the data does not belong to such library. 
While the most consequential results of the article are given 
in Section \ref{sec:out-of-library}--\ref{sec:output-containment-immersion}, a few propaedeutic principles are discussed in the preceding sections.

\subsection{Basic notions and working assumptions}

\subsubsection{Linear equations, images, and kernels}
We will often consider systems of linear equations of the form
$
H\alpha=v,
$
where $H$ and $v$ are given and $\alpha$ is the unknown coefficient vector.
If both $H$ and $v$ are real, then
solutions are understood over $\mathbb R$; if either $H$ or $v$ is complex,
then solutions are understood over $\mathbb C$.
The same convention applies to statements of the form
$v\in\im(H)$ and $\alpha\in\ker(H)$. The former is understood as the
solvability of $H\alpha=v$ for some coefficient vector $\alpha$, whereas the
latter is understood as the equation $H\alpha=0$, with the field determined
by $H$.

In particular, when $H$ and $v$ are real, if the system $H\alpha=v$ is
solvable over $\mathbb C$, then it also admits a real solution. Thus, in this
case, $H\alpha=v$ and $v\in\im(H)$ are understood over $\mathbb R$.
If $H$ is real, then $\alpha\in\ker(H)$ is understood over
$\mathbb R$.

\subsubsection{Observability index}
A point that will become clear in the sequel is that trajectory prediction does not require
identifying the system, nor does it require reconstructing its internal
state. The only requirement we need is that the observed past window is
long enough with respect to certain observability properties of the underlying dynamics. This ensures
that the data carry sufficient information about the observable subspace
that generates the observations. To this end, we recall the notion of \emph{observability index}.

Consider a discrete-time LTI system
\begin{equation}
  x_{t+1} = A x_t, \quad y_t = C x_t,
  \label{eq:ss}
\end{equation}
with $x_t \in \R^{n}$ and $y_t \in \R^p$. For an integer $\ell \ge 1$, the observability matrix of order $\ell$ is
\begin{equation} \label{eq:obs-k}
  \mathcal{O}_\ell(A,C)
  :=
  \begin{bmatrix}
    C \\ CA \\ \vdots \\ CA^{\ell-1}
  \end{bmatrix} \in \R^{\ell p \times n}.
\end{equation}

We recall the following result from linear systems theory:\footnote{Lemma (5.9.27) in Ruberti, Isidori. Teoria dei Sistemi. Bollati Boringhieri.} 
\begin{lemma}
There exists an integer $s\le n$ such that 
\begin{itemize}
\item[(i)] $\ker  \bigl(\mathcal{O}_k(A,C)\bigr)\supset \ker\bigl(\mathcal{O}_{k+1}(A,C)\bigr)$ for all $k<s$.
\item[(ii)] $\ker  \bigl(\mathcal{O}_k(A,C)\bigr)= \ker\bigl(\mathcal{O}_{k+1}(A,C)\bigr)$ for all $k\ge s$.\qed 
\end{itemize}  
\end{lemma}
We call this integer $s$ the \emph{observability index}. By the Fundamental Theorem of Linear Algebra, the relations above can be restated as 
 (i) $\im  \bigl(\mathcal{O}_k^\top(A,C)\bigr)\subset \im\bigl(\mathcal{O}^\top_{k+1}(A,C)\bigr)$ for all $k<s$ and (ii) $\im  \bigl(\mathcal{O}_k^\top(A,C)\bigr)= \im\bigl(\mathcal{O}^\top_{k+1}(A,C)\bigr)$ for all $k\ge s$. For convenience, we equivalently  formalize the definition of observability index as follows:
\begin{definition}[Observability index]  \label{def.obs.index}
The \emph{observability index} $s$ of the pair $(A,C)$ is defined as the smallest integer
$s \ge 1$ such that
$\rank\bigl(\mathcal{O}_s(A,C)\bigr) = \rank\bigl(\mathcal{O}_{s+1}(A,C)\bigr).$ \qed  
\end{definition}
We \emph{do not} assume $(A,C)$ is observable in the classical sense. In general,
$\rank(\mathcal{O}_s(A,C))$ may be strictly smaller than $n$. The classical notion of observability corresponds to the special case
$\rank(\mathcal{O}_s(A,C)) = n$. 

\subsubsection{Richness of data}
In our context, we predict dynamical behaviors from templates. Thus, it is natural to require
that  template trajectories encode enough information on the systems that generate them. This leads to a notion of
\emph{richness of data}.

From \eqref{eq:obs-k}, the $j$-th template trajectory from system $i$ can be expressed as
\[
  Y^{(i,j)} :=
  \begin{bmatrix}
    y_0^{(i,j)} \\ y_1^{(i,j)} \\ \vdots \\ y_{T-1}^{(i,j)}
  \end{bmatrix}
  =
  \begin{bmatrix}
    C_i \\ C_iA_i \\ \vdots \\ C_iA_i^{T-1}
  \end{bmatrix} x_0^{(i,j)}
  =\mathcal{O}_T(A_i,C_i)\, x_0^{(i,j)},
\] 
with $j=1,\dots,N_i$,
for some initial condition $x_0^{(i,j)}\in\R^{n_i}$.
Stacking these trajectories as columns yields the data matrix
\begin{equation}
  H_Y^{(i)} = \mathcal{O}_T(A_i,C_i) X_i,
\end{equation}
where $H_Y^{(i)}$ was defined in \eqref{eq:mode_library} and 
\begin{equation} \label{eq:Xinit}
  X_i := \begin{bmatrix} x_0^{(i,1)} & \dots & x_0^{(i,N_i)} \end{bmatrix}
  \in \R^{n_i \times N_i}.
\end{equation}

\begin{definition}[Data richness]
We say the data are \emph{rich} if $\rank X_i = n_i$ for every $i\in \mathcal{I}$. \qed
\end{definition}

Under this condition, the matrix $H_Y^{(i)}$ is a non-parametric representation of the system \eqref{eq:multimode_sys}, as stated in the result below, whose straightforward proof is omitted. 

\begin{lemma}[Non-parametric representation]\label{lem:willems.aut}
 Any real linear combination of the columns of $H_Y^{(i)}$ is a $T$-long output trajectory of \eqref{eq:multimode_sys}. Conversely, if $\rank X_i = n_i$, any trajectory $Y^\star\in \R^{pT}$ of 
\eqref{eq:multimode_sys} belongs to the image of $H_Y^{(i)}$.\qed
\end{lemma}

Yet, we will not use this notion to perform any form of system identification. In fact, identifying the individual
dynamics is not required procedurally: the approach we develop leverages the recorded trajectories directly,
treating them as a library from which future behaviors are predicted.


\section{Continuation map and prediction for single systems}
\label{sec:singlemode_continuation}

In this section we focus on the \emph{single-system} setting, namely the case in which both the
template library and the trajectory to be predicted are generated by \emph{one} (unknown) LTI
system~\eqref{eq:ss}. 
This ideal setting makes transparent a few key concepts that are used throughout the article and the mechanism that enables exact prediction from templates. 

We begin by introducing the notion of \emph{continuation map} and explain its role in the prediction problem for a single system \eqref{eq:ss}, in order to highlight the key algebraic mechanisms behind exact prediction from templates. In the next section, we will
extend the analysis to the multi-system setting, where the trajectory to be predicted belongs to one system in a finite
family. 

\subsection{Continuation map and exact prediction}

An important property  that we will extensively use in this paper is Lemma \ref{lem:key} below.
We first recall this fact. 

\begin{fact}\label{fact:rank.stability}
Let $s$ be the observability index of the pair $(A,C)$. 
Then, for any pair of  integers $\ell,r$ such that  $\ell\ge r\ge s$, there exists a matrix $Q_{\ell}\in \R^{p\times pr}$ such that $CA^{\ell}= Q_{\ell} \mathcal{O}_r(A,C)$. 
\end{fact}

\begin{proof} See Appendix A. 
\end{proof}

\begin{lemma}[Continuation map and latent representations]\label{lem:key}
Consider system~\eqref{eq:ss} and let $s$ denote its observability index.
Fix $r$ such that $s\le r<T$.
Let $H_Y$ be a library of $T$-long trajectories generated by~\eqref{eq:ss}, so that
$
H_Y = \mathcal{O}_T(A,C) X,
$
where $X$ 
collects the initial conditions
that generated the template trajectories (\emph{cf.} \eqref{eq:Xinit}). Partition $H_Y$ as
\[
H_Y = \begin{bmatrix} H_p\\[1mm] H_f\end{bmatrix}\in\R^{Tp\times N},
\quad
H_p\in\R^{rp\times N},\;\; H_f\in\R^{(T-r)p\times N}.
\]
Then:
\begin{itemize}
\item[(i)] \emph{(Continuation map).}
There exists a matrix $L_\star$ 
depending only on $(A,C,T,r)$ such that
\begin{equation}\label{eq:Lstar-system}
H_f = L_\star H_p,
\end{equation}
and for every length-$T$ trajectory $Y^\star$ of~\eqref{eq:ss}
we have
\begin{equation}\label{eq:Lstar-traj}
Y_{\mathrm{fut}} = L_\star\,Y_{\mathrm{past}}.
\end{equation}

\item[(ii)] \emph{(Exact prediction via latent variables).}
Under the data richness assumption $\rank(X)=n$, for every trajectory $Y^\star$ there exists at least one
$g^\star$ such that
\begin{equation}\label{eq:g-feasible}
Y_{\mathrm{past}}=H_p g^\star.
\end{equation}
Moreover, for \emph{any} $g$ satisfying \eqref{eq:g-feasible}, the prediction
\begin{equation}\label{eq:g-pred}
\widehat Y_{\mathrm{fut}}:=H_f g
\end{equation}
is exact:
\begin{equation}\label{eq:g-exact}
\widehat Y_{\mathrm{fut}}=Y_{\mathrm{fut}}.
\end{equation}
In particular, the predicted future is independent of the particular feasible $g$.

\item[(iii)] \emph{(Data-based computation of the continuation map and equivalence with the $g$-based prediction).}
If $\rank(X)=n$, any solution $L$ of the data equation
\begin{equation}\label{eq:L-from-data}
H_f = L H_p
\end{equation}
coincides with $L_\star$ on $\im(H_p)$ and therefore yields the same prediction on every admissible past window:
\begin{equation}\label{eq:L-exact}
Y_{\mathrm{fut}} = L Y_{\mathrm{past}}.
\end{equation}
Moreover,
the two predictors, namely the one based on the map $L$  and the one based on the latent variable $g$, agree:
$
H_f g = L(H_p g) = L Y_{\mathrm{past}}.
$
\end{itemize}
\end{lemma}

\begin{proof}
For compactness set $\mathcal{O}_k:=\mathcal{O}_k(A,C)$. Consider the partition
\[
\mathcal{O}_T=\begin{bmatrix}\mathcal{O}_r\\[1mm] M\end{bmatrix},
\quad
M:=\begin{bmatrix}CA^r\\ CA^{r+1}\\ \vdots\\ CA^{T-1}\end{bmatrix}.
\]
From $H_Y=\mathcal{O}_T X$ we obtain
\begin{equation}\label{eq:HpHf-factor}
H_p=\mathcal{O}_r X,\quad H_f=M X.
\end{equation}

\emph{(i).}
By Fact \ref{fact:rank.stability}, for every $\ell\ge r$ there exists $Q_\ell\in\R^{p\times rp}$ such that $CA^\ell=Q_\ell \mathcal{O}_r$.
Stacking all these matrices $Q_\ell$ for $\ell=r,\dots,T-1$ yields a matrix $L_\star$ such that
$
M=L_\star \mathcal{O}_r.
$
Combining \eqref{eq:HpHf-factor} with the identity $M=L_\star \mathcal{O}_r$ gives \eqref{eq:Lstar-system}.
Also, for any trajectory $Y^\star=\mathcal{O}_T x$, we have $Y_{\mathrm{past}}=\mathcal{O}_r x$ and
$Y_{\mathrm{fut}}=M x=(L_\star \mathcal{O}_r)x=L_\star Y_{\mathrm{past}}$, proving \eqref{eq:Lstar-traj}.

\emph{(ii).}
Let $Y^\star$ be any trajectory. Then $Y_{\mathrm{past}}=\mathcal{O}_r x$ for some $x\in\R^n$.
Since $\rank(X)=n$, we have $\im(X)=\R^n$, hence there exists $g^\star$ such that $x=Xg^\star$.
Using \eqref{eq:HpHf-factor}, we get
$Y_{\mathrm{past}} = \mathcal{O}_r (X g^\star) = H_p g^\star$,
which proves feasibility of \eqref{eq:g-feasible}.
Now take any $g$ satisfying $Y_{\mathrm{past}}=H_p g$.
Using \eqref{eq:Lstar-system} and \eqref{eq:Lstar-traj}, we obtain
$
H_f g = (L_\star H_p)g = L_\star (H_p g)=L_\star Y_{\mathrm{past}}=Y_{\mathrm{fut}},
$
which proves \eqref{eq:g-exact} and the independence from the choice of feasible $g$.

\emph{(iii).}
Let $L$ satisfy $H_f=LH_p$. Then 
$
L H_p = H_f \stackrel{\eqref{eq:Lstar-system}}{=}  L_\star H_p,
$ 
thus  $L$ and $L_\star$ coincide on $\im(H_p)$.
Under $\rank(X)=n$ every admissible past window satisfies $Y_{\mathrm{past}}\in\im(H_p)$,
hence $Y_{\mathrm{fut}}\stackrel{\eqref{eq:Lstar-traj}}{=}L_\star Y_{\mathrm{past}}=L Y_{\mathrm{past}}$, proving \eqref{eq:L-exact}.

Finally, if $Y_{\mathrm{past}}=H_p g$, then $H_f g = L(H_p g)=L Y_{\mathrm{past}}$.
\end{proof}

Some comments are in order.

\textbf{Existence and uniqueness of the continuation map.}
The matrix $L_\star$ in Lemma~\ref{lem:key}(i) represents a \emph{continuation map}: it specifies how a past output
window $Y_{\mathrm{past}}$ is extended into the corresponding future window $Y_{\mathrm{fut}}$.
Such a linear map exists as soon as $r\ge s$ (observability index).
However, $L_\star$ is not necessarily unique as a matrix:
different matrices can agree on the set of feasible past windows and still differ on vectors that
are \emph{not} realizable as past windows of the system.
What is unique is the \emph{continuation itself}:
for every admissible $Y_{\mathrm{past}}$, the value $Y_{\mathrm{fut}}$ is uniquely determined, and therefore
$L_\star Y_{\mathrm{past}}$ is unique, even if the matrix representation of $L_\star$ is not. 

\textbf{On data richness.}
Item (i) is fundamentally a property of the underlying dynamics: it can be phrased as
$M=L_\star \mathcal{O}_r$ and it holds independently of how the data were collected.
The richness assumption $\rank(X)=n$ is instead needed for two \emph{data-driven} purposes:
(1) It ensures \emph{feasibility} of the latent representation, i.e., that every admissible past window
can be written as $Y_{\mathrm{past}}=H_p g$ (Lemma~\ref{lem:key}(ii)). 
(2) It ensures that \emph{any} solution $L$ of the data equation $H_f=L H_p$ has the correct action on all
admissible past windows (Lemma~\ref{lem:key}(iii)).

Note, however, that exact prediction does not require data richness if one assumes $Y^\star \in \im(H_Y)$. Data richness is needed when we want to ensure feasibility over \emph{all} admissible trajectories of the system, or for computing $L$ from data. 
It is easy to verify that data richness is necessary to reconstruct $(A,C)$ from data, even only up to a similarity transformation, \emph{cf.} \cite[Chapter 9]{subspace-book2}.
 
\textbf{Encoder--decoder vs. direct predictors.} 
Lemma~\ref{lem:key} gives two equivalent ways to perform prediction in the single-mode case. The \emph{$g$-formulation} follows an encoder--decoder logic: one first encodes the observed past by fitting coefficients $g$ so that $Y_{\mathrm{past}}=H_p g$
and then decodes the future as $\widehat Y_{\mathrm{fut}}=H_f g$. The \emph{$L$-formulation} instead postulates a direct input--output (sequence-to-sequence) map and predicts by $\widehat Y_{\mathrm{fut}}=L\,Y_{\mathrm{past}}$.
While conceptually different, in the ideal setting with noiseless data the two viewpoints provide the same predictions. 
We will further compare the two approaches when dealing with noisy data in Section \ref{subsec:noiseLg}.

\textbf{Encoding as selection mechanism.}
Closely related results have appeared in the literature; in particular, we refer to \cite[Proposition 1]{Markov08}, which closely parallels Lemma~\ref{lem:key}(ii). There are, however, some differences. In \cite{Markov08}, and in the related literature, this result is framed within the behavioral approach. The matrix $H_Y$ is built from a single trajectory (typically input-output), and $g$ represents a sort of latent state. Indeed, in view of the identity $H_p = \mathcal{O}_r X$, solving $Y^*_{\mathrm{past}} = H_p g$ is equivalent to selecting an initial condition $x_g = X g$ (not necessarily unique) whose first $r$ outputs match the observed past. This interpretation is clearly possible here as well. However, our main interest is in the case where $H_Y$ is built from multiple trajectories collected from different systems, as we begin to discuss in Section~\ref{sec:multimode}. In this context, $g$ makes explicit the view of prediction as the \emph{selection and combination of templates} from a library. This interpretation will remain useful also in the study of out-of-library generalization.


\section{Trajectory prediction with templates} \label{sec:multimode}

In a setting where multiple systems contribute prototypical trajectories, there are (at least) two natural ways to leverage the available libraries. A first option is  \emph{separate} inference: for each system $i\in\mathcal{I}$ one fits coefficients using the submatrix of the template matrix $H_Y^{(i)}$ that corresponds to the observed trajectory $Y_{\mathrm{past}}$ (see \eqref{eq:new_traj_split_intro}),
and then predicts $Y_{\mathrm{fut}}^{(i)}$ via the remaining submatrix of  $H_Y^{(i)}$. One may then select among the candidate predictions using additional information (e.g., feasibility of the fit or residual size). This approach is conceptually simple and, in the ideal case where the trajectory indeed belongs to one of the library templates, it can be sufficient. 

In this paper we focus instead on the \emph{mixed} formulation, where all templates are aggregated into a single library
\begin{equation}\label{HY-mixed}
H_Y=\begin{bmatrix}H_Y^{(1)} & \cdots & H_Y^{(Q)}\end{bmatrix}
=\begin{bmatrix}H_p\\[1mm] H_f\end{bmatrix} 
\end{equation}
and prediction is performed by fitting a single coefficient vector $g\in\R^N$ against the mixed past matrix $H_p$ and decoding with the mixed future matrix $H_f$. The reason for emphasizing the mixed viewpoint is that it naturally supports \emph{cross-system} combinations: the predictor is allowed to synthesize a continuation by combining behaviors drawn from different templates in the library. This flexibility is not only convenient for analysis, but becomes essential in more complex scenarios considered later in the paper, notably when the trajectory to be predicted is generated by a system that is \emph{not} contained in the nominal library. In such out-of-library settings, accurate prediction may require blending template behaviors across  systems (rather than committing to a single one), and the mixed formulation provides the appropriate mathematical and algorithmic framework for doing so.

\subsection{Aggregate system and mixed template matrix}\label{subsec:mix-agg}
 
Consider a family of $Q$ discrete-time systems as in \eqref{eq:multimode_sys},
with common output dimension $p$.
For each  system $i$ we collect $N_i$ length-$T$ trajectories and form the template matrix
\[
  H_Y^{(i)}=\begin{bmatrix}Y^{(i,1)}&\cdots&Y^{(i,N_i)}\end{bmatrix}\in\R^{Tp\times N_i}.
\]
Each $H_Y^{(i)}$ can be factorized as $H_Y^{(i)}=\mathcal{O}_T(A_i,C_i)\,X_i$,
where $X_i\in\R^{n_i \times N_i}$ stacks the corresponding initial conditions, with $n_i$ the state dimension of the $i$-th system.

Let $n:=\sum_{i=1}^Q n_i$.
Introduce the extended state $x_k^{\mathrm{mix}}\in\R^{n}$ and define the aggregate pair
\begin{equation}\label{eq:agg-system-Q}
  x_{k+1}^{\mathrm{mix}} = A_{\mathrm{mix}} x_k^{\mathrm{mix}},\quad
  y_k = C_{\mathrm{mix}} x_k^{\mathrm{mix}},
\end{equation}
where
\[
  A_{\mathrm{mix}} := \mathrm{blkdiag}(A_1,\dots,A_Q), \quad 
  C_{\mathrm{mix}} := \begin{bmatrix}C_1&\cdots&C_Q\end{bmatrix}. \]
The corresponding order-$T$ observability matrix is
\[
\mathcal{O}_T(A_{\mathrm{mix}},C_{\mathrm{mix}})
  =
  \begin{bmatrix}
    C_1 & \cdots & C_Q\\
    C_1A_1 & \cdots & C_QA_Q\\
    \vdots & & \vdots\\
    C_1A_1^{T-1} & \cdots & C_QA_Q^{T-1}
  \end{bmatrix}. 
\]

Define the block-diagonal matrix of initial conditions
\begin{equation}\label{def:Xmix}
  X_{\mathrm{mix}} := \mathrm{blkdiag}(X_1,\dots,X_Q)\in\R^{n\times N}.
\end{equation}

\begin{lemma}[Factorization of the mixed library]\label{prop:agg-factorization}
The aggregate template matrix $H_Y$ admits the factorization
\[
  H_Y = \mathcal{O}_T(A_{\mathrm{mix}},C_{\mathrm{mix}}) X_{\mathrm{mix}}.
\]
\end{lemma}

\begin{proof}
A direct block multiplication gives
\begin{eqnarray}
&&  \mathcal{O}_T(A_{\mathrm{mix}},C_{\mathrm{mix}}) X_{\mathrm{mix}} \nonumber \\
&&  \quad = 
  \begin{bmatrix}
  \mathcal{O}_T(A_1,C_1)X_1 & \cdots &   \mathcal{O}_T(A_Q,C_Q)X_Q
  \end{bmatrix} \nonumber \\
 && \quad =
  \begin{bmatrix}
    H_Y^{(1)} & \cdots & H_Y^{(Q)}
  \end{bmatrix}
  =H_Y.
\end{eqnarray}
\end{proof}

\subsection{Exact prediction using the aggregate library}\label{subsec:mix-exact}

Let $s_{\mathrm{mix}}$ denote the observability index of $(A_{\mathrm{mix}},C_{\mathrm{mix}})$.
Fix $r<T$ and partition the aggregate library as in \eqref{HY-mixed}
\[
  H_Y=\begin{bmatrix}H_p\\[1mm] H_f\end{bmatrix}, \quad H_p \in \mathbb{R}^{rp \times N}, H_f \in \mathbb{R}^{(T-r)p \times N}.
\]
By Lemma~\ref{prop:agg-factorization}, Lemma~\ref{lem:key} (applied to the single LTI pair
$(A_{\mathrm{mix}},C_{\mathrm{mix}})$) yields the following prediction result using the aggregate library.

\begin{theorem}[Exact prediction via aggregate library]\label{thm:mixed-prediction}
Consider any length-$T$ trajectory $Y^\star$ generated by one (unknown) system $i^\star\in\mathcal I$.
Suppose $r\ge s_{\mathrm{mix}}$ and that the data are rich, i.e., $\rank(X_i)=n_i$ for every $i\in\mathcal I$.
Then:
\begin{itemize}
\item[(i)] there exists at least one $g$ such that
$Y_{\mathrm{past}} = H_p\,g^\star$;
\item[(ii)] for any $g$ satisfying $Y_{\mathrm{past}} = H_p g$, the prediction
$\widehat Y_{\mathrm{fut}} := H_f\,g$
is exact: $\widehat Y_{\mathrm{fut}} = Y_{\mathrm{fut}}$.
\end{itemize}
\end{theorem}

\begin{proof}
By Lemma~\ref{prop:agg-factorization} we have
$
H_Y=\mathcal{O}_T(A_{\mathrm{mix}},C_{\mathrm{mix}}) X_{\mathrm{mix}}
$
and richness implies $\rank(X_{\mathrm{mix}})=n$.
Further, any trajectory $Y^\star$ generated by one system $i^\star\in \mathcal{I}$ can be seen as a trajectory of the aggregate system
\eqref{eq:agg-system-Q} (take an initial condition with only the $i^\star$-th block possibly nonzero).
Since $r\ge s_{\mathrm{mix}}$, the assumptions of Lemma~\ref{lem:key} applied to
$(A_{\mathrm{mix}},C_{\mathrm{mix}})$ and $H_Y$ are satisfied, which yields (i) and (ii).
\end{proof}

\textbf{Continuation map for the aggregate system.} Theorem \ref{thm:mixed-prediction} favours the use of the latent variable $g$. However, prediction results could also be given via the continuation map of the aggregated system.  
In fact, if  $r\ge s_{\mathrm{mix}}$, by Lemma~\ref{lem:key} applied to the aggregate pair
$(A_{\mathrm{mix}},C_{\mathrm{mix}})$, 
there exists a linear continuation map
$L_{\mathrm{mix}}\in\R^{(T-r)p\times rp}$ such that
\[
H_f = L_{\mathrm{mix}}\,H_p,
  \quad
Y_{\mathrm{fut}} = L_{\mathrm{mix}}\,Y_{\mathrm{past}}
\]
for every length-$T$ trajectory $Y^\star$ generated by any  system in the family.  From a data-driven viewpoint, i.e., if the aggregated system is known only via $H_Y$, one may compute a representative continuation map by solving the linear
matrix equation $H_f=L\,H_p$ under the additional condition on data richness of the aggregated library ($\rank(X_i)=n_i$ for every $i\in\mathcal I$). As in the single-system case, $L_{\mathrm{mix}}$ need not be unique as a matrix, but its action on
admissible past windows is unique.

\textbf{Exact prediction does not imply system reconstruction.}
\label{rem:no-mode-reconstruction}
Theorem~\ref{thm:mixed-prediction} guarantees exact prediction from the aggregated library, but it does not, 
in general, allow one to identify the system that generated the trajectory to be completed.
In fact, since the aggregated matrix $H_Y$ may have many more columns than its row-rank ($N\gg n$), representations of $Y^\star$ in terms of the columns of $H_Y$ are generally \emph{not unique}. In particular, it may happen that there exist coefficients $\tilde g^{(1)},\dots,\tilde g^{(Q)}$, with more than one $\tilde g^{(i)}$ nonzero, such that
\[
  Y^\star
  = \sum_{i=1}^Q H_Y^{(i)}\,\tilde g^{(i)}
\]
Therefore, from the mixed representation alone, one cannot in general certify whether $Y^\star$ was generated by a single system   $i^\star$ or by an (algebraic) combination of columns drawn from multiple systems. Additional structure (e.g., distinguishability or sparsity/group-sparsity constraints on $g$) is needed if system  reconstruction is a goal.

Theorem \ref{thm:mixed-prediction} provides the first foundational result of the paper.
Building on this result, the next  four sections develop four main extensions:
\begin{enumerate}
\item \textbf{Noisy data (Section \ref{sec:mix-noise}).}
We study the robustness and quality of prediction when noise affects the data, both at the level of the observed measurements and of the training library.
\item \textbf{Out-of-library prediction (Section \ref{sec:out-of-library}).}
We address the problem of predicting trajectories generated by systems that are not represented in the nominal training library.
\item \textbf{Interconnected libraries (Section \ref{subsec:feedback_composed_libraries}).}
We introduce a composition-based generalization mechanism:  we interconnect atomic systems to give rise to libraries that generate modes not present in the atomic systems. 

\item \textbf{Immersion and nonlinear trajectories (Section \ref{sec:output-containment-immersion}).}
We show that the proposed framework naturally extends to certain nonlinear systems. In particular, we consider systems that can be \emph{immersed} into a finite-dimensional linear system. We show that, under suitable conditions, the same prediction procedure based on linear templates remains exact. 
\end{enumerate}

Prior to that, in the next subsection we discuss a point that deserves consideration. We discuss it here to maintain the analysis to a simple level but the considerations we make here apply to the rest of the paper as well.

\subsection{Reduced-order latent variables via a basis of $\im(H_p)$} \label{subsec:reduced}

In this section we give a variant of Theorem \ref{thm:mixed-prediction} with reduced-order latent variable that replaces $g$ therein.  Let
\[
d := \rank(H_p) = \dim\bigl(\im(H_p)\bigr).
\]
Under the noise-free assumption and data richness,  $d$ coincides with the rank of
the truncated observability matrix of the mixed system, i.e.,
$d=\rank\bigl(\mathcal{O}_{\mathrm{mix}}(r)\bigr)\le n$, and equals $n$ when the aggregate pair is fully observable.

Consider the compact SVD of $H_p$
\[
H_p = 
\begin{bmatrix} U & \hat U\end{bmatrix}
\begin{bmatrix} 
\Sigma & 0_{d\times (N-d)}\\
0_{(rp-d)\times d} & 0_{(rp-d)\times (N-d)}
\end{bmatrix}
\begin{bmatrix} 
V^\top\\
\hat V^\top
\end{bmatrix}
= U \Sigma V^\top,
\]
where $U\in\mathbb{R}^{rp\times d}$ and $V\in\mathbb{R}^{N\times d}$ have orthonormal columns and
$\Sigma\in\mathbb{R}^{d\times d}$ is diagonal with positive entries. Set 
\begin{equation}\label{Bp-Bf}
B_p := U\Sigma \in\mathbb{R}^{rp\times d}, \quad 
B_f := H_f V \in\mathbb{R}^{(T-r)p\times d}.
\end{equation}

\begin{theorem}[Prediction with reduced-order latent variable]\label{thm:mixed-prediction-reduced}
Consider the same assumptions as in Theorem \ref{thm:mixed-prediction}. 
Then:
\begin{itemize}
\item[(i)] there exists a unique $\alpha^\star$ such that
$Y_{\mathrm{past}} = B_p\,\alpha^\star$;
\item[(ii)] the prediction
$\widehat Y_{\mathrm{fut}} := B_f\,\alpha^\star$
is exact: $\widehat Y_{\mathrm{fut}} = Y_{\mathrm{fut}}$.
\end{itemize}
\end{theorem}

\begin{proof}
By data richness, there exists $g^\star$ such that $Y^\star = H_Y g^\star$. Thus, $Y_{\mathrm{past}}=H_p g^\star$ and $Y_{\mathrm{past}}\in \im(H_p)$. As $H_p=U \Sigma V^\top$ and $V^\top$ has full row rank, the identity $\im(H_p)=\im(B_p)$ holds. This implies that $Y_{\mathrm{past}}\in \im(B_p)$, hence there exists $\alpha^\star$ such that $Y_{\mathrm{past}} = B_p \alpha^\star$. Since $B_p$ has full column rank, $\alpha^\star$ is unique. This shows (i). 

Note that $Y_{\mathrm{past}}= H_p g^\star = U\Sigma V^\top g^\star = B_p V^\top g^\star$. As the solution to  $Y_{\mathrm{past}}=B_p \alpha^\star$ is unique, we obtain $V^\top g^\star= \alpha^\star$. The definition $\widehat Y_{\mathrm{fut}} := B_f\,\alpha^\star$  yields 
$\widehat Y_{\mathrm{fut}} = B_f \, V^\top g^\star$. Finally, recall that $r\ge s_{\mathrm{mix}}$ implies $H_f  = 
L_{\mathrm{mix}} H_p$. Therefore,  $B_f:=H_f V= 
L_{\mathrm{mix}} U\Sigma$ and 
$H_f  = 
L_{\mathrm{mix}} U\Sigma V^\top = B_f V^\top$. Hence, $\widehat Y_{\mathrm{fut}} = B_f \, V^\top g^\star = H_f g^\star$, which shows $\widehat Y_{\mathrm{fut}} = Y_{\mathrm{fut}}$ by Theorem \ref{thm:mixed-prediction}(ii). 
\end{proof}

Unlike the dimension of $g$, the dimension $d$ of $\alpha^\star$ depends only on the rank of the matrix of template behaviors $H_p$  (and is at most $n$), not on the number $N$ of stored trajectories.


\section{Trajectory prediction from noisy data}
\label{sec:mix-noise}

We now study robustness of the aggregated predictor when noise affects the data.

\subsection{Noisy measurements with exact library}
\label{subsec:mix-noisy-meas}

We first assume that the blocks $(H_p,H_f)$ of the aggregated library \eqref{HY-mixed} are noise-free, and that  only the observed window
$Y_{\mathrm{past}}$ of $Y^\star$
is corrupted by additive noise:
\begin{equation}\label{eq:mix-noisy-meas}
  Z_{\mathrm{past}} = Y_{\mathrm{past}} + n_{\mathrm{past}},
  \quad
  n_{\mathrm{past}}\in\R^{rp}.
\end{equation}
Since $Z_{\mathrm{past}}$ need not lie in $\im(H_p)$, we estimate coefficients via least squares\footnote{Appendix \ref{thm.least.squares} provides a summary of least squares solutions. 
}

\begin{equation}\label{eq:mix-LS}
  \hat g \in \arg\min_{g\in\R^{N}} \ \|Z_{\mathrm{past}}-H_pg\|_2^2,
\end{equation}
and decode the prediction as
\begin{equation}\label{eq:mix-pred}
  \widehat Y_{\mathrm{fut}} := H_f \hat g.
\end{equation}

As shown below, while the minimizer of \eqref{eq:mix-LS} may be non-unique (typically $N\gg rp$), the prediction is
unique as soon as the library satisfies the nominal continuation identity \eqref{eq:mix-continuation-recall} below.
Recall from the noise-free  analysis  for the aggregated system that, under  the condition that $r\ge s_{\mathrm{mix}}$, 
there exists a
linear continuation map $L_{\mathrm{mix}}\in\R^{(T-r)p\times rp}$ such that
\begin{equation}\label{eq:mix-continuation-recall}
  H_f = L_{\mathrm{mix}}\,H_p,
  \quad
  Y_{\mathrm{fut}} = L_{\mathrm{mix}}\,Y_{\mathrm{past}}
\end{equation}
This structural relation is the key ingredient
behind the following result, 
which shows uniqueness of the prediction along with an exact quantification of the prediction error. 

\begin{theorem}[Prediction with noisy observations]\label{thm:mixed-prediction-noisy}
Consider any trajectory $Y^\star$ generated by one (unknown) system $i^* \in \mathcal{I}$,
and let $Z_{\mathrm{past}}$ be as in \eqref{eq:mix-noisy-meas}. 
Suppose $r\ge s_{\mathrm{mix}}$.  
Then:
\begin{itemize}
\item[(i)] {\rm [Prediction is unique even if the least-squares minimizer is not]} Let $\hat g_1,\hat g_2$ be any two minimizers of~\eqref{eq:mix-LS}. Then
$
H_f\hat g_1=H_f\hat g_2,
$
therefore $\widehat Y_{\mathrm{fut}}=H_f\hat g$ is well-defined (independent of the chosen minimizer).

\item[(ii)] {\rm [Prediction error with noisy observations] } Consider the least-squares predictor~\eqref{eq:mix-LS}--\eqref{eq:mix-pred}. 
If $\rank(X_i)=n_i$ for every $i\in\mathcal I$, then the prediction error
$\epsilon:=\widehat Y_{\mathrm{fut}}-Y_{\mathrm{fut}}$ satisfies the identity
\begin{equation}\label{eq:mix-eps-noisy-meas}
  \epsilon = L_{\mathrm{mix}}\,\Pi_p\,n_{\mathrm{past}},
\end{equation}
where 
$
\Pi_p:=H_pH_p^\dagger
$ 
is the orthogonal projector onto $\im(H_p)$, 
and the bound
\begin{equation}\label{eq:mix-eps-noisy-meas-bound}
  \|\epsilon\|_2 \le \|L_{\mathrm{mix}}\|_2\,\|n_{\mathrm{past}}\|_2 .
\end{equation}
\end{itemize}
\end{theorem}

\begin{proof}
(i)
As $\hat g = H_p^\dag Z_{\mathrm{past}}+(I_N-H_p^\dag H_p)\gamma$, all least squares minimizers produce the same fitted value $H_p\hat g= H_p H_p^\dag Z_{\mathrm{past}}$, equivalently, $H_p(\hat g_1-\hat g_2)=0$.
Then, using $H_f=L_{\mathrm{mix}}H_p$, we obtain
$
  H_f(\hat g_1-\hat g_2) = L_{\mathrm{mix}}H_p(\hat g_1-\hat g_2)=0.
$
This concludes the proof  of (i).\\
(ii) As remarked in the proof of (i) above, every minimizer of \eqref{eq:mix-LS} satisfies $H_p\hat g=\Pi_p Z_{\mathrm{past}}$, and therefore
\begin{equation}\label{eq:mix-pred-proj}
  \widehat Y_{\mathrm{fut}}
\stackrel{\eqref{eq:mix-pred}}{=} H_f\hat g
\stackrel{\eqref{eq:mix-continuation-recall}}
  {=} L_{\mathrm{mix}}\,H_p\hat g
= L_{\mathrm{mix}}\,\Pi_p Z_{\mathrm{past}}.
\end{equation}
From \eqref{eq:mix-pred-proj} and $Z_{\mathrm{past}}=Y_{\mathrm{past}}+n_{\mathrm{past}}$,
\[
  \widehat Y_{\mathrm{fut}}
  = L_{\mathrm{mix}}\Pi_p(Y_{\mathrm{past}}+n_{\mathrm{past}})
  = L_{\mathrm{mix}}Y_{\mathrm{past}} + L_{\mathrm{mix}}\Pi_p n_{\mathrm{past}},
\]
where we have used $Y_{\mathrm{past}}\in\im(H_p)$, thus $\Pi_pY_{\mathrm{past}}=Y_{\mathrm{past}}$.
Finally, $Y_{\mathrm{fut}}=L_{\mathrm{mix}}Y_{\mathrm{past}}$ by \eqref{eq:mix-continuation-recall}, hence
$\epsilon=L_{\mathrm{mix}}\Pi_p n_{\mathrm{past}}$. The bound \eqref{eq:mix-eps-noisy-meas-bound} follows from $\|\Pi_p\|_2=1$.
\end{proof}

\begin{remark}[A fully data-dependent bound]\label{rem:dd-bound} It is immediate to give a fully data-dependent bound on $\epsilon$. In fact,
\[
  \epsilon 
\stackrel{\eqref{eq:mix-eps-noisy-meas}}{=}
  L_{\mathrm{mix}}\,\Pi_p\,n_{\mathrm{past}}=
  L_{\mathrm{mix}}\,H_pH_p^\dagger\,n_{\mathrm{past}}
\stackrel{\eqref{eq:mix-continuation-recall}}{=}
  H_f H_p^\dagger n_{\mathrm{past}},
\]
so that
$
\|\epsilon\|_2 \le \|H_f H_p^\dagger\|_2\,\|n_{\mathrm{past}}\|_2,
$
which is computable from data as long as we know a bound on the noise. This bound is consistent with what we would obtain using any solution $L$ of $H_f = L H_p$, which is computable from data as well. In fact, similar to Lemma~\ref{lem:key}(iii), one can show that any such map coincides with $L_{\mathrm{mix}}$ on $\im(H_p)$, i.e., $L H_p =L_{\mathrm{mix}} H_p$. Therefore 
\begin{eqnarray}
 \widehat Y_{\mathrm{fut}} \stackrel{\eqref{eq:mix-pred}}{:=} H_f \hat g=L H_p   \hat g=L  \Pi_p Z_{\mathrm{past}}=L  \Pi_p  Y_{\mathrm{past}} +
 L  \Pi_p  n_{\mathrm{past}} \nonumber \\
 =L_{\mathrm{mix}}  \Pi_p  Y_{\mathrm{past}} +
 L  \Pi_p  n_{\mathrm{past}}
  = Y_{\mathrm{fut}} +
 L  \Pi_p  n_{\mathrm{past}} \nonumber
\end{eqnarray}
from which $\epsilon =  L  \Pi_p  n_{\mathrm{past}}= 
H_f H_p^\dagger n_{\mathrm{past}}$, where the second equality follows from $L$ being any solution of $H_f = L H_p$. 
\end{remark}

The uniqueness and error identities above hinge on the exact structural relation $H_f=L_{\mathrm{mix}}H_p$,
and this occurs whenever the library contains noise-free trajectories. This scenario is realistic in the case of synthetic data, i.e., data generated artificially, and it is particularly relevant in the context of out-of-library prediction (Section \ref{sec:out-of-library}). In this case, in order to improve prediction performance, one may consider populating the library with template behaviors, which can also be generated artificially.

On the other hand, when the library blocks themselves are noisy, this identity is generally lost and
different least-squares minimizers may lead to different predicted values. In this case, additional choices 
are necessary to obtain a well-posed predictor. We discuss this point in the following section.

\subsection{Noisy library and noisy observations}
\label{subsec:mix-noisyH}

Let $(H_p,H_f)$ be the (unknown) noise-free blocks
in \eqref{HY-mixed}
satisfying $H_f=L_{\mathrm{mix}}H_p$ as in~\eqref{eq:mix-continuation-recall}, and assume we observe
\begin{equation}\label{eq:mix-noisyH-model}
  \bar H_p := H_p + \Delta_p,
  \quad
  \bar H_f := H_f + \Delta_f,
\end{equation}
for unknown perturbations $\Delta_p,\Delta_f$. We keep the measurement model~\eqref{eq:mix-noisy-meas}.
Given $(\bar H_p,\bar H_f)$, we solve
\begin{equation}\label{eq:mix-LS-noisyH}
  \hat g \in \arg\min_{g\in\R^{N}} \ \|Z_{\mathrm{past}}-\bar H_pg\|_2^2,
\end{equation}
and predict
\begin{equation}\label{eq:mix-pred-noisyH}
  \widehat Y_{\mathrm{fut}} := \bar H_f \hat g.
\end{equation}
Unlike the exact-library case where the identity $H_f=L_{\mathrm{mix}} H_p$ was extensively used to show the existence of a unique prediction and to bound the residual as in Theorem \ref{thm:mixed-prediction-noisy}, we now generally have $\bar H_f\neq L_{\mathrm{mix}}\bar H_p$ (in fact, a matrix $\bar L$ such that 
$\bar H_f= \bar L_{\mathrm{mix}}\bar H_p$ does not generally  exist) and the decoded prediction $\bar H_f\hat g$ will generally depend on the chosen minimizer if~\eqref{eq:mix-LS-noisyH} is not uniquely solved;
nevertheless, the following error identity holds for \emph{any} minimizer.

\begin{theorem}[Prediction error with noisy library]
\label{thm:mix-noisyH-decomp}
Assume $r\ge s_{\mathrm{mix}}$, so that for the \emph{noise-free} blocks $H_f=L_{\mathrm{mix}}H_p$
and $Y_{\mathrm{fut}}=L_{\mathrm{mix}}Y_{\mathrm{past}}$. Let $\hat g$ be any minimizer of~\eqref{eq:mix-LS-noisyH}
and define $\widehat Y_{\mathrm{fut}}$ by~\eqref{eq:mix-pred-noisyH}. Define the measured LS residual
\begin{equation}\label{eq:mix-residual-noisyH}
  \bar\rho := Z_{\mathrm{past}}-\bar H_p\hat g.
\end{equation}
Then the prediction error $\epsilon:=\widehat Y_{\mathrm{fut}}-Y_{\mathrm{fut}}$ admits the decomposition
\begin{equation}\label{eq:mix-eps-decomp}
  \epsilon
  =
  \underbrace{-\,L_{\mathrm{mix}}\,\bar\rho}_{\text{\rm LS residual}}
  \;+\;
  \underbrace{L_{\mathrm{mix}}\,n_{\mathrm{past}}}_{\text{\rm measurement noise}}
  \;+\;
  \underbrace{\bigl(\Delta_f - L_{\mathrm{mix}}\Delta_p\bigr)\hat g}_{\text{\rm library perturbation}}.
\end{equation}
Consequently,
\begin{equation}\label{eq:mix-eps-bound}
  \|\epsilon\|_2
  \le
  \|L_{\mathrm{mix}}\|_2 (\|\bar\rho\|_2
  + \|n_{\mathrm{past}}\|_2)
  + \|\Delta_f - L_{\mathrm{mix}}\Delta_p\|_2\,\|\hat g\|_2.
\end{equation}
\end{theorem}

\begin{proof}
Start from $\epsilon=\bar H_f\hat g-Y_{\mathrm{fut}}=(H_f+\Delta_f)\hat g-L_{\mathrm{mix}}Y_{\mathrm{past}}$ and use
$Z_{\mathrm{past}}=Y_{\mathrm{past}}+n_{\mathrm{past}}$ to write
\[
\epsilon = H_f\hat g - L_{\mathrm{mix}}Z_{\mathrm{past}} + 
L_{\mathrm{mix}}n_{\mathrm{past}} + \Delta_f\hat g.
\]
Since $H_f=L_{\mathrm{mix}}H_p$,
$H_f\hat g - L_{\mathrm{mix}}Z_{\mathrm{past}}
= L_{\mathrm{mix}}(H_p\hat g - Z_{\mathrm{past}}).$
Moreover,
\[
  Z_{\mathrm{past}} - H_p\hat g
  = (Z_{\mathrm{past}}-\bar H_p\hat g) + (\bar H_p-H_p)\hat g
  = \bar\rho + \Delta_p\hat g.
\]
Substituting yields~\eqref{eq:mix-eps-decomp}, and taking norms gives~\eqref{eq:mix-eps-bound}.
\end{proof}

\begin{remark}
Even if $n_{\mathrm{past}}=0$, the error can be nonzero since typically $Y_{\mathrm{past}}\notin\im(\bar H_p)$,
yielding 
$\bar\rho\stackrel{\eqref{eq:mix-residual-noisyH}}{=} (I- \bar H_p \bar H_p^\dagger)Z_{\mathrm{past}}
=(I- \bar H_p \bar H_p^\dagger)Y_{\mathrm{past}}
\neq 0$, and since the term $(\Delta_f-L_{\mathrm{mix}}\Delta_p)\hat g$ remains.
This is the main difference with the exact-library case of Theorem~
\ref{thm:mixed-prediction-noisy}. 
\end{remark}

\begin{example}
\begin{figure}[t]
    \centering
    \includegraphics[width=0.8\linewidth]{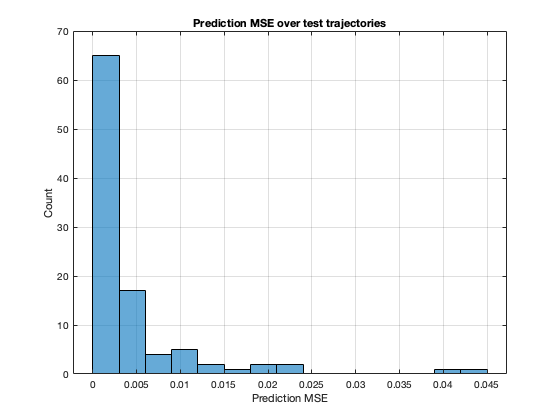}
    \caption{Histogram of the mean squared prediction error over $100$ test trajectories in the noisy-data setting. The horizontal axis reports the MSE on the predicted future window, while the vertical axis reports the number of test trajectories associated with each error interval.}
    \label{fig:noisy_hist}
\end{figure}

We illustrate the noisy-data setting of Section~V with a simple numerical experiment. We generate a library from $Q=10$ random stable discrete-time systems of order $3$ with a scalar output, each contributing $50$ template trajectories of length $T=20$. Additive noise with a prescribed $25$dB signal-to-noise ratio is injected both into the library trajectories and into the observed past window. Namely, denoted by $P_y$ the power of any of trajectories of interest (both the past trajectories to be continued and those in the libraries)  and by $P_n$ the power of the added noise, their ratio satisfies $[P_y/P_n]_{\text{dB}}=25$.
%
%
Prediction is then performed with horizon split $r=15$ by solving \eqref{eq:mix-LS-noisyH}-\eqref{eq:mix-pred-noisyH}. Performance is evaluated over $100$ newly generated trajectories, each drawn from one of the nominal systems, through the mean squared prediction error 
$MSE:=\sum_{k=0}^{T-r-1}\|\widehat Y_{\mathrm{fut},k}-Y_{\mathrm{fut},k}\|^2$
on the future window. Figure~\ref{fig:noisy_hist} reports the histogram of these prediction errors; the horizontal axis shows the MSE value, and the vertical axis reports how many test trajectories achieve MSE values in each interval.
\end{example}

\subsection{Regularization and link between $g$- and $L$-based predictors in the noisy setting} \label{subsec:noiseLg}

In the noisy multi-mode setting we observe a perturbed past window
$Z_{\mathrm{past}} = Y_{\mathrm{past}} + n_{\mathrm{past}}$
and only have access to perturbed library blocks
$(\bar H_p,\bar H_f) = (H_p+\Delta_p,\; H_f+\Delta_f)$,
As shown by Theorem~\ref{thm:mix-noisyH-decomp}, two phenomena arise that are absent in the ideal case: the exact continuation identity $H_f=L_{\mathrm{mix}}H_p$ is generally lost, and least-squares encoding may become ill-posed, in the sense that different minimizers can lead to different decoded futures $\bar H_f \hat g$. This motivates regularization and clarifies the relation between the $g$- and $L$-based viewpoints.

A natural way to regularize the latent-variable formulation is ridge regression:
\[
\hat g_\mu
\in
\arg\min_{g\in\mathbb{R}^N}
\|Z_{\mathrm{past}}-\bar H_p g\|_2^2+\mu\|g\|_2^2,
\quad \mu>0.
\tag{38}
\]
This gives a unique solution and the prediction
$
\hat Y_{\mathrm{fut}}:=\bar H_f \hat g_\mu .
$
More generally, we can replace the ridge penalty by structured regularizers (e.g., group sparsity) to promote combinations of templates. In this sense, regularization is also a modeling choice.
For ridge regularization, the encoder--decoder scheme admits an explicit induced input--output form. Indeed,
$
\hat g_\mu=(\bar H_p^\top \bar H_p+\mu I)^{-1}\bar H_p^\top Z_{\mathrm{past}},
$
hence
\[
\hat Y_{\mathrm{fut}}
=
\bar H_f(\bar H_p^\top \bar H_p+\mu I)^{-1}\bar H_p^\top Z_{\mathrm{past}}
=: \bar L_\mu Z_{\mathrm{past}}.
\tag{39}
\]
Therefore, ridge-regularized decoding is exactly equivalent to a direct linear predictor with map $\bar L_\mu$.

An alternative is to learn an input--output map directly from the noisy library, for instance via
\[
\hat L_\lambda
\in
\arg\min_L
\|\bar H_f-L\bar H_p\|_F^2+\lambda\|L\|_F^2,
\quad \lambda>0,
\tag{40}
\]
which yields the predictor $\hat Y_{\mathrm{fut}}=\hat L_\lambda Z_{\mathrm{past}}$, with
\[
\hat L_\lambda
=
\bar H_f \bar H_p^\top(\bar H_p\bar H_p^\top+\lambda I_{rp})^{-1}.
\]
In general, $\hat L_\lambda$ does not coincide with the induced map $\bar L_\mu$ in~(39): the two regularizations act on different objects (coefficients versus operator) and may react differently to errors in the library.

In the noiseless regime, the existence of a continuation map ensures equivalence between the $g$- and $L$-based predictors on admissible past windows. Under noisy measurements and noisy libraries this invariance is generally lost. Regularization restores well-posedness and often yields an explicit linear predictor, while the $g$-viewpoint remains preferable when one wishes to encode structural constraints on how templates are selected and combined.

\subsection{Other topologies for the aggregate system}
All the results derived so far apply to the aggregate library obtained by
column-wise concatenation of the atomic libraries, together with its auxiliary
aggregate realization \eqref{eq:agg-system-Q}. 
Beyond this construction, one may consider libraries associated with genuine interconnections of the atomic systems in $\mathcal{I}$, such as feedback and series interconnections, as well as parallel interconnections with inputs\footnote{In the autonomous case discussed in Section~\ref{subsec:mix-agg}, the aggregate system can be interpreted as a parallel interconnection of the template systems, with output given by the sum of the individual outputs. There, we used the term \emph{aggregate} system to emphasize the algebraic concatenation of template libraries rather than an explicit interconnection mechanism.}.
Such libraries can considerably enlarge the class of
systems for which prediction from templates is possible, especially in
out-of-library settings.

A key point is that these interconnected libraries need not be obtained by
collecting data from the interconnected systems themselves. In many cases,
this may be impractical, costly, or even impossible. Instead, as we show in
Section \ref{subsec:feedback_composed_libraries}, they can be constructed
directly from the atomic libraries, in a way that preserves the same
factorization property as in Lemma \ref{lemma:fb_factorization}. This makes
it possible to endow predictors built from atomic data with the ability to
reproduce trajectories that are substantially different from those generated
by the individual atomic systems taken in isolation.


\section{Out-of-library prediction}
\label{sec:out-of-library}

A key advantage of the mixed library is that it implicitly spans a \emph{larger} behavior space:
by taking linear combinations of columns across blocks, it can represent trajectories that do not
belong to any single nominal mode. This is precisely the mechanism that enables a form of
\emph{generalization} beyond the training modes. We will start with a result that holds generically for any dynamical system, even a nonlinear one, that does not belong to the library of template modes. Then, we will specialize the analysis to linear systems and highlight properties of the library under which we have exact prediction for out-of-library systems.
For simplicity, we will focus on the noise-free case. 

\subsection{Generic mismatch}
\label{subsec:out-of-library-generic}  

So far, we have assumed that the trajectory to be predicted is generated by a nominal model.
We now consider the \emph{out-of-library} case, where the system generating the data is not necessarily a template mode. 
Throughout this subsection we keep the same notation
$H_p,H_f$ for the aggregated library and we assume $r\ge s_{\mathrm{mix}}$, so that
\begin{equation}\label{eq:mix-continuation-outlib}
  H_f = L_{\mathrm{mix}}\,H_p
\end{equation}
for a continuation map $L_{\mathrm{mix}}\in\R^{(T-r)p\times rp}$ depending only on the
aggregated dynamics (cf. Section \ref{sec:multimode}).

Let $Y^\star$ be an arbitrary
length-$T$ sequence, not assumed to be generated by one of the template  systems.
We consider the usual least-squares encoding
\begin{equation}\label{eq:outlib-LS}
  \hat g \in \arg\min_{g\in\R^{N}} \ \|Y_{\mathrm{past}}-H_pg\|_2^2,
  \quad
  \widehat Y_{\mathrm{fut}}:=H_f\hat g .
\end{equation}

\begin{proposition}[Out-of-library error decomposition]\label{prop:outlib-decomp}
Consider the aggregate library \eqref{eq:mode_library}, \eqref{eq:global_library}, decomposed as in \eqref{eq:new_traj_split_intro}, where $r\ge s_{\mathrm{mix}}$ and $s_{\mathrm{mix}}$ is the observability index of the aggregate system. Consider also   the continuation map $L_{\mathrm{mix}}$ of the aggregate system, which  satisfies \eqref{eq:mix-continuation-outlib}. For any sequence $Y^\star\in \R^{pT}$ and any least squares minimizer in
\eqref{eq:outlib-LS}, the prediction error $\epsilon:=\widehat Y_{\mathrm{fut}}-Y_{\mathrm{fut}}$ satisfies
\begin{equation}\label{eq:outlib-eps-decomp}
  \epsilon
  = -\,L_{\mathrm{mix}}(I-\Pi_p)Y_{\mathrm{past}}
    \;-\;\bigl(Y_{\mathrm{fut}}-L_{\mathrm{mix}}Y_{\mathrm{past}}\bigr),
\end{equation}
and therefore
\begin{equation}\label{eq:outlib-eps-bound}
  \|\epsilon\|_2
  \le \|L_{\mathrm{mix}}\|_2\,\rho_{\mathrm{feas}} + \delta_{\mathrm{cont}},
\end{equation}
where we define
\begin{equation}\label{eq:rho-feas}
  \rho_{\mathrm{feas}}
  := \min_{g}\|Y_{\mathrm{past}}-H_pg\|_2
\quad \text{\emph{(feasibility residual)}}
\end{equation}
and
\begin{equation}\label{eq:delta-cont-def}
  \delta_{\mathrm{cont}}
  := \|Y_{\mathrm{fut}}-L_{\mathrm{mix}}Y_{\mathrm{past}}\|_2 \quad \text{\emph{(continuation defect)}}.
\end{equation}
\end{proposition}

\begin{proof}
Every minimizer \eqref{eq:outlib-LS} satisfies $H_p\hat g=\Pi_pY_{\mathrm{past}}$, hence
\begin{equation}\label{eq:outlib-pred-proj}
  \widehat Y_{\mathrm{fut}}
  = H_f\hat g
  = L_{\mathrm{mix}}\,H_p\hat g
  = L_{\mathrm{mix}}\,\Pi_p\,Y_{\mathrm{past}} .
\end{equation}
Using \eqref{eq:outlib-pred-proj}, write
\[
\epsilon
= L_{\mathrm{mix}}\Pi_pY_{\mathrm{past}}-Y_{\mathrm{fut}}
= L_{\mathrm{mix}}(Y_{\mathrm{past}}-(I-\Pi_p)Y_{\mathrm{past}})-Y_{\mathrm{fut}},
\]
hence
\[
\epsilon
= -\,L_{\mathrm{mix}}(I-\Pi_p)Y_{\mathrm{past}}-\bigl(Y_{\mathrm{fut}}-L_{\mathrm{mix}}Y_{\mathrm{past}}\bigr),
\]
which is \eqref{eq:outlib-eps-decomp}. Taking norms and using \eqref{eq:rho-feas}--\eqref{eq:delta-cont-def}
gives \eqref{eq:outlib-eps-bound}.
\end{proof}

If $Y_{\mathrm{past}}\in \im(H_p)$, then $\widehat Y_{\mathrm{fut}}= L_{\mathrm{mix}}\,\Pi_p\,Y_{\mathrm{past}}$ would become $L_{\mathrm{mix}}Y_{\mathrm{past}}$. Thus, the term $-\,L_{\mathrm{mix}}(I-\Pi_p)Y_{\mathrm{past}}$ in $\epsilon$ measures the deviation of the actual prediction from the one obtainable if $Y_{\mathrm{past}}$ were in $\im(H_p)$. 
On the other hand, $\delta_{\mathrm{cont}}=0$ precisely when the pair $(Y_{\mathrm{past}},Y_{\mathrm{fut}})$ is
consistent with the continuation map $L_{\mathrm{mix}}$, regardless of how $Y^\star$ is generated. Hence, the two quantities $\rho_{\mathrm{feas}}$ and $\delta_{\mathrm{cont}}$ isolate, respectively,
(i) how far the past window is from the library span, and (ii) how far the observed continuation deviates
from the library continuation.

Proposition~\ref{prop:outlib-decomp} is \emph{purely algebraic}: it uses only the existence of a library
continuation map \eqref{eq:mix-continuation-outlib} and the geometry of orthogonal projection onto
$\im(H_p)$; in fact, $Y^\star$ may even come from a nonlinear or time-varying mechanism.
The specialization to the linear case is discussed hereafter.
 

\subsection{Specialization to the linear case}

We now assume that $Y^\star$ is produced by a \emph{new} LTI system $(A_{\mathrm{new}},C_{\mathrm{new}})$, and discuss how the bound on the prediction error can be tailored.  
Let $s_{\mathrm{new}}$ be the observability index of $(A_{\mathrm{new}},C_{\mathrm{new}})$ and assume
\begin{equation}\label{eq:ass:r-ge-snew}
r\ge s_{\mathrm{new}}.
\end{equation}
Then there exists a continuation map $L_{\mathrm{new}}\in\R^{(T-r)p\times rp}$ such that, for every
length-$T$ trajectory of the new system,
\begin{equation}\label{eq:new-continuation}
Y_{\mathrm{fut}} = L_{\mathrm{new}}\,Y_{\mathrm{past}}.
\end{equation}

Proposition \ref{prop:outlib-decomp} specializes as follows:
\begin{corollary}[Out-of-library error decomposition for linear systems]\label{cor:out-of-lib-error}
Let $Y^\star\in \R^{pT}$ be a trajectory generated by the LTI system $(A_{\mathrm{new}},C_{\mathrm{new}})$ and decomposed as in \eqref{eq:new_traj_split_intro}. Assume its observability index satisfies \eqref{eq:ass:r-ge-snew} and let $L_{\mathrm{new}}$ be the continuation map satisfying \eqref{eq:new-continuation}. 
Under the conditions of Proposition \ref{prop:outlib-decomp}, the predicted future trajectory $  \widehat Y_{\mathrm{fut}}$ in \eqref{eq:outlib-LS} yields a prediction error $\epsilon:=\widehat Y_{\mathrm{fut}}-Y_{\mathrm{fut}}$ that 
satisfies\footnote{If
$
\dim\im(\mathcal O_{\mathrm{new}}(r))
\le
\dim\im(H_p),
$
then $\|(I-\Pi_p)\Pi_{\mathrm{new}}\|_2$ is the sine of the largest
principal angle between $\im(H_p)$ and
$\im(\mathcal O_{\mathrm{new}}(r))$ \cite{moor}.
}
\begin{equation}\label{eq:outlib-LTI-bound}
\!\|\epsilon\|_2 \!
\le \!
\Bigl(\|L_{\mathrm{mix}}\|_2\,
\|(I-\Pi_p)\Pi_{\mathrm{new}}\|_2
+\|L_{\mathrm{new}}-L_{\mathrm{mix}}\|_2\Bigr)\,\|Y_{\mathrm{past}}\|_2.
\end{equation}
\end{corollary}

\begin{proof}
Recall the bound \eqref{eq:outlib-eps-bound} and definitions \eqref{eq:rho-feas}, \eqref{eq:delta-cont-def}. In view of \eqref{eq:delta-cont-def} and \eqref{eq:new-continuation}, it is immediately seen that $\delta_{\mathrm{cont}}\le \|L_{\mathrm{new}}-L_{\mathrm{mix}}\|_2\,\|Y_{\mathrm{past}}\|_2$. We now look for a bound on $\rho_{\mathrm{feas}}$. Solving for $g$, we have $\rho_{\mathrm{feas}}=\|(I-\Pi_p)Y_{\mathrm{past}}\|_2$. As $Y_{\mathrm{past}}\in 
\im\left(\mathcal{O}_{\mathrm{new}}(r)\right)$, letting $\Pi_{\mathrm{new}}$ be an orthogonal projector onto 
$\im\left(\mathcal{O}_{\mathrm{new}}(r)\right)$, we have $\rho_{\mathrm{feas}}=\|(I-\Pi_p)\Pi_{\mathrm{new}}Y_{\mathrm{past}}\|_2
\le 
\|(I-\Pi_p)\Pi_{\mathrm{new}}\|_2 
\|Y_{\mathrm{past}}\|_2
$. 
\end{proof}

\begin{remark}[Bounding $\rho_{\mathrm{feas}}$ via subspace gap]
Recall that $\rho_{\mathrm{feas}}=\|(I-\Pi_p)Y_{\mathrm{past}}\|_2$. 
Suppose that the new system that is generating $Y^\star$ is known via a \emph{template matrix} $H_{\mathrm{new}}$  analogous to the ones in \eqref{eq:mode_library} and that $Y^\star\in \im(H_{\mathrm{new}})$.
In particular $Y_{\mathrm{past}}\in \im(H_{p,\mathrm{new}})$. Then 
$Y_{\mathrm{past}}= H_{p,\mathrm{new}} H_{p,\mathrm{new}}^{\dagger} Y_{\mathrm{past}}=: \Pi_{p,\mathrm{new}}Y_{\mathrm{past}}$, which yields
\[
\rho_{\mathrm{feas}}= \|(\Pi_{p,\mathrm{new}}-\Pi_p)Y_{\mathrm{past}}\|_2 \le 
\|\Pi_{p,\mathrm{new}}-\Pi_p\|_2 \, \|Y_{\mathrm{past}}\|_2
\] 
If the two subspaces $\im(H_p)$ and $\im (H_{p,\mathrm{new}})$ have the {\em same dimension}, then 
the norm $\|\Pi_{p,\mathrm{new}}-\Pi_p\|_2$ defines the {\em gap metric} between the two subspaces and coincides with the sine of the largest principal angle $\theta_{\max}$ between them. Its value can be efficiently obtained by computing 
an orthonormal basis $W_1$ for one of them and an orthonormal basis $Z_2$ of the orthogonal complement of the other (e.g., via SVD) and then computing $\|\Pi_{p,\mathrm{new}}-\Pi_p\|_2=\|W_1^\top Z_2\|_2$ \cite{padoan2022behavioral,golub2013matrix,partington2004linear}. 
\end{remark}

\subsection{Exact prediction for out-of-library systems}

The foregoing analysis suggests a simple condition for \emph{exact} out-of-library
prediction: both error terms in \eqref{eq:outlib-eps-decomp}, i.e., in 
$\epsilon
  = -\,L_{\mathrm{mix}}(I-\Pi_p)Y_{\mathrm{past}}
    \;-\;\bigl(Y_{\mathrm{fut}}-L_{\mathrm{mix}}Y_{\mathrm{past}}\bigr)$,
 must vanish. The first term
vanishes if the past trajectory is representable by the aggregate library; the second vanishes if the
new system shares the same continuation map $L_{\mathrm{mix}}$ over the horizon $(r,T)$.

\begin{theorem}[Exact prediction beyond the training modes]
\label{thm:mix-ool-exact}
Assume $r\ge s_{\mathrm{mix}}$ 
so that $H_f=L_{\mathrm{mix}}H_p$.
Let $Y^\star$ be a trajectory of
a (possibly unknown) system. Suppose that:
\begin{enumerate}
\item[(i)] (\emph{Feasibility}) $Y_{\mathrm{past}}\in\im(H_p)$;
\item[(ii)] (\emph{Shared continuation}) $Y_{\mathrm{fut}}=L_{\mathrm{mix}}Y_{\mathrm{past}}$.
\end{enumerate}
Then, for any $g$ such that $Y_{\mathrm{past}}=H_pg$, the prediction
$\widehat Y_{\mathrm{fut}}:=H_fg$ is exact:
\[
  \widehat Y_{\mathrm{fut}} = Y_{\mathrm{fut}}.
\]
In particular, the predicted future is independent of the chosen feasible $g$.
\end{theorem}

\begin{proof}
If $Y_{\mathrm{past}}=H_pg$, then $H_fg=L_{\mathrm{mix}}H_pg=L_{\mathrm{mix}}Y_{\mathrm{past}}$.
By assumption (ii), $L_{\mathrm{mix}}Y_{\mathrm{past}}=Y_{\mathrm{fut}}$, hence
$\widehat Y_{\mathrm{fut}}=Y_{\mathrm{fut}}$.
\end{proof}

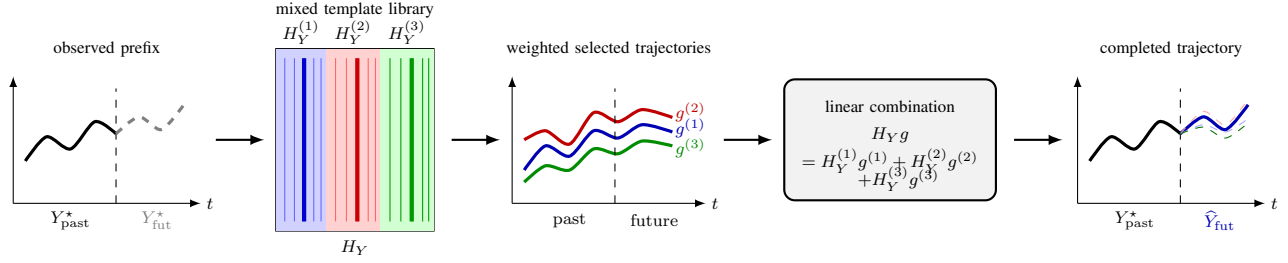
\begin{figure*}[t]
\centering
\footnotesize
\begin{tikzpicture}[>=latex, scale=0.8, transform shape]


\begin{scope}[shift={(0,0)}]
\node at (1.55,2.55) {observed prefix};

\draw[->] (0,0) -- (3.1,0) node[right] {$t$};
\draw[->] (0,0) -- (0,2.1);

\draw[very thick]
  plot[smooth] coordinates {(0.2,0.7) (0.55,1.1) (0.95,0.9) (1.35,1.35) (1.7,1.15)};
\draw[dashed] (1.7,0) -- (1.7,1.9);

\draw[gray, dashed, very thick]
  plot[smooth] coordinates {(1.7,1.15) (2.08,1.43) (2.45,1.22) (2.82,1.63)};

\node at (0.95,-0.28) {$Y_{\mathrm{past}}^\star$};
\node[gray] at (2.38,-0.28) {$Y_{\mathrm{fut}}^\star$};
\end{scope}

\draw[->,thick] (3.35,1.05) -- (4.15,1.05);

\begin{scope}[shift={(4.35,-0.45)}]
\node at (1.3,3.65) {mixed template library};

\draw[thick] (0,0) rectangle (2.6,3.0);
\draw[thick] (0.82,0) -- (0.82,3.0);
\draw[thick] (1.72,0) -- (1.72,3.0);

\fill[blue!18]  (0,0) rectangle (0.82,3.0);
\fill[red!18]   (0.82,0) rectangle (1.72,3.0);
\fill[green!18] (1.72,0) rectangle (2.6,3.0);

\foreach \x in {0.14,0.30,0.46,0.62,0.74}
  \draw[blue!60] (\x,0.15) -- (\x,2.85);
\foreach \x in {0.98,1.16,1.34,1.52,1.64}
  \draw[red!60] (\x,0.15) -- (\x,2.85);
\foreach \x in {1.88,2.06,2.24,2.42,2.52}
  \draw[green!60!black] (\x,0.15) -- (\x,2.85);

\node at (0.41,3.28) {$H_Y^{(1)}$};
\node at (1.27,3.28) {$H_Y^{(2)}$};
\node at (2.16,3.28) {$H_Y^{(3)}$};
\node at (1.3,-0.28) {$H_Y$};

\draw[blue!80!black, ultra thick] (0.46,0.15) -- (0.46,2.85);
\draw[red!80!black, ultra thick] (1.34,0.15) -- (1.34,2.85);
\draw[green!60!black, ultra thick] (2.24,0.15) -- (2.24,2.85);
\end{scope}

\draw[->,thick] (7.25,1.05) -- (8.05,1.05);

\begin{scope}[shift={(8.25,0)}]
\node at (1.6,2.55) {weighted selected trajectories};

\draw[->] (0,0) -- (3.2,0) node[right] {$t$};
\draw[->] (0,0) -- (0,2.1);

\draw[blue!70!black, very thick]
  plot[smooth] coordinates {(0.2,0.55) (0.55,0.95) (0.95,0.78) (1.35,1.20) (1.75,1.08) (2.15,1.30) (2.65,1.18)};
\draw[red!75!black, very thick]
  plot[smooth] coordinates {(0.2,1.05) (0.55,1.20) (0.95,0.98) (1.35,1.50) (1.75,1.35) (2.15,1.55) (2.65,1.42)};
\draw[green!55!black, very thick]
  plot[smooth] coordinates {(0.2,0.35) (0.55,0.62) (0.95,0.55) (1.35,0.90) (1.75,0.82) (2.15,1.03) (2.65,0.95)};

\draw[dashed] (1.7,0) -- (1.7,1.9);

\node[blue!70!black] at (2.95,1.25) {$g^{(1)}$};
\node[red!75!black] at (2.95,1.52) {$g^{(2)}$};
\node[green!55!black] at (2.95,0.92) {$g^{(3)}$};

\node at (0.95,-0.28) {$\mathrm{past}$};
\node at (2.35,-0.28) {$\mathrm{future}$};
\end{scope}

\draw[->,thick] (11.75,1.05) -- (12.55,1.05);

\begin{scope}[shift={(12.75,0.05)}]
\draw[rounded corners, thick, fill=gray!10] (0,0) rectangle (3.5,2.0);
\node at (1.75,1.55) {linear combination};
\node at (1.75,1.08) {$H_Y g$};
\node at (1.75,0.65) {$=H_Y^{(1)}g^{(1)}+H_Y^{(2)}g^{(2)}$};
\node at (1.75,0.35) {$\quad +H_Y^{(3)}g^{(3)}$};
\end{scope}

\draw[->,thick] (16.55,1.05) -- (17.35,1.05);

\begin{scope}[shift={(17.60,0)}]
\node at (1.55,2.55) {completed trajectory};

\draw[->] (0,0) -- (3.1,0) node[right] {$t$};
\draw[->] (0,0) -- (0,2.1);

\draw[very thick]
  plot[smooth] coordinates {(0.2,0.7) (0.55,1.1) (0.95,0.9) (1.35,1.35) (1.7,1.15)};
\draw[dashed] (1.7,0) -- (1.7,1.9);

\draw[blue!70!black, very thick]
  plot[smooth] coordinates {(1.7,1.15) (2.08,1.43) (2.45,1.22) (2.82,1.63)};

\draw[blue!35, dashed]
  plot[smooth] coordinates {(1.7,1.15) (2.08,1.30) (2.45,1.18) (2.82,1.40)};
\draw[red!35, dashed]
  plot[smooth] coordinates {(1.7,1.15) (2.08,1.52) (2.45,1.30) (2.82,1.72)};
\draw[green!45!black, dashed]
  plot[smooth] coordinates {(1.7,1.15) (2.08,1.26) (2.45,1.08) (2.82,1.33)};

\node at (0.95,-0.28) {$Y_{\mathrm{past}}^\star$};
\node[blue!70!black] at (2.35,-0.28) {$\widehat Y_{\mathrm{fut}}$};
\end{scope}

\end{tikzpicture}
\caption{Machine that performs out-of-library prediction from a mixed template library. The observed past $Y_{\mathrm{past}}^\star$ is fitted by combining selected trajectories from different template blocks. The same coefficients are then used to complete the trajectory, yielding the predicted future $\widehat Y_{\mathrm{fut}}$.}
\label{fig:out_of_library_templates}
\end{figure*}

\subsection{Exact prediction with prototype systems}

A concrete mechanism to realize Theorem~\ref{thm:mix-ool-exact} is to populate the aggregate
library  with \emph{prototype} systems whose solutions span a family of basic
output primitives. The mixed library then generates, by superposition,
new behaviors that may not coincide with any nominal system. For simplicity we will discuss the case of systems with a scalar output. 

For a square matrix $A$, $\spec(A)$ denotes the set of distinct
eigenvalues of $A$. We also introduce the notion of \emph{visible spectrum}
used below. For a linear pair $(A,C)$, define the visible spectrum as
\[
\spec_C(A)
:=
\{\lambda\in\spec(A): C\ker(\lambda I-A)\neq\{0\}\}.
\]
Hence $\spec_C(A)$ contains the eigenvalues whose eigenspaces are not
annihilated by the output map. Eigenvalues with multiplicity larger than
one are included at most once in $\spec_C(A)$; therefore $\spec_C(A)$ only
contains distinct eigenvalues.

\begin{proposition}[Exact prediction from a library of prototype systems]
\label{prop:mix-sentinels}
Assume $r\ge s_{\mathrm{mix}}$ and mixed data richness, i.e.,
$\rank(X_i)=n_i$ for every $i\in\mathcal I$. Suppose $p=1$.
Let the prototype system matrices $A_i$
$i=1,\ldots,Q$, and the new system matrix $A_{\mathrm{new}}$ be diagonalizable. 
Assume that
\begin{equation}\label{spectral.condition}
\spec_{C_{\mathrm{new}}}(A_{\mathrm{new}})
\subseteq
\bigcup_{i=1}^Q \spec_{C_i}(A_i).
\end{equation}
where $\bigcup_{i=1}^Q \spec_{C_i}(A_i)=
\spec_{C_{\mathrm{mix}}}(A_{\mathrm{mix}})$. 
Then every $T$-long trajectory $Y^\star$ of
$(A_{\mathrm{new}},C_{\mathrm{new}})$ satisfies the hypotheses of
Theorem~\ref{thm:mix-ool-exact}. Hence, the predictor achieves exact
prediction.
\end{proposition}

Before proving this result, it is worth emphasizing an important difference between
Proposition~\ref{prop:mix-sentinels} and Theorem~\ref{thm:mix-ool-exact}.
Unlike Theorem~\ref{thm:mix-ool-exact}, the present proposition includes a richness
assumption on the data. The reason is that Proposition~\ref{prop:mix-sentinels} aims
at a stronger conclusion: not exact prediction for one specific feasible trajectory,
but exact prediction for \emph{every} length-$T$ trajectory generated by the new system (in fact, by {\em any} new single-output diagonalizable system that satisfies \eqref{spectral.condition}). 

Indeed, under the spectral inclusion condition \eqref{spectral.condition}, every
trajectory of $(A_{\mathrm{new}},C_{\mathrm{new}})$ can be expressed as a linear
combination of primitives  associated with visible eigenvalues of the
prototype systems. Data richness then guarantees that such primitives are actually
encoded in the data library, so that every trajectory of the new system belongs to
$\im(H_Y)$ and therefore satisfies the hypotheses of
Theorem~\ref{thm:mix-ool-exact}.
This also clarifies that data richness is not intrinsically needed for exact prediction
of a \emph{specific} new trajectory. For a fixed system $(A_{\mathrm{new}},C_{\mathrm{new}})$,
or even for a fixed trajectory $Y^\star$, it is enough that the library encode the
particular natural modes that are active in that trajectory. Richness is only needed
to guarantee this property uniformly over all trajectories generated by the prototype
or new systems under consideration.

\begin{proof}
We first prove that visible the spectral inclusion implies
\begin{equation} \label{eq:image-inclusion}
\im\!\bigl(\mathcal{O}_T(A_{\mathrm{new}},C_{\mathrm{new}})\bigr)
\subseteq
\im\!\bigl(\mathcal{O}_T(A_{\mathrm{mix}},C_{\mathrm{mix}})\bigr).
\end{equation}
Let
\[
Y\in\im\!\bigl(\mathcal{O}_T(A_{\mathrm{new}},C_{\mathrm{new}})\bigr).
\]
If $Y=0$, then trivially
$
Y\in\im\!\bigl(\mathcal{O}_T(A_{\mathrm{mix}},C_{\mathrm{mix}})\bigr).
$
Assume therefore that $Y\neq0$. By definition of the image, there exists
an initial state $x_0 \in \mathbb{R}^\nu$ such that
$
Y=\mathcal{O}_T(A_{\mathrm{new}},C_{\mathrm{new}})x_0.
$
Let $\mu_1,\ldots,\mu_m$ denote the distinct eigenvalues of
$A_{\mathrm{new}}$. Since $A_{\mathrm{new}}$ is diagonalizable, we can write
\[
x_0=\sum_{\ell=1}^m x_\ell
\]
where $x_\ell\in\ker(\mu_\ell I-A_{\mathrm{new}})$. Moreover, each $x_\ell$ satisfies 
$A_{\mathrm{new}}^k x_\ell = \mu_\ell^k x_\ell$ for all $k \geq 0$.
Hence,
\[
Y
=
\sum_{\ell=1}^m
\mathcal{O}_T(A_{\mathrm{new}},C_{\mathrm{new}})x_\ell
=
\sum_{\ell=1}^m
\eta_\ell v_T(\mu_\ell),
\]
where $\eta_\ell:=C_{\mathrm{new}}x_\ell$ and
\[
v_T(\mu):=
\begin{bmatrix}
1&\mu&\cdots&\mu^{T-1}
\end{bmatrix}^{\top}
\]
is the $T$-long \emph{Vandermonde vector} associated with $\mu$.

Since $Y\neq0$, at least one coefficient $\eta_\ell$ is nonzero.
For every $\ell$ such that $\eta_\ell\neq0$, we have  
$x_\ell\in\ker(\mu_\ell I-A_{\mathrm{new}})$ and $C_{\mathrm{new}}x_\ell=\eta_\ell\neq0$.
Hence
$
\mu_\ell\in\spec_{C_{\mathrm{new}}}(A_{\mathrm{new}}).
$
By the spectral inclusion assumption,
\[
\mu_\ell\in\spec_{C_{\mathrm{mix}}}(A_{\mathrm{mix}}).
\]
Hence, there exists a vector
$
w_\ell\in\ker(\mu_\ell I-A_{\mathrm{mix}})
$
such that
$
C_{\mathrm{mix}}w_\ell\neq0.
$
Consequently,
\[
\mathcal{O}_T(A_{\mathrm{mix}},C_{\mathrm{mix}})w_\ell
=
(C_{\mathrm{mix}}w_\ell)v_T(\mu_\ell),
\]
and therefore
$
v_T(\mu_\ell)
\in
\im\!\bigl(\mathcal{O}_T(A_{\mathrm{mix}},C_{\mathrm{mix}})\bigr).
$
This holds for every $\ell$ such that $\eta_\ell\neq0$. Hence
\[
Y
=
\sum_{\ell:\,\eta_\ell\neq0}
\eta_\ell v_T(\mu_\ell)
\in
\im\!\bigl(\mathcal{O}_T(A_{\mathrm{mix}},C_{\mathrm{mix}})\bigr),
\]
because
$
\im\!\bigl(\mathcal{O}_T(A_{\mathrm{mix}},C_{\mathrm{mix}})\bigr)
$
is a linear subspace. Since $Y$ was arbitrary, we have proved \eqref{eq:image-inclusion}.

We now verify assumptions (i)--(ii) of
Theorem~\ref{thm:mix-ool-exact}. Let $Y^\star$ be an arbitrary length-$T$
trajectory of $(A_{\mathrm{new}},C_{\mathrm{new}})$. By the inclusion just
proved,
$
Y^\star
\in
\im\!\bigl(\mathcal{O}_T(A_{\mathrm{mix}},C_{\mathrm{mix}})\bigr).
$
Recall, see Lemma~\ref{prop:agg-factorization}, that the output library
admits the factorization
\[
H_Y
=
\mathcal{O}_T(A_{\mathrm{mix}},C_{\mathrm{mix}})X_{\mathrm{mix}}.
\]
By data richness, $\rank(X_{\mathrm{mix}})=n$, and therefore
\[
\im(H_Y)
=
\im\!\bigl(\mathcal{O}_T(A_{\mathrm{mix}},C_{\mathrm{mix}})\bigr).
\]
Hence
\[
Y^\star\in\im(H_Y).
\]
In particular,
$
Y_{\mathrm{past}}\in\im(H_p),
$
which is item~(i) of Theorem~\ref{thm:mix-ool-exact}.

We now verify item~(ii). Since $Y^\star\in\im(H_Y)$, there exists $g$ such that
$
Y^\star=H_Yg.
$
Therefore,
\[
Y_{\mathrm{past}}=H_pg,
\quad
Y_{\mathrm{fut}}=H_fg.
\]
Since $r\ge s_{\mathrm{mix}}$, the continuation identity for the aggregate
system gives
$
H_f=L_{\mathrm{mix}}H_p.
$
Thus
\[
Y_{\mathrm{fut}}
=
H_fg
=
L_{\mathrm{mix}}H_pg
=
L_{\mathrm{mix}}Y_{\mathrm{past}}.
\]
Therefore item~(ii) of Theorem~\ref{thm:mix-ool-exact} holds for
$Y^\star$.

Since $Y^\star$ was arbitrary, Theorem~\ref{thm:mix-ool-exact} applies and
yields exact prediction.
\end{proof}

From Proposition \ref{prop:mix-sentinels} and its proof, we can extract the following corollary, which is instrumental for the developments of Section \ref{sec:output-containment-immersion}. 

\begin{corollary}[Spectral condition and output inclusion]
\label{cor:mix-sentinels}
Let the prototype system matrices  
$A_i$, $i=1,\ldots,Q$, and the new system matrix
$A_{\mathrm{new}}$ be diagonalizable. Suppose $p=1$. If the spectral condition \eqref{spectral.condition} holds, then the output inclusion property 
\eqref{eq:image-inclusion} holds, i.e. 
\[\begin{array}{c}
\spec_{C_{\mathrm{new}}}(A_{\mathrm{new}})
\subseteq
\spec_{C_{\mathrm{mix}}}(A_{\mathrm{mix}})\\[0.1cm]
\Downarrow \\[0.1cm]
\im (\mathcal{O}_T(A_{\mathrm{new}},C_{\mathrm{new}}))
\subseteq
\im (\mathcal{O}_T(A_{\mathrm{mix}},C_{\mathrm{mix}})).
\end{array}\]
\end{corollary}

We refer to $\im(\mathcal{O}_T(A_{\mathrm{new}},C_{\mathrm{new}}))
\subseteq
\im(\mathcal{O}_T(A_{\mathrm{mix}},C_{\mathrm{mix}}))$ as an output inclusion property because of its equivalence to 
\begin{equation}\label{output-set-inclusion}
\mathcal{Y}_{T,\mathrm{new}}\subseteq \mathcal{Y}_{T,\mathrm{mix}},
\end{equation}
where the sets $\mathcal{Y}_{T}$ collect all the possible length-$T$ output trajectories that can be generated by a system: 
\[
\begin{array}{rcl}
\mathcal{Y}_{T,\mathrm{new}} &=& \{\mathcal{O}_T(A_{\mathrm{new}},C_{\mathrm{new}})x_{0,\mathrm{new}} \colon x_{0,\mathrm{new}}  \in \mathbb{R}^\nu\},\\[0.1cm]
\mathcal{Y}_{T,\mathrm{mix}} &=& \{\mathcal{O}_T(A_{\mathrm{mix}},C_{\mathrm{mix}})x_{0,\mathrm{mix}} \colon  x_{0,\mathrm{mix}} \in \mathbb{R}^{n}\}.
\end{array}\]
We will return to the output inclusion property \eqref{output-set-inclusion} in Section \ref{sec:output-containment-immersion} when dealing with the prediction of trajectories generated by nonlinear systems. 

Proposition~\ref{prop:mix-sentinels} and Corollary \ref{cor:mix-sentinels} show that, in the scalar-output case,
visible spectral inclusion is sufficient for exact prediction. The next
result shows that, when the horizon $T$ is sufficiently long, the same
spectral condition is also necessary for reproducing all length-$T$ output
signals of the new system. 

\begin{proposition}[Necessary spectral condition for scalar-output reproducibility] \label{prop:necessity}
Let $(A_{\mathrm{mix}},C_{\mathrm{mix}})$ and $(A_{\mathrm{new}},C_{\mathrm{new}})$ be linear systems, with
$A_{\mathrm{mix}}\in\mathbb R^{n\times n}$ and 
$A_{\mathrm{new}}\in\mathbb R^{\nu\times \nu}$ 
diagonalizable. 
Suppose $p=1$.
Assume that
\begin{equation}
\im (\mathcal{O}_T(A_{\mathrm{new}},C_{\mathrm{new}})) \subseteq \im (\mathcal{O}_T(A_{\mathrm{mix}},C_{\mathrm{mix}}))
\end{equation}
for some
$
T>\operatorname{card}\bigl(\spec_{C_{\mathrm{mix}}}(A_{\mathrm{mix}})\bigr)
$.
Then
\[
\spec_{C_{\mathrm{new}}}(A_{\mathrm{new}})\subseteq \spec_{C_{\mathrm{mix}}}(A_{\mathrm{mix}}).
\]
\end{proposition}

\begin{proof}
Let
$
\lambda_1,\ldots,\lambda_m
$ 
be the distinct visible eigenvalues of $A_{\mathrm{mix}}$, 
i.e., $\spec_{C_{\mathrm{mix}}}(A_{\mathrm{mix}})=\{\lambda_1,\ldots,\lambda_m\}$. As shown in
Proposition~\ref{prop:mix-sentinels}, since $A_{\mathrm{mix}}$ is diagonalizable, every
$T$-long output trajectory of $(A_{\mathrm{mix}},C_{\mathrm{mix}})$ belongs to
\[
\operatorname{span}\{v_T(\lambda_i):i=1,\ldots,m\}
\]
where $v_T(\lambda_i)$ is the $T$-long {Vandermonde vector} associated with $\lambda_i$. 

Hence, letting
\[
\Phi_T:=
\begin{bmatrix}
1 & 1 & \cdots & 1\\
\lambda_1 & \lambda_2 & \cdots & \lambda_m\\
\lambda_1^2 & \lambda_2^2 & \cdots & \lambda_m^2\\
\vdots & \vdots & & \vdots\\
\lambda_1^{T-1} & \lambda_2^{T-1} & \cdots & \lambda_m^{T-1}
\end{bmatrix},
\]
we have $\im (\mathcal{O}_T(A_{\mathrm{mix}},C_{\mathrm{mix}})) \subseteq \im(\Phi_T)$.

Further recall that, for each $\lambda_i\in\spec_{C_{\mathrm{mix}}}(A_{\mathrm{mix}})$, there exists
$
w_i\in\ker(\lambda_i I-A_{\mathrm{mix}})
$
such that $C_{\mathrm{mix}}w_i\neq0$. Then
\[
\mathcal{O}_T(A_{\mathrm{mix}},C_{\mathrm{mix}})
w_i=(C_{\mathrm{mix}} w_i)v_T(\lambda_i),
\]
and therefore $v_T(\lambda_i) \in \im (\mathcal{O}_T( A_{\mathrm{mix}},C_{\mathrm{mix}}))$. 

Thus
\[
\im (\mathcal{O}_T(A_{\mathrm{mix}},C_{\mathrm{mix}})) = \im(\Phi_T)
\]

Let now $\mu$ be a visible eigenvalue of $A_{\mathrm{new}}$.
Reasoning as before, we have 
$
v_T(\mu) \in \im (\mathcal{O}_TA_{\mathrm{new}},C_{\mathrm{new}})).
$
By the assumption
$
\im (\mathcal{O}_T(A_{\mathrm{new}},C_{\mathrm{new}})) \subseteq \im (\mathcal{O}_T(A_{\mathrm{mix}},C_{\mathrm{mix}}))
$
it follows that
\[
v_T(\mu) \in \im (\mathcal{O}_T(A_{\mathrm{mix}},C_{\mathrm{mix}})).
\]
Suppose, by contradiction, that $\mu$
is not a visible eigenvalue of $A_{\mathrm{mix}}$. 
Then $\lambda_1,\ldots,\lambda_m,\mu$ are pairwise distinct. Since $T>m$, the columns of 
\[
\overline \Phi_T := \begin{bmatrix} \Phi_T &  v_T(\mu) \end{bmatrix}
\]
are linearly independent, because they contain an $(m+1)\times(m+1)$
Vandermonde matrix associated with the distinct eigenvalues
$\lambda_1,\ldots,\lambda_m,\mu$. Hence
$
v_T(\mu)\notin \im(\Phi_T).
$
This contradicts $v_T(\mu)\in\im(\mathcal{O}_T(A_{\mathrm{mix}},C_{\mathrm{mix}}))=\im(\Phi_T)$. Therefore
$\mu\in\spec_{C_{\mathrm{mix}}}(A_{\mathrm{mix}})$.
Since $\mu$ was arbitrary in the visible spectrum of $A_{\mathrm{new}}$, the claim follows.
\end{proof}


\subsection{Sparse prediction in modal coordinates}
\label{subsec:sparse-modal-prediction}

The previous analysis suggests that accurate predictions may be more easily obtained with large libraries, namely libraries containing a broad variety of natural modes. This, however, raises a practical issue. To obtain theoretical guarantees on the prediction, one has to satisfy the condition $r\ge s_{\mathrm{mix}}$, and the index $s_{\mathrm{mix}}$ generally increases with the dimension $n$ of the aggregate system associated with the library. This is not entirely surprising, and reflects a natural trade-off between the descriptive capability of the library and its complexity. Indeed, a large library allows us to predict systems whose order can be as large as that represented by the aggregate library. Thus, the condition $r\ge s_{\mathrm{mix}}$ is not an artifact, but rather the price of requiring a prediction guarantee uniformly over the whole aggregate library, independently of the specific order $\nu$ of the system to be predicted.

If, however, $\nu \ll n$, one may seek a condition that depends on $\nu$ rather than on $n$. In other words, one would like to control the \emph{complexity} of the solution. A simple way to do this is to represent the library directly in modal coordinates. This is natural when the templates are synthetic or designed from prescribed modes. In this case, the columns of the library are not arbitrary sampled trajectories, but the elementary modal trajectories associated with the visible eigenvalues. Thus, sparsity has a direct dynamical meaning: a trajectory generated by a system of order $\nu$ activates at most $\nu$ visible modes, independently of how many additional modes are available in the library.

Let $\mathcal V:=\{\lambda_1,\ldots,\lambda_M\}$ denote the set of distinct visible eigenvalues contained in the aggregate template family, i.e., $\mathcal V=\spec_{C_{\mathrm{mix}}}(A_{\mathrm{mix}})$. Define the modal library
\[
\Phi_T:=
\begin{bmatrix}
1 & 1 & \cdots & 1\\
\lambda_1 & \lambda_2 & \cdots & \lambda_M\\
\lambda_1^2 & \lambda_2^2 & \cdots & \lambda_M^2\\
\vdots & \vdots & & \vdots\\
\lambda_1^{T-1} & \lambda_2^{T-1} & \cdots & \lambda_M^{T-1}
\end{bmatrix},
\quad
\Phi_T=
\begin{bmatrix}
\Phi_p\\[1mm]
\Phi_f
\end{bmatrix},
\]
where $\Phi_p\in \mathbb{C}^{r \times M}$ contains the first $r$ rows and $\Phi_f\in \mathbb{C}^{T-r\times M}$ the remaining $T-r$ rows. In these coordinates, a coefficient vector $a\in\mathbb{C}^M$ specifies directly which visible modes are active. This should be contrasted with the data library $H_Y$, whose columns are generic sample trajectories and therefore generic linear combinations of modes. A vector that is sparse in modal coordinates need not be sparse in the original data coordinates.

We recall the following notion. For a matrix $M$, its spark is
\[
\operatorname{spark}(M)
:=
\min\{\|z\|_0:\ z\neq 0,\ Mz=0\}.
\]
Thus, $\operatorname{spark}(M)$ is the smallest support size of a nonzero vector in the kernel of $M$. Equivalently, it is the smallest number of columns of $M$ that can satisfy a nontrivial linear dependence. The condition $\operatorname{spark}(\Phi_p)>2\nu$ ensures uniqueness of $\nu$-sparse representations from the past window.

\begin{proposition}[Exact prediction from a sparse modal library]
\label{prop:sparse-modal-prediction}
Consider the same scalar-output setting of Proposition~\ref{prop:mix-sentinels}. Let $\mathcal V$ be the set of distinct visible eigenvalues of the aggregate template family $(A_{\mathrm{mix}},C_{\mathrm{mix}})$, and let $\Phi_T$ be the corresponding modal library. Let $(A_{\mathrm{new}},C_{\mathrm{new}})$ be diagonalizable and suppose that
\[
\spec_{C_{\mathrm{new}}}(A_{\mathrm{new}})
\subseteq
\mathcal V .
\]
Assume that $(A_{\mathrm{new}},C_{\mathrm{new}})$ has order $\nu$. If
\[
\operatorname{spark}(\Phi_p)>2\nu,
\]
then, for every length-$T$ trajectory
\[
Y^\star=
\begin{bmatrix}
Y_{\mathrm{past}}\\[1mm]
Y_{\mathrm{fut}}
\end{bmatrix}
\]
generated by $(A_{\mathrm{new}},C_{\mathrm{new}})$, there exists a unique $\nu$-sparse vector $a^\star \in \mathbb{C}^M$, 
i.e.~$\|a^\star\|_0 \le \nu$, satisfying
\[
Y_{\mathrm{past}}=\Phi_p a^\star .
\]
Moreover, the prediction
\[
\widehat Y_{\mathrm{fut}}:=\Phi_f a^\star
\]
is exact, namely
\[
\widehat Y_{\mathrm{fut}}=Y_{\mathrm{fut}}.
\]
\end{proposition}

\begin{proof}
Let $Y^\star$ be generated by $(A_{\mathrm{new}},C_{\mathrm{new}})$. Since the system is diagonalizable and has scalar output, its output can be written as a linear combination of the modal trajectories associated with its distinct visible eigenvalues. Hence there exist coefficients $\eta_\ell$ and visible eigenvalues $\mu_\ell\in\spec_{C_{\mathrm{new}}}(A_{\mathrm{new}})$ such that
\[
Y^\star
=
\sum_{\ell=1}^{m}
\eta_\ell v_{\mu_\ell}^{(T)},
\quad
v_{\mu}^{(T)}
:=
\begin{bmatrix}
1 & \mu & \cdots & \mu^{T-1}
\end{bmatrix}^{\top},
\]
where $m\le \nu$. By the spectral inclusion assumption, each $\mu_\ell$ belongs to $\mathcal V$. Therefore, each vector $v_{\mu_\ell}^{(T)}$ is a column of the modal library $\Phi_T$. Thus there exists a coefficient vector $a^\star\in\mathbb{C}^M$ with
$
\|a^\star\|_0\le \nu
$
such that
$
Y^\star=\Phi_Ta^\star.
$
Partitioning $\Phi_T$ into past and future rows gives
\[
Y_{\mathrm{past}}=\Phi_p a^\star,
\quad
Y_{\mathrm{fut}}=\Phi_f a^\star.
\]

It remains to show uniqueness among $\nu$-sparse representations. Suppose that $a_1$ and $a_2$ are two vectors satisfying
\[
\Phi_p a_1=\Phi_p a_2,
\quad
\|a_1\|_0\le \nu,
\quad
\|a_2\|_0\le \nu.
\]
Then
$
\Phi_p(a_1-a_2)=0.
$
Moreover, $a_1-a_2$ has support size at most $2\nu$. If $a_1\neq a_2$, then $a_1-a_2$ is a nonzero vector in $\ker(\Phi_p)$ with at most $2\nu$ nonzero entries. This contradicts
$
\operatorname{spark}(\Phi_p)>2\nu.
$
Hence, $a_1=a_2$. Therefore, the $\nu$-sparse vector $a^\star$ matching the past window is unique.

Consequently, the predicted future satisfies
\[
\widehat Y_{\mathrm{fut}}
=
\Phi_f a^\star
=
Y_{\mathrm{fut}},
\]
which proves exact prediction.
\end{proof}

\begin{remark}[Spark condition and visible modes] \label{rem:sparsity}
The bound in Proposition~\ref{prop:sparse-modal-prediction} is
stated in terms of the order $\nu$ of the new system, since a
diagonalizable scalar-output system of order $\nu$ can activate at most
$\nu$ visible modes. If more information is available and the trajectory
is known to activate only $\nu_o \le \nu$ modes, then the same proof shows
that the condition can be weakened to
$
\operatorname{spark}(\Phi_p)>2\nu_o.
$
\end{remark}

In the scalar-output case, as the columns of $\Phi_p$ are generated by distinct modes, $\Phi_p$ is a Vandermonde matrix. Thus, if the eigenvalues are distinct and the number of visible modes is larger than $r$, then
$
\operatorname{spark}(\Phi_p)=r+1.
$
Thus the condition $\operatorname{spark}(\Phi_p)>2\nu$ is ensured by
\begin{equation}\label{r.maggiore.di.2nu}
r\ge 2\nu.
\end{equation}
This replaces the global requirement $r\ge s_{\mathrm{mix}}$ by a condition which depends on the number of modes of the trajectory to be predicted, rather than on the total number of modes contained in the library. 

While this discussion is restricted to scalar-output systems, it suggests that the complexity of a large library can be controlled: one may preserve a high descriptive capability without requiring long past trajectory portions to guarantee correct predictions.

\begin{example}[Sparse prediction from modal library]
Consider the problem of predicting the output of a linear system with\footnote{Later in the paper, we will use again this system in the context of nonlinear immersion.}
\[
A_{\mathrm{new}}=\diag(\rho,\mu,\rho^2),
\quad
C_{\mathrm{new}}=
\begin{bmatrix}
0&1&2
\end{bmatrix}.
\]
We set $\rho=0.8$ and $\mu=0.5$.
Note that this system has order three but its output depends only on the two visible modes $\mu$ and $\rho^2$.

We build a library using five second-order linear systems with matrices
\[
A_1=\begin{bmatrix} \rho & 0 \\ 0 & 0.15 \end{bmatrix},\,
A_2=\begin{bmatrix} \mu & 0 \\ 0 & -0.25 \end{bmatrix},\,
A_3=\begin{bmatrix} \rho^2 & 0 \\ 0 & 0.35 \end{bmatrix},
\]
\[
A_4=\begin{bmatrix} 0.72 & 0 \\ 0 & -0.55 \end{bmatrix},\,
A_5=\begin{bmatrix} 0.05 & 0 \\ 0 & 0.90 \end{bmatrix},
\]
and $C_i=\begin{bmatrix}1&1\end{bmatrix}$, $i=1,\ldots,5$.

The aggregate sentinel library therefore contains ten modes.
Since these modes are distinct and visible, the aggregate library
has $s_{\mathrm{mix}}=10$.
Thus, applying Proposition~\ref{prop:mix-sentinels} directly to the full aggregate library would require a past window of length $r\geq 10$.
Instead, we use the modal form of the same library. Namely, we define
\[
\Phi_T
=
\begin{bmatrix}
1 & 1 & \cdots & 1\\
\lambda_1 & \lambda_2 & \cdots & \lambda_{10}\\
\lambda_1^2 & \lambda_2^2 & \cdots & \lambda_{10}^2\\
\vdots & \vdots & & \vdots\\
\lambda_1^{T-1} & \lambda_2^{T-1} & \cdots & \lambda_{10}^{T-1}
\end{bmatrix}
=
\begin{bmatrix}
\Phi_p\\
\Phi_f
\end{bmatrix},
\]
where $\Phi_p$ contains the first $r$ rows and $\Phi_f$ the remaining $T-r$ rows. 
The coefficient vector in this representation directly selects the active
modes. In the present example, although $A_{\mathrm{new}}$ has order
$\nu=3$, the output uses only the $\nu_o=2$ visible modes $\mu$ and
$\rho^2$. Hence the corresponding trajectory is $\nu_o$-sparse in the
modal library, i.e., there exists a vector $a^\star$, with
$\|a^\star\|_0\le \nu_o$, such that
$
Y^\star=\Phi_T a^\star .
$

We therefore take $r=4$. Since the modes $\lambda_i$ are all distinct,
$\Phi_p$ is a Vandermonde matrix and
\[
\operatorname{spark}(\Phi_p)=r+1=5>2\nu_o.
\]
By Proposition~\ref{prop:sparse-modal-prediction}, together with
Remark~\ref{rem:sparsity}, the $\nu_o$-sparse representation matching
the past is unique. Thus, if $a^\star$ denotes the unique
$\nu_o$-sparse solution of
$
Y_{\mathrm{past}}=\Phi_p a^\star,
$
then the future is recovered exactly as
$
\widehat Y_{\mathrm{fut}}=\Phi_f a^\star
$.

Figure~\ref{fig:sparse-modal-prediction} shows the result taking $T=10$ and $r=4$. 
The prediction is exact even though $r \ll s_{\mathrm{mix}}$. 
\end{example}

\begin{figure}[t]
\centering
\includegraphics[width=0.75\linewidth]{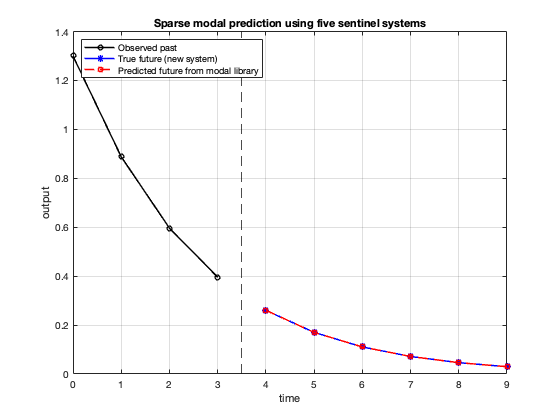}
\caption{Sparse modal prediction using five second-order sentinel systems. The full library contains ten visible modes, but the new output trajectory activates only two of them. Hence $r=4$ is sufficient for exact sparse prediction.}
\label{fig:sparse-modal-prediction}
\end{figure}

For the sake of brevity, this discussion has been restricted to the case of exact spectral inclusion. A rigorous extension to the mismatched case, possibly building on Corollary~\ref{cor:out-of-lib-error}, is left for future research. Such an extension would be useful to understand how large libraries can provide rich descriptive capabilities while controlling the complexity of the solution. In what follows, we take a different perspective and focus on a more conceptual aspect of data-driven \emph{generative} capability, 
namely how new trajectory libraries can be formed through \emph{interconnections of data}.

\section{Out-of-library prediction via interconnected libraries}
\label{subsec:feedback_composed_libraries}

In the previous section we have shown that mixed libraries can generalize beyond the nominal modes because linear combinations of library columns may generate trajectories that do not belong to any single atomic system. We now describe a complementary mechanism, which we call \emph{compositional generalization}: new behaviors may also arise by \emph{interconnecting} atomic systems. In particular, feedback interconnections may generate closed-loop natural modes that are not present in either subsystem taken in isolation. This suggests that, for exact out-of-library prediction, it is not necessary to require that the new system be represented by the union of the atomic spectra only; it is enough that its trajectories belong to a larger behavior generated by suitable interconnections of atomic libraries.

For simplicity, and consistently with Section~6.2.1, we restrict the attention to the SISO case. Unlike the autonomous setting considered in the rest of the paper, in this subsection each atomic system is viewed through its input-output behavior over a finite horizon. The goal is to construct, directly from the individual atomic libraries, a new library associated with their feedback interconnection, without collecting data from the interconnected system itself. 

\subsubsection{Feedback-interconnected library}

Let $\Sigma_1$ and $\Sigma_2$ be two SISO strictly causal linear systems with input-output pairs $(u_1,y_1)$ and $(u_2,y_2)$ 
and state-space representations$(A_i,B_i,C_i)$.\footnote{Strict causality makes 
sure that the feedback interconnection of $\Sigma_1$ and $\Sigma_2$ is well posed.}
Fix a horizon $T$ and, for each system $\Sigma_i$, assume that a noise-free library of length-$T$ \emph{input-output trajectories} is available in the form
\begin{equation}
H_W^{(i)}=
\begin{bmatrix}
H_u^{(i)}\\
H_y^{(i)}
\end{bmatrix}
\in\mathbb{R}^{2T\times N_i},
\quad i=1,2,
\label{eq:atomic_full_io_library}
\end{equation}
where $H_u^{(i)},H_y^{(i)}\in\mathbb{R}^{T\times N_i}$. 
We do not assume that any such trajectory of $H_W^{(i)}$ has been generated by feedback interconnecting $\Sigma_1$ and $\Sigma_2$. 
Nonetheless, we want to show how to 
 generate a data library containing trajectories from a dynamical system corresponding to the feedback interconnection of $\Sigma_1$ and $\Sigma_2$. Clearly, this is useful when data come from the individual systems.
 
As a first step, observe that, for $i=1,2$,  $\im\,(H_W^{(i)})$ is the space of all input-output trajectories of $\Sigma_i$ spanned by the length-$T$ trajectories of  $H_W^{(i)}$.  To build from $\im\,(H_W^{(1)})$ and $\im\,(H_W^{(2)})$ a data library containing trajectories of the feedback interconnection of $\Sigma_1$ and $\Sigma_2$, 
we consider the subspace of the trajectories of $\im\,(H_W^{(1)})$ and $\im\,(H_W^{(2)})$ that satisfy the feedback constraint
\begin{equation}
u_1 = y_2,
\quad
u_2 = y_1.
\label{eq:fb_constraint_full_window}
\end{equation}
This is tantamount to find a basis for the intersection of $\im\,(H_W^{(1)})$ and $\im\,(PH_W^{(2)})$, where $P$ is the block swap matrix $P=\left[\begin{smallmatrix}0 & I_T\\ I_T & 0\end{smallmatrix}\right]$. Equivalently, one could find the basis of the intersection of $\im\,(P H_W^{(1)})$ and $\im\,(H_W^{(2)})$. These bases can be computed as follows: 

\begin{lemma}[Feedback compatibility of atomic libraries]
\label{lem:library-from-atomic-trajectories-fb}
Set
\[
\mathcal Q :=
\begin{bmatrix}
H_u^{(1)} & -H_y^{(2)}\\
-H_y^{(1)} & H_u^{(2)}
\end{bmatrix}.
\]
Let $\Theta \in\mathbb{R}^{(N_1+N_2)\times \nu}$ be such that
\begin{equation}
\mathrm{Im}(\Theta)=\ker(\mathcal Q),
\label{eq:nullspace_basis_comp}
\end{equation}
and consider the partition
\begin{equation}
\Theta=
\begin{bmatrix}
\Theta_1\\
\Theta_2
\end{bmatrix},
\quad
\Theta_1\in\mathbb{R}^{N_1\times \nu}, \,\,
\Theta_2\in\mathbb{R}^{N_2\times \nu}.
\label{eq:nullspace_partition_comp}
\end{equation}
Then
\[
\im(H_W^{(1)}) \cap \im(PH_W^{(2)}) = \im(H_W^{(1)}\Theta_1),
\]
and
\[
\im(H_W^{(2)}) \cap \im(PH_W^{(1)}) = \im(H_W^{(2)}\Theta_2).
\]
\end{lemma}

\begin{proof}
We first prove
\[
\im(H_W^{(1)}\Theta_1)=\im(PH_W^{(2)}\Theta_2).
\]
By \eqref{eq:nullspace_basis_comp}, we have
$
\mathcal Q \Theta = 0.
$
Using the block partition \eqref{eq:nullspace_partition_comp}, this is equivalent to
$H_u^{(1)}\Theta_1 = H_y^{(2)}\Theta_2$,
and 
$H_y^{(1)}\Theta_1 = H_u^{(2)}\Theta_2$.
Therefore,
\[
H_W^{(1)}\Theta_1
=
\begin{bmatrix}
H_u^{(1)}\Theta_1\\
H_y^{(1)}\Theta_1
\end{bmatrix}
=
\begin{bmatrix}
H_y^{(2)}\Theta_2\\
H_u^{(2)}\Theta_2
\end{bmatrix}
=
PH_W^{(2)}\Theta_2.
\]
Hence,  $\im(H_W^{(1)}\Theta_1)=\im(PH_W^{(2)}\Theta_2)$, as claimed.  

We now prove
\[
\im(H_W^{(1)}) \cap \im(PH_W^{(2)}) = \im(H_W^{(1)}\Theta_1).
\]

($\supseteq$) First, let $v\in \im(H_W^{(1)}\Theta_1)$. Then there exists $\alpha\in\mathbb R^\nu$ such that
$
v = H_W^{(1)}\Theta_1\alpha.
$
Clearly, $v\in \im(H_W^{(1)})$. Further, by the identity already proved,
$
v\in \im(PH_W^{(2)}\Theta_2)\subseteq \im(PH_W^{(2)}).
$
Thus
\[
v\in \im(H_W^{(1)}) \cap \im(PH_W^{(2)}).
\]

($\subseteq$) Conversely, assume
$
v\in \im(H_W^{(1)}) \cap \im(PH_W^{(2)}).
$
Then there exist $g_1\in\mathbb R^{N_1}$ and $g_2\in\mathbb R^{N_2}$ such that
\[
v = H_W^{(1)}g_1 = PH_W^{(2)}g_2.
\]
Writing this equality blockwise gives
\[
\begin{bmatrix}
H_u^{(1)}g_1\\
H_y^{(1)}g_1
\end{bmatrix}
=
\begin{bmatrix}
H_y^{(2)}g_2\\
H_u^{(2)}g_2
\end{bmatrix}.
\]
Equivalently,
\[
\mathcal Q
\begin{bmatrix}
g_1\\
g_2
\end{bmatrix}
=0,
\]
and therefore
\[
\begin{bmatrix}
g_1\\
g_2
\end{bmatrix}
\in \ker(\mathcal Q)=\im(\Theta).
\]
So there exists $\alpha\in\mathbb R^\nu$ such that
\[
\begin{bmatrix}
g_1\\
g_2
\end{bmatrix}
=
\Theta\alpha
=
\begin{bmatrix}
\Theta_1\\
\Theta_2
\end{bmatrix}\alpha.
\]
In particular,
$
g_1=\Theta_1\alpha.
$
Substituting this into $v=H_W^{(1)}g_1$ yields
$
v=H_W^{(1)}\Theta_1\alpha \in \im(H_W^{(1)}\Theta_1).
$
This gives the inclusion
$
\im(H_W^{(1)}) \cap \im(PH_W^{(2)}) \subseteq \im(H_W^{(1)}\Theta_1).
$.
Combining the two inclusions, we obtain
$
\im(H_W^{(1)}) \cap \im(PH_W^{(2)}) = \im(H_W^{(1)}\Theta_1),
$
as claimed.

The proof of
$
\im(H_W^{(2)}) \cap \im(PH_W^{(1)}) = \im(H_W^{(2)}\Theta_2)
$
is completely symmetric.
\end{proof}

We have shown that the trajectories compatible with the feedback interconnection of $\Sigma_1$ and $\Sigma_2$ can be characterized as 
$\im(H_W^{(1)}\Theta_1)$ or $\im(H_W^{(2)}\Theta_2)$. Therefore, the output library $\mathcal H_Y$ associated with such interconnection can be characterized as 
$\im(H_y^{(1)}\Theta_1)$ or $\im(H_y^{(2)}\Theta_2)$, depending on which output we select as the output of the feedback system. In the
following, we fix one of them, say
\begin{equation}\label{H_Y_fb}
\mathcal H_Y:=H_y^{(1)}\Theta_1,
\end{equation}
which corresponds to choose the output $y_1$ of $\Sigma_1$ as the output of the interconnection.

A schematic representation of the mechanism generating a feedback-interconnection from atomic libraries 
is depicted in Figure \ref{fig:feedback_library_construction}.

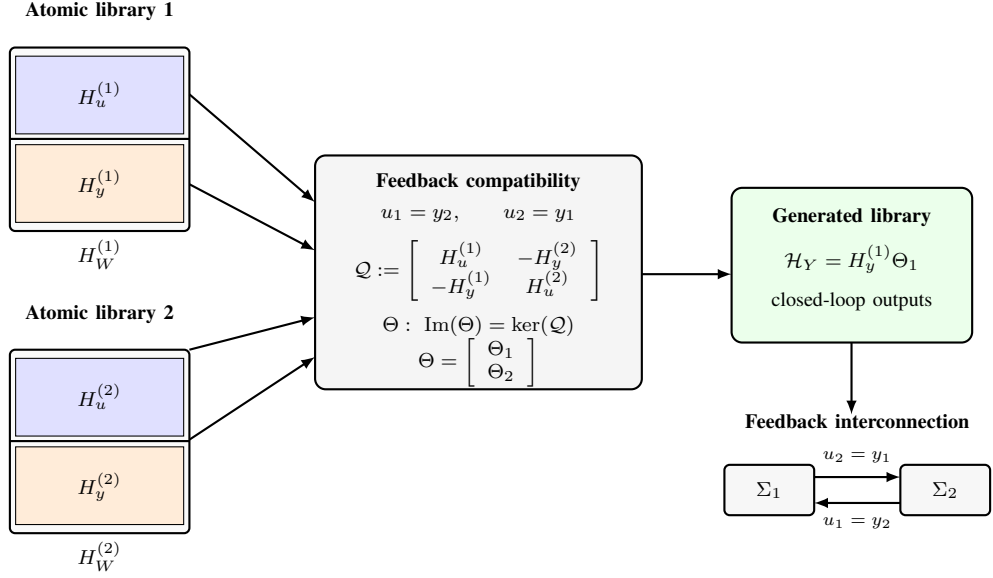
\begin{figure*}[t]
\centering
\footnotesize
\begin{tikzpicture}[>=latex, scale=0.95, transform shape]

\tikzset{
libbox/.style={thick, rounded corners=2pt, fill=gray!4},
inputblock/.style={fill=blue!12},
outputblock/.style={fill=orange!15},
mainbox/.style={thick, rounded corners=4pt, fill=gray!8},
genbox/.style={thick, rounded corners=4pt, fill=green!8},
arr/.style={->, thick}
}

\begin{scope}[shift={(0,2.65)}]
\node at (1.25,3.05) {\textbf{Atomic library 1}};
\draw[libbox] (0,0) rectangle (2.5,2.55);

\draw[inputblock] (0.08,1.35) rectangle (2.42,2.43);
\draw[outputblock] (0.08,0.12) rectangle (2.42,1.20);
\draw[thick] (0,1.28) -- (2.5,1.28);

\node at (1.25,1.90) {$H_u^{(1)}$};
\node at (1.25,0.65) {$H_y^{(1)}$};
\node at (1.25,-0.32) {$H_W^{(1)}$};
\end{scope}

\begin{scope}[shift={(0,-1.55)}]
\node at (1.25,3.05) {\textbf{Atomic library 2}};
\draw[libbox] (0,0) rectangle (2.5,2.55);

\draw[inputblock] (0.08,1.35) rectangle (2.42,2.43);
\draw[outputblock] (0.08,0.12) rectangle (2.42,1.20);
\draw[thick] (0,1.28) -- (2.5,1.28);

\node at (1.25,1.90) {$H_u^{(2)}$};
\node at (1.25,0.65) {$H_y^{(2)}$};
\node at (1.25,-0.32) {$H_W^{(2)}$};
\end{scope}

\begin{scope}[shift={(4.25,0.45)}]
\draw[mainbox] (0,0) rectangle (4.55,3.25);

\node at (2.275,2.88) {\textbf{Feedback compatibility}};
\node at (2.275,2.42) {$u_1=y_2,\qquad u_2=y_1$};

\node at (2.275,1.65) {
$\mathcal Q :=
\left[
\begin{array}{cc}
H_u^{(1)} & -H_y^{(2)}\\
-H_y^{(1)} & H_u^{(2)}
\end{array}
\right]$
};

\node at (2.275,0.88) {$\Theta:\ \im(\Theta)=\ker(\mathcal Q)$};
\node at (2.275,0.38) {$\Theta=
\left[
\begin{array}{c}
\Theta_1\\
\Theta_2
\end{array}
\right]$};
\end{scope}

\draw[arr] (2.5,4.55) -- (4.25,3.05);
\draw[arr] (2.5,3.30) -- (4.25,2.38);

\draw[arr] (2.5,1.00) -- (4.25,1.45);
\draw[arr] (2.5,-0.25) -- (4.25,0.90);

\begin{scope}[shift={(10.05,1.10)}]
\draw[genbox] (0,0) rectangle (3.35,2.15);
\node at (1.675,1.74) {\textbf{Generated library}};
\node at (1.675,1.16) {$\mathcal H_Y=H_y^{(1)}\Theta_1$};
\node at (1.675,0.58) {closed-loop outputs};
\end{scope}

\draw[arr] (8.80,2.05) -- (10.05,2.05);

\begin{scope}[shift={(9.95,-1.65)}]
\node at (1.85,1.65) {\textbf{Feedback interconnection}};

\draw[thick, rounded corners=2pt, fill=gray!6] (0,0.35) rectangle (1.25,1.05);
\draw[thick, rounded corners=2pt, fill=gray!6] (2.45,0.35) rectangle (3.70,1.05);

\node at (0.625,0.70) {$\Sigma_1$};
\node at (3.075,0.70) {$\Sigma_2$};

\draw[arr] (1.25,0.88) -- (2.45,0.88);
\draw[arr] (2.45,0.52) -- (1.25,0.52);

\node at (1.85,1.16) {\scriptsize $u_2=y_1$};
\node at (1.85,0.24) {\scriptsize $u_1=y_2$};
\end{scope}

\draw[arr] (11.72,1.10) -- (11.72,0.10);

\end{tikzpicture}
\caption{Construction of a feedback-generated library from two atomic input-output libraries. The matrix $\Theta$ selects pairs of trajectories satisfying the feedback constraints, producing a new output library associated with the closed-loop interconnection.}
\label{fig:feedback_library_construction}
\end{figure*}

The following lemma provides a state-space factorization of the feedback-interconnected library.

\begin{lemma}[Factorization of the feedback-interconnected library]
\label{lemma:fb_factorization}
Let $\mathcal{Q}$ and $\Theta$ be defined as in Lemma 
\ref{lem:library-from-atomic-trajectories-fb}. Let
\begin{equation}
\label{eq:cl_matrices}
\mathcal A :=
\begin{bmatrix}
A_1 & B_1 C_2\\
B_2 C_1 & A_2
\end{bmatrix},
\quad
\mathcal C :=
\begin{bmatrix}
C_1 & 0
\end{bmatrix}.
\end{equation}
Then there exist matrices $X_1\in\mathbb R^{n_1\times N_1}$ and
$X_2\in\mathbb R^{n_2\times N_2}$ such that, defining
\[
\mathcal X :=
\begin{bmatrix}
X_1\Theta_1\\
X_2\Theta_2
\end{bmatrix},
\]
the feedback-interconnected library satisfies
\[
\mathcal H_Y = O_T(\mathcal A,\mathcal C) \mathcal X,
\]
where $\mathcal{H}_Y$ is defined as in \eqref{H_Y_fb}.
\end{lemma}

\begin{proof}
For each column of $H_W^{(i)}$, let $(u^{(i,j)},y^{(i,j)})$ denote the
corresponding length-$T$ input-output trajectory of $\Sigma_i$, generated from
the initial condition $x_0^{(i,j)}$. Collect these initial conditions as
\[
X_i :=
\begin{bmatrix}
x_0^{(i,1)} & \cdots & x_0^{(i,N_i)}
\end{bmatrix}.
\]
Then, by linearity of $\Sigma_i$, for every coefficient vector
$g_i\in\mathbb R^{N_i}$, the linear combination
$
H_W^{(i)} g_i
$
is again a valid length-$T$ input-output trajectory of $\Sigma_i$. Its input,
output, and initial condition are obtained by combining those of the library
columns with the same coefficients $g_i$; in particular, the corresponding
initial condition is $x_i(0)=X_i g_i$.
Now let $\alpha\in\mathbb R^\nu$ be arbitrary and consider 
$g_1=\Theta_1\alpha$,
$g_2=\Theta_2\alpha$.
Since $\im(\Theta)=\ker(\mathcal Q)$, then $H_u^{(1)}\Theta_1 = H_y^{(2)}\Theta_2$ and $H_y^{(1)}\Theta_1 = H_u^{(2)}\Theta_2$ 
and the pair $(g_1,g_2)$ satisfies the feedback compatibility
constraint, namely
\[
H_u^{(1)}g_1 = H_y^{(2)}g_2,
\quad
H_u^{(2)}g_2 = H_y^{(1)}g_1.
\]
The two trajectories $(H_u^{(1)}g_1,H_y^{(1)}g_1)$ and $(H_u^{(2)}g_2,H_y^{(2)}g_2)$ thus satisfy
the feedback interconnection \eqref{eq:fb_constraint_full_window} over the whole horizon. Therefore, $H_y^{(1)}g_1$ is a valid output trajectory of the feedback interconnection stemming from  the initial condition 
$\left[\begin{smallmatrix} X_1 \Theta_1 \\ X_2 \Theta_2\end{smallmatrix}\right]\alpha=: \mathcal{X}\alpha$. Hence, 
\[
H_y^{(1)} \Theta_1 \alpha= \mathcal H_Y\alpha = O_T(\mathcal A,\mathcal C) \mathcal X\alpha.
\]
As this identity holds for every $\alpha\in\mathbb R^\nu$, it follows that
$
\mathcal H_Y = O_T(\mathcal A,\mathcal C) \mathcal X.
$
\end{proof}

Next, we continue with series and parallel interconnections.

\subsubsection{Series-interconnected library}

Let $\Sigma_1$ and $\Sigma_2$ be two SISO strictly causal LTI systems with input-output pairs
$(u_1,y_1)$ and $(u_2,y_2)$ and state-space realizations $(A_i,B_i,C_i)$, $i=1,2$.
We consider the series interconnection in which $\Sigma_1$ is upstream and $\Sigma_2$ is downstream,
namely
\begin{equation}
u_2 = y_1.
\label{eq:series_constraint}
\end{equation}
The external input of the interconnected system is $u:=u_1$, while the external output is $y:=y_2$.

As before, for each system $\Sigma_i$, assume that a noise-free library of length-$T$
input-output trajectories is available in the form
\begin{equation}
H_W^{(i)}=
\begin{bmatrix}
H_u^{(i)}\\
H_y^{(i)}
\end{bmatrix}
\in\mathbb{R}^{2T\times N_i},
\quad i=1,2,
\label{eq:series_atomic_library}
\end{equation}
with $H_u^{(i)},H_y^{(i)}\in\mathbb{R}^{T\times N_i}$. 

\begin{lemma}[Series compatibility of atomic libraries]
\label{lem:library-from-atomic-trajectories-series}
Set
\[
\mathcal Q_s :=
\begin{bmatrix}
H_y^{(1)} & -H_u^{(2)}
\end{bmatrix}.
\]
Let $\Theta \in\mathbb{R}^{(N_1+N_2)\times \nu_s}$ be such that
\begin{equation} 
\mathrm{Im}(\Theta)=\ker(\mathcal Q_s),
\label{eq:series_nullspace}
\end{equation}
and consider the partition
\begin{equation}
\Theta=
\begin{bmatrix}
\Theta_1\\
\Theta_2
\end{bmatrix},
\quad
\Theta_1\in\mathbb{R}^{N_1\times \nu_s},\,\,
\Theta_2\in\mathbb{R}^{N_2\times \nu_s}.
\label{eq:series_nullspace_partition}
\end{equation}
Then
\begin{equation}
\left\{
w  \in \im(H_W^{(2)}):\ P_{\rm in} w \in \im(H_y^{(1)})
\right\}
=
\im(H_W^{(2)}\Theta_2)
\label{eq:series_main}
\end{equation}
where
$
P_{\rm in}:=\begin{bmatrix}I_T & 0\end{bmatrix}.
$
Hence, the set of output trajectories of the downstream system $\Sigma_2$ that admit a
series-compatible completion with some trajectory of $\Sigma_1$ is
\[
\im(H_y^{(2)}\Theta_2).
\]
\end{lemma}

\begin{proof}
$(\supseteq)$ Let $w\in \im(H_W^{(2)}\Theta_2)$. Then there exists $\alpha\in\mathbb R^{\nu_s}$ such that
$
w=H_W^{(2)}\Theta_2\alpha.
$
Moreover,
\[
P_{\rm in} w = H_u^{(2)}\Theta_2\alpha = H_y^{(1)}\Theta_1\alpha \in \im(H_y^{(1)}).
\]
Here, we used the identity $H_u^{(2)}\Theta_2=H_y^{(1)}\Theta_1$ which follows directly from $\mathcal Q_s\Theta=0$.
This shows the first inclusion.

$(\subseteq)$  Let
\[
w\in
\left\{
w   \in \im(H_W^{(2)}):\ P_{\rm in} w \in \im(H_y^{(1)})
\right\}. 
\]
Then, there exists $g_2\in\mathbb R^{N_2}$ such that
$
w=H_W^{(2)}g_2.
$
Further, $P_{\rm in}w\in\im(H_y^{(1)})$, so there exists $g_1\in\mathbb R^{N_1}$ such that
\[
H_u^{(2)}g_2 = H_y^{(1)}g_1.
\]
Therefore
$\left[
\begin{smallmatrix}
g_1\\
g_2
\end{smallmatrix}\right]
\in \ker(\mathcal Q_s)=\im(\Theta),
$
so there exists $\alpha\in\mathbb R^{\nu_s}$ such that
$\left[
\begin{smallmatrix}
g_1\\
g_2
\end{smallmatrix}\right]=
\Theta\alpha$.
In particular,
$g_2=\Theta_2\alpha$, 
and therefore
$w=H_W^{(2)}g_2=H_W^{(2)}\Theta_2\alpha\in \im(H_W^{(2)}\Theta_2)$. 
This shows the converse inclusion.
\end{proof}

Lemma \ref{lem:library-from-atomic-trajectories-series} identifies the downstream
input-output trajectories consistent with some upstream trajectory of
$\Sigma_1$ through the series constraint $u_2=y_1$. In particular, every
vector $\alpha\in\mathbb R^{\nu_s}$ defines a compatible pair of coefficient
vectors
$g_1=\Theta_1\alpha$ and $g_2=\Theta_2\alpha$,
for which
$
H_y^{(1)}g_1 = H_u^{(2)}g_2.
$
Since the external input of the series interconnection is the input of
$\Sigma_1$, while the external output is the output of $\Sigma_2$, this leads
naturally to the following input-output library for the interconnected system:
\[
\mathcal H_W^{\,s}:=
\begin{bmatrix}
H_u^{(1)}\Theta_1\\
H_y^{(2)}\Theta_2
\end{bmatrix}.
\]

At this point, we note that this library cannot, however, be used within the autonomous
framework developed so far, since it describes a forced behavior. Our case is obtained when $\Sigma_1$ is itself autonomous,
so that
$
H_u^{(1)}=0.
$
In this case, the signal driving the downstream system $\Sigma_2$ is generated internally by
$\Sigma_1$, and the resulting series interconnection defines an autonomous system. We therefore
define the corresponding output library as
\begin{equation}\label{H_Y_series}
\mathcal H_Y^{\,s}:=H_y^{(2)}\Theta_2.
\end{equation}

\begin{remark}[Series-interconnected library with $\Sigma_1$ autonomous]
What we just stated above, that is, $\mathcal H_Y^{\,s}=H_y^{(2)}\Theta_2$, stems from characterizing the series-interconnected input-output library and then setting $H_u^{(1)}=0$. Alternatively, we could have started from characterizing the intersection 
$\im(P_{\rm out} H_W^{(1)}) \cap \im(P_{\rm in}H_W^{(2)})$, where $P_{\rm out}:=\begin{bmatrix}0 & I_T\end{bmatrix}$, which is the subspace of output trajectories of $\Sigma_1$ built from $\im(H_W^{(1)})$ and input trajectories of $\Sigma_2$ built from $\im(H_W^{(2)})$ that satisfy the series connection constraint \eqref{eq:series_constraint}. An equivalent expression of this subspace is $\im(H_u^{(2)}\Theta_2)$. By linearity, the output trajectories of the series interconnection are obtainable from $\im(H_u^{(2)}\Theta_2)$ by considering $\im(H_y^{(2)}\Theta_2)$. This underscores how deriving the series-interconnected output library only requires $H_y^{(1)}$ and $H_W^{(2)}$, that is, allows for $\Sigma_1$ to be autonomous. 
\end{remark}

The following lemma provides a state-space factorization of the resulting output library.

\begin{lemma}[Factorization of the series-interconnected library]
\label{lem:series_factorization}
Let $\mathcal{Q}_s$ and $\Theta$ be defined as in Lemma 
\ref{lem:library-from-atomic-trajectories-series}.
Let
\begin{equation}
\mathcal A_s :=
\begin{bmatrix}
A_1 & 0\\
B_2 C_1 & A_2
\end{bmatrix},
\quad
\mathcal C_s :=
\begin{bmatrix}
0 & C_2
\end{bmatrix}.
\label{eq:series_cl_matrices}
\end{equation}
Then there exist matrices $X_1\in\mathbb R^{n_1\times N_1}$ and
$X_2\in\mathbb R^{n_2\times N_2}$ such that, defining
\[
\mathcal X_s :=
\begin{bmatrix}
X_1\Theta_1\\
X_2\Theta_2
\end{bmatrix},
\]
the output library of the autonomous series interconnection satisfies
\begin{equation}
\mathcal H_Y^{\,s}
=
O_T(\mathcal A_s,\mathcal C_s) \mathcal X_s,
\label{eq:series_factorization}
\end{equation}
where $\mathcal H_Y^{\,s}=H_y^{(2)}\Theta_2$.
\end{lemma}

\begin{proof}
Bearing in mind \eqref{eq:series_cl_matrices}, first note that we are considering the series connection of the {\em autonomous} upstream system $\Sigma_1$ and the downstream system $\Sigma_2$. 
Moreover, as in Lemma \ref{lemma:fb_factorization}, for $i=1,2$, there exist matrices $X_i\in\mathbb R^{n_i\times N_i}$ of initial conditions of $\Sigma_i$ associated to the trajectories of  $H_y^{(1)}$ and $H_W^{(2)}$. 

By linearity of $\Sigma_i$, for every coefficient vector $g_i\in\mathbb R^{N_i}$, the trajectories 
$H_y^{(1)}g_1$ and $H_W^{(2)}g_2$ are valid length-$T$ output and input-output trajectories of $\Sigma_1$ and  $\Sigma_2$ with initial conditions $X_1 g_1$ and $X_2 g_2$.
Now let $\alpha\in\mathbb R^{\nu_s}$ be arbitrary and define $g_1=\Theta_1\alpha$, $g_2=\Theta_2\alpha$.
Since $\im(\Theta)=\ker(\mathcal Q_s)$, the pair $(g_1,g_2)$ satisfies the 
constraint $H_y^{(1)}g_1 = H_u^{(2)}g_2$.
Therefore, $H_y^{(2)}g_2$ is a valid output trajectory of the series interconnection \eqref{eq:series_cl_matrices} stemming from the initial condition $\left[\begin{smallmatrix} X_1 \Theta_1 \\ X_2 \Theta_2\end{smallmatrix}\right]\alpha=
\mathcal{X}_s\alpha$. 
Since
$H_y^{(2)}g_2=H_y^{(2)}\Theta_2 \alpha$, 
we have
$H_y^{(2)}\Theta_2 \alpha
=
O_T(\mathcal A_s,\mathcal C_s)\mathcal X_s\alpha$. 
As this holds for every $\alpha\in\mathbb R^{\nu_s}$, we conclude that
$\mathcal H_Y^{\,s}
=
H_y^{(2)}\Theta_2
=
O_T(\mathcal A_s,\mathcal C_s)\mathcal X_s$. 
\end{proof}

We close this section by discussing the case of parallel interconnection.

\subsubsection{Parallel-interconnected library}

Let $\Sigma_1$ and $\Sigma_2$ be two SISO strictly causal LTI systems with input-output pairs
$(u_1,y_1)$ and $(u_2,y_2)$ and state-space realizations $(A_i,B_i,C_i)$, $i=1,2$.
We consider the classical parallel interconnection in which the two systems are driven by the same
external input and their outputs are summed, namely
\begin{equation}
u_1=u_2=:u,
\quad
y:=y_1+y_2.
\label{eq:parallel_constraint}
\end{equation}

Fix a horizon $T$ and, for each system $\Sigma_i$, assume that a noise-free library of length-$T$
input-output trajectories is available in the form
\begin{equation}
H_W^{(i)}=
\begin{bmatrix}
H_u^{(i)}\\
H_y^{(i)}
\end{bmatrix}
\in\mathbb{R}^{2T\times N_i},
\quad i=1,2,
\label{eq:parallel_atomic_library}
\end{equation}
with $H_u^{(i)},H_y^{(i)}\in\mathbb{R}^{T\times N_i}$.

The following lemma characterizes the 
space of trajectories obtained from $\im(H_W^{(1)})$ and $\im(H_W^{(2)})$ whose input components satisfy the first parallel connection constraint \eqref{eq:parallel_constraint}. The output trajectories of the parallel connection are obtained from this space through the output libraries $H_y^{(1)}$, $H_y^{(2)}$, linearity and the second constraint in  \eqref{eq:parallel_constraint}. 

\begin{lemma}[Parallel compatibility of atomic libraries]
Set
\[
\mathcal Q_p :=
\begin{bmatrix}
H_u^{(1)} & -H_u^{(2)}
\end{bmatrix}.
\]
Let $\Theta \in\mathbb{R}^{(N_1+N_2)\times \nu_p}$ be such that
\begin{equation}
\im(\Theta)=\ker(\mathcal Q_p),
\label{eq:parallel_nullspace}
\end{equation}
and partition
\begin{equation}
\Theta=
\begin{bmatrix}
\Theta_1\\
\Theta_2
\end{bmatrix},
\quad
\Theta_1\in\mathbb{R}^{N_1\times \nu_p},\,\,
\Theta_2\in\mathbb{R}^{N_2\times \nu_p}.
\label{eq:parallel_nullspace_partition}
\end{equation}
Then
\begin{equation}
\left\{
w\in \im(H_W^{(2)}):\ P_{\rm in} w \in \im(H_u^{(1)})
\right\}=
\im(H_W^{(2)}\Theta_2),
\label{eq:parallel_main}
\end{equation}
and 
\begin{equation}
\left\{
w\in \im(H_W^{(1}):\ P_{\rm in} w \in \im(H_u^{(2)})
\right\}=\im(H_W^{(1)}\Theta_1),
\label{eq:parallel_main_2}
\end{equation}
where
$
P_{\rm in}:=\begin{bmatrix}I_T & 0\end{bmatrix}.
$
Therefore, summing the outputs of all parallel-compatible pairs yields the set of output trajectories of the interconnected system:
\[
\im(H_y^{(1)}\Theta_1 + H_y^{(2)}\Theta_2).
\]
\end{lemma}

\begin{proof}
$(\supseteq)$ Let $w\in \im(H_W^{(2)}\Theta_2)$. Then there exists $\alpha\in\mathbb R^{\nu_p}$ such that
$
w=H_W^{(2)}\Theta_2\alpha.
$
Moreover,
\[
P_{\rm in} w = H_u^{(2)}\Theta_2\alpha = H_u^{(1)}\Theta_1\alpha \in \im(H_u^{(1)}),
\]
where we used the identity $H_u^{(2)}\Theta_2=H_u^{(1)}\Theta_1$, which follows directly from
$\mathcal Q_p\Theta=0$.
This proves the first inclusion.

$(\subseteq)$ Let
\[
w\in
\left\{
w\in \im(H_W^{(2)}):\ P_{\rm in} w \in \im(H_u^{(1)})
\right\}.
\]
Then, there exists $g_2\in\mathbb R^{N_2}$ such that
$
w=H_W^{(2)}g_2.
$
Moreover,
$
P_{\rm in}w=H_u^{(2)}g_2\in\im(H_u^{(1)}),
$
so there exists \(g_1\in\mathbb R^{N_1}\) such that
$
H_u^{(2)}g_2 = H_u^{(1)}g_1.
$
Hence
\[
\mathcal Q_p
\begin{bmatrix}
g_1\\
g_2
\end{bmatrix}
=0.
\]
Therefore
\[
\begin{bmatrix}
g_1\\
g_2
\end{bmatrix}
\in \ker(\mathcal Q_p)=\im(\Theta),
\]
so there exists \(\alpha\in\mathbb R^{\nu_p}\) such that
\[
\begin{bmatrix}
g_1\\
g_2
\end{bmatrix}
=
\Theta\alpha.
\]
In particular,
$
g_2=\Theta_2\alpha,
$
and therefore
\[
w=H_W^{(2)}g_2=H_W^{(2)}\Theta_2\alpha\in \im(H_W^{(2)}\Theta_2).
\]
This proves \eqref{eq:parallel_main}. The identity \eqref{eq:parallel_main_2} is proven similarly.

Finally, for every compatible pair
\[
g_1=\Theta_1\alpha,\quad g_2=\Theta_2\alpha,
\]
the corresponding output of the parallel interconnection is
\[
y = H_y^{(1)}g_1 + H_y^{(2)}g_2
   = H_y^{(1)}\Theta_1\alpha + H_y^{(2)}\Theta_2\alpha.
\]
Hence, the set of output trajectories generated by all parallel-compatible pairs is precisely as claimed. 
\end{proof}

We therefore define the parallel-interconnected input-output library as
\[
\mathcal H_W^{\,p}:=
\begin{bmatrix}
H_u^{(1)}\Theta_1\\
H_y^{(1)}\Theta_1 + H_y^{(2)}\Theta_2
\end{bmatrix}.
\]
As in the series case, this library cannot, in general, be used directly within the autonomous
prediction framework developed so far, since it describes a forced behavior.

We therefore restrict attention to the case in which both systems
$\Sigma_1$ and $\Sigma_2$ are autonomous, so that
\[
H_u^{(1)}=0,
\quad
H_u^{(2)}=0.
\]
In this case, the compatibility condition becomes trivial. Indeed,
$
\mathcal Q_p=
\begin{bmatrix}
H_u^{(1)} & -H_u^{(2)}
\end{bmatrix}
=0,
$
so that every pair of trajectories is compatible. Equivalently, one may take
\[
\Theta = I_{N_1+N_2},
\quad
\Theta_1=\begin{bmatrix}I_{N_1} & 0\end{bmatrix},\,\,
\Theta_2=\begin{bmatrix}0 & I_{N_2}\end{bmatrix}.
\]
Accordingly, the output library of the autonomous parallel interconnection is
\begin{equation}\label{H_Y_parallel}
\mathcal H_Y^{\,p}:=H_y^{(1)}\Theta_1 + H_y^{(2)}\Theta_2
=
\begin{bmatrix}
H_y^{(1)} & H_y^{(2)}
\end{bmatrix}.
\end{equation}

\begin{remark}[Parallel-interconnected library with $\Sigma_1, \Sigma_2$ autonomous]
If one is dealing with autonomous systems $\Sigma_1$, $\Sigma_2$ and only output libraries $H_y^{(1)}$, $H_y^{(2)}$ are available (as in the case of the aggregate system of Section \ref{subsec:mix-agg}), the space of output trajectories obtainable from those would be directly provided by $\im(H_y^{(1)})+\im(H_y^{(2)})$, which corresponds to the parallel-interconnected output library given by \eqref{H_Y_parallel}.  
\end{remark}

The following lemma provides a state-space factorization of this output library.  

\begin{lemma}[Factorization of the parallel-interconnected library]
\label{lem:parallel_factorization_autonomous}
Let
\begin{equation}
\mathcal A_p :=
\begin{bmatrix}
A_1 & 0\\
0 & A_2
\end{bmatrix},
\quad
\mathcal C_p :=
\begin{bmatrix}
C_1 & C_2
\end{bmatrix}.
\label{eq:parallel_cl_matrices}
\end{equation}
Then there exist matrices $X_1\in\mathbb R^{n_1\times N_1}$ and
$X_2\in\mathbb R^{n_2\times N_2}$ such that, defining
\[
\mathcal X_p :=
\begin{bmatrix}
X_1 & 0\\
0 & X_2
\end{bmatrix},
\]
the output library of the autonomous parallel interconnection satisfies
\begin{equation}
\mathcal H_Y^{\,p}
=
O_T(\mathcal A_p,\mathcal C_p)\mathcal X_p.
\label{eq:parallel_factorization_autonomous}
\end{equation}
\end{lemma}

\begin{proof}
Since each column of $H_W^{(i)}$ is a length-$T$ input-output trajectory of $\Sigma_i$, there
exists, for each column of $H_W^{(i)}$, an associated initial condition of $\Sigma_i$.
Collecting these initial conditions columnwise defines a matrix
$X_i\in\mathbb R^{n_i\times N_i}$, $i=1,2$.

Because both systems are autonomous, their input blocks vanish, and the corresponding output
libraries satisfy
\[
H_y^{(1)} = O_T(A_1,C_1)X_1,
\quad
H_y^{(2)} = O_T(A_2,C_2)X_2.
\]
Therefore,
\[
\mathcal H_Y^{\,p}
=
\begin{bmatrix}
H_y^{(1)} & H_y^{(2)}
\end{bmatrix}
=
\begin{bmatrix}
O_T(A_1,C_1)X_1 & O_T(A_2,C_2)X_2
\end{bmatrix}.
\]
Since
\[
O_T(\mathcal A_p,\mathcal C_p)
=
\begin{bmatrix}
O_T(A_1,C_1) & O_T(A_2,C_2)
\end{bmatrix},
\]
it follows that
\[
\mathcal H_Y^{\,p}
=
O_T(\mathcal A_p,\mathcal C_p)
\begin{bmatrix}
X_1 & 0\\
0 & X_2
\end{bmatrix}
=
O_T(\mathcal A_p,\mathcal C_p)\mathcal X_p.
\]
This proves \eqref{eq:parallel_factorization_autonomous}.
\end{proof}

It is worth emphasizing that, in the autonomous case, the data-richness condition for the
parallel-interconnected library coincides with that of the composed library considered in the
previous sections: richness is guaranteed as soon as the individual state matrices $X_1$ and $X_2$
have full row rank. The reason is that, once the common input is identically zero, the compatibility
constraint of the parallel interconnection becomes trivial and no additional restriction is imposed
on the admissible pairs of trajectories. Nonetheless, the parallel-interconnected library is conceptually
different from the mixed library  previously introduced. In the mixed-library setting, one
simply aggregates autonomous trajectories coming from different prototype systems, whereas in the
parallel case one forms a new autonomous behavior whose output is the sum of the outputs of the two
subsystems. Thus, although the rank condition for richness is the same, the mechanism generating new
behaviors is different. 

Using the previous factorization results, we can state a unified continuation and prediction result
covering all the interconnection mechanisms discussed above.

\begin{proposition}[Continuation identity for an interconnection-generated library]
\label{prop:continuation_interconnection_library}
Let $\mathcal H_Y\in\mathbb R^{T\times \nu}$ be an output library generated by one of the
interconnection constructions above (feedback, series, or parallel), and let it admit
a
factorization
\[
\mathcal H_Y = O_T(\mathcal A,\mathcal C) \mathcal X
\]
for some autonomous SISO LTI realization $(\mathcal A,\mathcal C)$.
Partition $\mathcal H_Y$ as
\begin{equation}
\mathcal H_Y=
\begin{bmatrix}
\mathcal H_p\\
\mathcal H_f
\end{bmatrix},
\quad
\mathcal H_p\in\mathbb R^{r\times \nu},\,\,
\mathcal H_f\in\mathbb R^{(T-r)\times \nu}.
\label{eq:Hp_Hf_interconnection}
\end{equation}
Assume that
$
r\ge s_{\mathrm{int}},
$
where $s_{\mathrm{int}}$ denotes the observability index of $(\mathcal A,\mathcal C)$.
Then there exists a matrix
\[
\mathcal L\in\mathbb R^{(T-r)\times r}
\]
such that
\begin{equation}
\mathcal H_f=\mathcal L\mathcal H_p.
\label{eq:continuation_identity_interconnection}
\end{equation}
\end{proposition}

\begin{proof}
Since
$
\mathcal H_Y = O_T(\mathcal A,\mathcal C)\mathcal X,
$
partitioning $\mathcal H_Y$ according to \eqref{eq:Hp_Hf_interconnection} yields
\[
\mathcal H_p = O_r(\mathcal A,\mathcal C)\,\mathcal X,
\quad
\mathcal H_f = \mathcal M\,\mathcal X,
\]
where $\mathcal M$ collects the output blocks from time $r$ to time $T-1$.
Since $r\ge s_{\mathrm{int}}$, Lemma~\ref{lem:key} applies to the pair $(\mathcal A,\mathcal C)$ and
gives a matrix $\mathcal L$ such that
$
\mathcal H_f=\mathcal L\mathcal H_p.
$
\end{proof}

\begin{theorem}[Exact prediction from an interconnection-generated library]
\label{thm:exact_prediction_interconnection_library}
Assume that $r\ge s_{\mathrm{int}}$, where $s_{\mathrm{int}}$ denotes the observability index of
$(\mathcal A,\mathcal C)$, so that
$
\mathcal H_f=\mathcal L\mathcal H_p.
$
Let
\[
Y^\star=
\begin{bmatrix}
Y_{\mathrm{past}}\\
Y_{\mathrm{fut}}
\end{bmatrix}\in\mathbb R^T
\]
be a length-$T$ trajectory of a (possibly unknown) system. Suppose that:
\begin{enumerate}
\item[(i)] (\emph{Feasibility}) $Y_{\mathrm{past}}\in \im(\mathcal H_p)$;
\item[(ii)] (\emph{Shared continuation}) $Y_{\mathrm{fut}}=\mathcal L Y_{\mathrm{past}}$.
\end{enumerate}
Then, for any $g$ satisfying
\begin{equation}
Y_{\mathrm{past}}=\mathcal H_p g,
\label{eq:feasible_g_interconnection}
\end{equation}
the prediction
\begin{equation}
\widehat Y_{\mathrm{fut}}:=\mathcal H_f g
\label{eq:prediction_interconnection_library}
\end{equation}
is exact:
\[
\widehat Y_{\mathrm{fut}}=Y_{\mathrm{fut}}.
\]
In particular, the predicted future is independent of the chosen feasible coefficient vector $g$.
\end{theorem}

\begin{proof} The proof is the same as the one of Theorem \ref{thm:mix-ool-exact}.
\end{proof}

\begin{proposition}[Exact prediction from an interconnection-generated library under spectral inclusion]
\label{prop:interconnection_spectral}
Let $\mathcal H_Y$ be an output library generated by one of the interconnection constructions above
(feedback, series, or parallel), and 
let it admit a factorization
$
\mathcal H_Y = O_T(\mathcal A,\mathcal C)\,\mathcal X
$
for some autonomous SISO LTI realization $(\mathcal A,\mathcal C)$ of order $n_{\mathrm{int}}$.
Assume
$\rank(\mathcal X)=n_{\mathrm{int}}$ and
$r\ge s_{\mathrm{int}}$,
where $s_{\mathrm{int}}$ denotes the observability index of $(\mathcal A,\mathcal C)$.
Assume moreover that $\mathcal A$ is diagonalizable.
Let $A_{\mathrm{new}}$ be a diagonalizable matrix of the SISO system $(A_{\mathrm{new}},C_{\mathrm{new}})$ 
that satifies
\begin{equation}
\spec_{ C_{\mathrm{new}} }(A_{\mathrm{new}})\subseteq \spec_{\mathcal C}(\mathcal A).
\label{eq:generic_interconnection_spectral_condition}
\end{equation}
Then every length-$T$ trajectory $Y^\star$  of $(A_{\mathrm{new}},C_{\mathrm{new}})$ satisfies the assumptions of
Theorem~\ref{thm:exact_prediction_interconnection_library}. Consequently, the predictor induced by
the interconnection-generated library $\mathcal H_Y$ is exact.
\end{proposition}

\begin{proof}
The proof follows a similar pattern as the one of Proposition \ref{prop:mix-sentinels}. 
\end{proof}

\begin{remark}[Data richness and diagonalizability]
Proposition~\ref{prop:interconnection_spectral} applies uniformly to the three cases discussed
above. In each case, the only ingredients needed are: (i) an autonomous output library
$\mathcal H_Y$, (ii) a state-space factorization
$\mathcal H_Y=O_T(\mathcal A,\mathcal C)\mathcal X,
$and (iii) the full row rank of $\mathcal X$. 
As explained before, the output library $\mathcal H_Y$ is obtainable from the atomic libraries (see \eqref{H_Y_fb}, \eqref{H_Y_series}, \eqref{H_Y_parallel}). The factorization always holds as established in Lemma \ref{lemma:fb_factorization}, \ref{lem:series_factorization} and \ref{lem:parallel_factorization_autonomous}. On the other hand, the full row rank property of $\mathcal X$ does not automatically derive from the data richness of the single libraries (except for the parallel interconnection). 
The specific interconnection mechanism only affects how
the library is constructed and how the matrices $(\mathcal A,\mathcal C,\mathcal X)$ are obtained.

A further difference with respect to the compositional setting considered earlier concerns the
diagonalizability assumption. In the compositional and parallel case, diagonalizability follows directly from the
fact that the prototype systems are diagonal. By contrast, for feedback and series
interconnections, diagonalizability of the resulting realization $(\mathcal A,\mathcal C)$ is not
automatic, even if the individual prototype systems are diagonalizable. Indeed, interconnection may
modify the spectral structure of the aggregate dynamics. Nevertheless, when the prototype systems are synthetically generated, diagonalizability of the
interconnected realization can still be checked explicitly and, if needed, enforced by design.
\end{remark}

Interconnecting libraries is useful both conceptually and practically. From a
conceptual viewpoint, it shows that libraries are not merely static collections
of sample trajectories, but compositional objects that can be combined to
generate new behaviors. From a practical viewpoint, it makes it possible to
construct richer predictors directly from atomic data, without collecting new
data from the interconnected systems themselves, which may be costly or even
impractical. In particular, feedback interconnection can enlarge the spectral
set of the original systems, thereby enabling prediction of trajectories
generated by systems outside the initial family. In this sense, feedback-based
library construction provides a genuine form of \emph{generalization from templates}.

\begin{example}[Generalization from templates] \label{exmpl:feedback-generalization}
We illustrate Lemma~\ref{lem:library-from-atomic-trajectories-fb} and the subsequent out-of-library prediction mechanism
with a simple scalar example. Consider the two prototype systems
\[
\Sigma_1:\quad x_{1,t+1}=0.2\,x_{1,t}+u_{1,t}, \quad y_{1,t}=x_{1,t},
\]
\[
\Sigma_2:\quad x_{2,t+1}=0.3\,x_{2,t}+u_{2,t}, \quad y_{2,t}=x_{2,t}.
\]
For each system, we generate a rich library of $30$ input-output
trajectories of length $T=8$,
\[
H_W^{(i)}=
\begin{bmatrix}
H_u^{(i)}\\[1mm]
H_y^{(i)}
\end{bmatrix} \in \mathbb R^{16 \times 30},
\quad i=1,2,
\]
by sampling multiple initial conditions and input sequences from the normal distribution. We then impose
the feedback interconnection constraints
\[
u_1=y_2,\quad u_2=y_1,
\]
and define the matrices $\mathcal Q \in \mathbb R^{16 \times 60}$ and 
$\Theta \in \mathbb R^{60 \times 44}$ as in 
Lemma~\ref{lem:library-from-atomic-trajectories-fb}.\footnote{The fact that $\Theta$ has $44$ 
columns means that the feedback constraints
leave a $44$-dimensional space of admissible coefficient pairs $(g_1,g_2)$,
that is, $44$ independent ways of combining the two atomic libraries while
respecting the interconnection law.}
Since $\mathrm{Im}(\Theta)=\ker(\mathcal Q)$,
according to Lemma~\ref{lem:library-from-atomic-trajectories-fb}, the output library generated by the feedback
interconnection is
$
\mathcal H_Y=H_y^{(1)}\Theta_1.
$
By Lemma~\ref{lemma:fb_factorization}, $\mathcal H_Y$ admits the
factorization
$
\mathcal H_Y = O_T(\mathcal A, \mathcal C) \mathcal X,
$
where
\[
\mathcal A =
\begin{bmatrix}
0.2 & 1\\
1 & 0.3
\end{bmatrix},
\quad
\mathcal C =
\begin{bmatrix}
1 & 0
\end{bmatrix}.
\]

The eigenvalues of $\mathcal A$ are
\[
\lambda_{1,2}
=
\frac{0.5\pm\sqrt{4.01}}{2}
\]
which are different from the eigenvalues $0.2$ and $0.3$ of the individual
prototype systems. The matrix $\mathcal X$ has full row rank.
By Proposition~\ref{prop:interconnection_spectral}, the feedback-generated library thus
supports exact prediction for every diagonalizable LTI system whose visible spectrum is contained in 
$\spec_{ \mathcal C }( \mathcal A ) = \{\lambda_{1},\lambda_{2}\}$. 
To highlight this fact, we consider the new system
\[
x_{t+1}=\lambda_1 x_t,\quad y_t=x_t,
\]
and generate a trajectory $Y^*$ of length $8$. We choose $r=3$, which corresponds to predicting $5$ 
samples of the trajectory by observing $3$ of them; we then solve
$
Y_{\rm past}=H_p g
$
and predict the future via
$
\widehat Y_{\rm fut}=H_f g.
$

Figure \ref{fig:OoL-prediction} confirms that this method is able to predict exactly
the trajectory even though its natural modes are not present in either atomic
system. In particular, $\|Y_{\rm fut}-\widehat Y_{\rm fut}\|_2\approx 4\cdot 10^{-15}$ at machine precision.

\begin{figure}[t]
    \centering
    \includegraphics[width=0.8\linewidth]{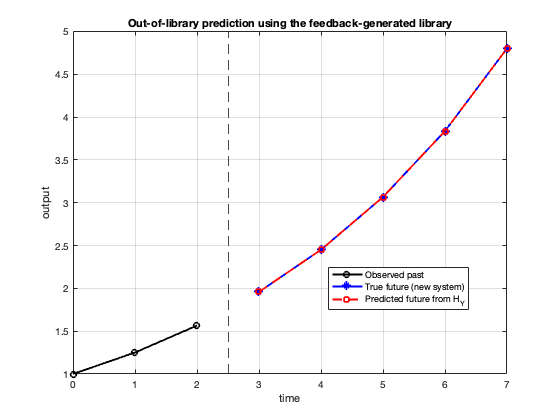}
    \caption{Trajectory prediction by observing $r=3$ samples (black dots). 
     The feedback-generated library is able to predict exactly
the trajectory even though its natural modes are not present in either atomic
system.}
    \label{fig:OoL-prediction}
\end{figure}

This example shows that our method may be useful beyond pure prediction, for instance to gain insight into whether a given interconnection is likely to be stable directly from atomic libraries, without physically interconnecting the systems.
\end{example}

\section{Output inclusion, immersion, and the 
prediction of nonlinear trajectories from linear templates}\label{sec:output-containment-immersion}

Predicting trajectories generated by nonlinear dynamical systems is, in general, substantially more difficult than in the linear setting. Nonlinear systems may not admit a finite-dimensional linear realization in the original state variables. However, in some cases, the evolution of selected observables evolves according to a finite-dimensional linear law, a perspective closely related to the Koopman operator framework.

In this section, we show that our approach extends to this class of nonlinear systems. We focus on systems whose output observations generate a finite-dimensional observation space. Such systems can be immersed into a finite-dimensional linear system, and, under suitable conditions, their output trajectories are contained in the library generated by the linear template system, possibly obtained as an aggregate or interconnected composition of linear systems.

\subsection{Spectral inclusion as an output containment property}

We start with a pair of reformulations of the spectral condition used in Proposition \ref{prop:mix-sentinels}. 
The purpose of these results is to make explicit the following viewpoint: prediction from a library of linear templates is possible when the new system -- whose output we want to predict -- is immersed into the aggregate template system. This property will be used later on to derive an extension of our framework to nonlinear systems with finite-dimensional observation spaces.

Let $(\mathcal A,\mathcal C)$ denote the linear template system associated with the library used for prediction, and let
$
(A_{\mathrm{new}},C_{\mathrm{new}})
$
be a new linear system.  
Let $n$ and $\nu$ denote the state-space dimensions of $(\mathcal A,\mathcal C)$ and $(A_{\mathrm{new}},C_{\mathrm{new}})$, respectively, and assume that the two systems have the same output dimension $p$.

The following result holds.

\begin{lemma}[A reformulation of the spectral condition] \label{lem:an.spectr.cond}
Let $\mathcal A$ and $A_{\mathrm{new}}$ be diagonalizable. Also assume that, for every 
$\mu\in\spec(A_{\mathrm{new}})$,
\begin{equation}\label{eq:complex-visible-inclusion}
C_{\mathrm{new}}\bigl(\ker(\mu I-A_{\mathrm{new}})\bigr)
\subseteq
\mathcal C\bigl(\ker(\mu I-\mathcal A)\bigr).
\end{equation}
Then
\begin{equation}\label{eq:spectral-inclusion}
\spec_{C_{\mathrm{new}}}(A_{\mathrm{new}})
\subseteq
\spec_{\mathcal C}(\mathcal A).
\end{equation}
If $p=1$, then the converse also holds.
\end{lemma}

\begin{proof}
If $C_{\mathrm{new}}= 0$ then the result is trivial. We thus assume $C_{\mathrm{new}} \neq 0$.
Assume that \eqref{eq:complex-visible-inclusion} holds for every $\mu\in\spec(A_{\mathrm{new}})$. Let $\mu\in\spec_{C_{\mathrm{new}}}(A_{\mathrm{new}})$. Then $\mu\in\spec(A_{\mathrm{new}})$, and there exists
$
v_{\mathrm{new}}\in\ker(\mu I-A_{\mathrm{new}})
$
such that $C_{\mathrm{new}}v_{\mathrm{new}}\neq 0$. By \eqref{eq:complex-visible-inclusion},
\[
C_{\mathrm{new}}v_{\mathrm{new}}
\in
\mathcal C\bigl(\ker(\mu I-\mathcal A)\bigr).
\]
Hence, there exists $v\in\ker(\mu I-\mathcal A)$ such that
$
\mathcal Cv=C_{\mathrm{new}}v_{\mathrm{new}}.
$
Since $C_{\mathrm{new}}v_{\mathrm{new}}\neq 0$, we have $\mathcal Cv\neq 0$. Thus $\mu$ is a $\mathcal C$-visible eigenvalue of $\mathcal A$, namely
$
\mu\in\spec_{\mathcal C}(\mathcal A).
$
This shows the first part of the statement. 

Conversely, let $p=1$ and
$
\spec_{C_{\mathrm{new}}}(A_{\mathrm{new}})
\subseteq
\spec_{\mathcal C}(\mathcal A).
$
Take $\mu\in\spec(A_{\mathrm{new}})$. If $\mu\notin\spec_{C_{\mathrm{new}}}(A_{\mathrm{new}})$, then
\[
C_{\mathrm{new}}\bigl(\ker(\mu I-A_{\mathrm{new}})\bigr)=\{0\},
\]
and \eqref{eq:complex-visible-inclusion} is immediate.
Suppose instead $\mu\in\spec_{C_{\mathrm{new}}}(A_{\mathrm{new}})$. 
By assumption, $\mu\in\spec_{\mathcal C}(\mathcal A)$, and thus
\[
\mathcal C\bigl(\ker(\mu I-\mathcal A)\bigr)\neq\{0\}.
\]
Since $p=1$, the set $\mathcal C(\ker(\mu I-\mathcal A))$ is a nonzero complex subspace of $\mathbb C$. Therefore,
$
\mathcal C\bigl(\ker(\mu I-\mathcal A)\bigr)=\mathbb C.
$
On the other hand,
$
C_{\mathrm{new}}\bigl(\ker(\mu I-A_{\mathrm{new}})\bigr)
\subseteq \mathbb C,
$
and hence \eqref{eq:complex-visible-inclusion} follows. This proves the converse implication.
\end{proof}

Lemma \ref{lem:an.spectr.cond} allows us to restate the condition on the spectral inclusion as the existence of a linear immersion map.

\begin{lemma}[Linear immersion map]\label{lem:spect.cond.for.output.inclusion}
Let $\mathcal A\in\mathbb R^{n\times n}$ and $A_{\mathrm{new}}\in\mathbb R^{\nu \times \nu}$ be diagonalizable. Suppose that, for every $\mu\in\spec(A_{\mathrm{new}})$, condition \eqref{eq:complex-visible-inclusion} holds.
Then there exists a map $T_o:\mathbb R^{\nu}\to\mathbb R^n$ such that
\begin{equation}\label{imm.cond.again.cor}
C_{\mathrm{new}}=\mathcal C T_o,
\quad
T_oA_{\mathrm{new}}=\mathcal A T_o.
\end{equation}
Consequently, every output trajectory of $(A_{\mathrm{new}},C_{\mathrm{new}})$ is also an output trajectory of $(\mathcal A,\mathcal C)$, i.e.
\begin{equation}\label{ouput-containment}
\mathcal{Y}_{T,\mathrm{new}}\subseteq \mathcal{Y}_{T},
\end{equation}
where
\[
\begin{array}{rcl}
\mathcal{Y}_{T,\mathrm{new}} &=& \{\mathcal{O}_T(A_{\mathrm{new}},C_{\mathrm{new}})x_{0,\mathrm{new}} \colon x_{0,\mathrm{new}}  \in \mathbb{R}^\nu\},\\[0.1cm]
\mathcal{Y}_{T} &=& \{\mathcal{O}_T(\mathcal{A},\mathcal{C})x_{0} \colon x_{0} \in \mathbb{R}^{n}\}.
\end{array}
\]
\end{lemma}

\begin{proof}
Let $\{\mu_1,\ldots,\mu_m\}$ be the distinct eigenvalues of $A_{\mathrm{new}}$. For each $\ell=1,\ldots,m$, define
\[
\mathcal E^{\mathrm{new}}_\ell:=\ker(\mu_\ell I-A_{\mathrm{new}}),
\quad
\mathcal E_\ell:=\ker(\mu_\ell I-\mathcal A).
\]
Since $A_{\mathrm{new}}$ is diagonalizable 
\[
\mathbb C^{\nu}=\bigoplus_{\ell=1}^{m} \mathcal E^{\mathrm{new}}_\ell.
\]

Let $d_\ell:=\dim\mathcal E^{\mathrm{new}}_\ell$ and choose a basis
\[
e^{\mathrm{new}}_{\ell,1},\ldots,e^{\mathrm{new}}_{\ell,d_\ell}
\]
of $\mathcal E^{\mathrm{new}}_\ell$. 
By \eqref{eq:complex-visible-inclusion}, for every $e^{\mathrm{new}}_{\ell,j}$ 
there exists $w_{\ell,j}\in\mathcal E_\ell$ such that
\begin{equation}\label{output.identity}
\mathcal Cw_{\ell,j}=C_{\mathrm{new}}e^{\mathrm{new}}_{\ell,j}.
\end{equation}
Collect all these vectors in the matrices
\[
E_{\mathrm{new}}:=\begin{bmatrix}
 e^{\mathrm{new}}_{1,1} \,\, \cdots \,\, e^{\mathrm{new}}_{1,d_1} \,\, \cdots \,\, e^{\mathrm{new}}_{m,1} \,\, \cdots \,\, e^{\mathrm{new}}_{m,d_m}
\end{bmatrix}
\in\mathbb C^{\nu \times \nu},
\]
\[
W:=\begin{bmatrix}
 w_{1,1} \,\, \cdots \,\, w_{1,d_1} \,\, \cdots \,\, w_{m,1} \,\, \cdots \,\, w_{m,d_m}
\end{bmatrix}
\in\mathbb C^{n \times \nu}.
\]
The matrix $E_{\mathrm{new}}$ is invertible because its columns form a basis of $\mathbb C^{\nu}$. Define
\[
\mathcal T_o:=W(E_{\mathrm{new}})^{-1}.
\]
Notice that $\mathcal T_o$ is generally a complex-valued matrix in $\mathbb C^{n \times \nu}$.
By construction, \eqref{output.identity} gives
$
\mathcal CW=C_{\mathrm{new}}E_{\mathrm{new}},
$
and therefore,
$
\mathcal C\mathcal T_oE_{\mathrm{new}}=\mathcal CW=C_{\mathrm{new}}E_{\mathrm{new}}.
$
Since $E_{\mathrm{new}}$ is invertible, this yields
\[
\mathcal C\mathcal T_o=C_{\mathrm{new}}.
\]

We now consider the relation involving the state matrices.
Since each $e^{\mathrm{new}}_{\ell,j}$ is an eigenvector of $A_{\mathrm{new}}$ with eigenvalue $\mu_\ell$, we have
$
A_{\mathrm{new}}E_{\mathrm{new}}=E_{\mathrm{new}}\Lambda_{\mathrm{new}},
$
where
\[
\Lambda_{\mathrm{new}}=\diag(
\underbrace{\mu_1,\ldots,\mu_1}_{d_1},
\ldots,
\underbrace{\mu_m,\ldots,\mu_m}_{d_m})
\in\mathbb C^{ \nu \times \nu }.
\]
Similarly, since $w_{\ell,j}\in\mathcal E_\ell$, we have 
$
\mathcal A w_{\ell,j} =\mu_\ell w_{\ell,j}
$
for all $\ell=1,\ldots,m$ and $j=1,\ldots,d_\ell$. In compact form,
$
\mathcal AW=W\Lambda_{\mathrm{new}}.
$
Thus
\[
\mathcal T_oA_{\mathrm{new}}E_{\mathrm{new}}=\mathcal T_oE_{\mathrm{new}}\Lambda_{\mathrm{new}}=W\Lambda_{\mathrm{new}}.
\]
Moreover,
\[
\mathcal A \mathcal T_oE_{\mathrm{new}}=\mathcal AW=W\Lambda_{\mathrm{new}}.
\]
Hence,
$
\mathcal T_o A_{\mathrm{new}}E_{\mathrm{new}}=\mathcal A \mathcal T_oE_{\mathrm{new}}.
$
Since $E_{\mathrm{new}}$ is invertible, we conclude that
\[
\mathcal T_oA_{\mathrm{new}}=\mathcal A \mathcal T_o.
\]
Define $T_o := \text{Re}\{\mathcal T_o\}$. Since $\mathcal A,\mathcal C,A_{\mathrm{new}}$ and $C_{\mathrm{new}}$ are all real matrices, taking real parts in
$
\mathcal C\mathcal T_o=C_{\mathrm{new}},
$
and 
$
\mathcal T_oA_{\mathrm{new}}=\mathcal A\mathcal T_o
$
gives
\[
\mathcal C T_o=C_{\mathrm{new}},
\quad
T_oA_{\mathrm{new}}=\mathcal A T_o.
\]
This shows \eqref{imm.cond.again.cor}.

Finally, let $x_{0}\in\mathbb R^{\nu}$ be an arbitrary initial state and consider the corresponding output trajectory of $(A_{\mathrm{new}},C_{\mathrm{new}})$. By \eqref{imm.cond.again.cor}, the same output trajectory is generated by the initial condition $T_o x_{0}$ for the system $(\mathcal A,\mathcal C)$. Thus, every output trajectory of $(A_{\mathrm{new}},C_{\mathrm{new}})$ is an output trajectory of $(\mathcal A,\mathcal C)$.
Indeed, using $T_oA_{\mathrm{new}}=\mathcal A T_o$, we obtain
\[
T_oA_{\mathrm{new}}^t=\mathcal A^tT_o
\]
for all $t\geq 0$. Hence,
\[
C_{\mathrm{new}}A_{\mathrm{new}}^t x_{0}
=
\mathcal C\mathcal A^t T_o x_{0},
\]
for all $t\geq 0$. This concludes the proof.
\end{proof}

\begin{remark}[Further equivalences to the spectral condition]
We conclude this section by pointing out the equivalence of the properties in Lemma \ref{lem:an.spectr.cond} and \ref{lem:spect.cond.for.output.inclusion} to the spectral condition. Consider the case of $\mathcal{A}$ and $A_{\mathrm{new}}$ diagonalizable, 
scalar outputs ($p=1$) and sufficiently long trajectories 
($T>\operatorname{card}\bigl(\spec_{\mathcal{C}}(\mathcal{A})\bigr)$). The properties \eqref{eq:complex-visible-inclusion}, \eqref{eq:spectral-inclusion},  \eqref{imm.cond.again.cor} and \eqref{ouput-containment} are all equivalent: 
\[\begin{array}{rcccl}
&
\eqref{eq:complex-visible-inclusion} &
\stackrel{\mathrm{Lem}.\ref{lem:an.spectr.cond}}{\Longleftrightarrow} & \eqref{eq:spectral-inclusion} & \\
\mathrm{Lem.\ref{lem:spect.cond.for.output.inclusion}}
\!\!\!\!\!
&
\Downarrow & 
& \Updownarrow & \!\!\!\!\!\!\!\mathrm{Cor.\ref{cor:mix-sentinels}\;\&\; Prop.\ref{prop:necessity}}\\
&
\eqref{imm.cond.again.cor} & 
\stackrel{\mathrm{Lem}.\ref{lem:spect.cond.for.output.inclusion}}{\Longrightarrow}
& 
\eqref{ouput-containment} &
\end{array}\]
\end{remark}

\subsection{Immersible systems and prediction of nonlinear trajectories from linear templates}

The output inclusion property is closely connected with the notion of \emph{immersion} for nonlinear systems \cite{fliess2005finite,isidori1995nonlinear}. 
Consider nonlinear dynamical systems of the form
\begin{equation}
x_{t+1}=f(x_t), \quad y_t=h(x_t),
\end{equation}
where $x\in\mathcal M$ is the state and $y\in\mathbb R^p$ is the output. We refer to such systems by the triplet $(f,h,\mathcal M)$.

\begin{definition}[Immersion] \label{def:immersion}
A system $(f,h,\mathcal M)$ is said to be immersed into a 
system $(\overline f, \overline h, \overline{\mathcal M})$ if there exists a map
\[
\tau: \mathcal M \to \overline{\mathcal M},
\quad
\tau(0)=0,
\]
such that
\begin{equation}\label{imm.cond.nl}
h(x)
=
\overline h(\tau(x)),
\quad
\tau(f(x))
=
\overline f(\tau(x))
\end{equation}
for all $x \in \mathcal M$. In this case, $(\overline f, \overline h,\overline{\mathcal M})$ is called ambient or embedding system.
\end{definition}

For a system $(f,h,\mathcal M)$, define the $T$-long output trajectory generated from $x\in\mathcal M$ by
\[
\mathcal O_T(f,h,x):=
\begin{bmatrix}
h(x)\\
h(f(x))\\
\vdots\\
h(f^{T-1}(x))
\end{bmatrix},
\]
where $f^0(x)=x$ and $f^k$ denotes the $k$-fold composition of $f$. We also define
\begin{equation}\label{set.output.trajcs.finite.length.nl}
\begin{split}
\overline{\mathcal Y}_{T}
&:=
\{\mathcal O_T(\overline f, \overline h, \overline x):\overline x\in\overline{\mathcal M}\},\\[0.1cm]
\mathcal Y_{T}
&:=
\{\mathcal O_T(f,h,x): x \in\mathcal M\}.
\end{split}
\end{equation}
If \eqref{imm.cond.nl} holds, then
\[
\tau(f^k(x))
=
{\overline f}{}^k(\tau(x))
\]
for every $k\geq 0$. Consequently,
\[
h(f^k(x))
=
\overline h({\overline f}{}^k(\tau(x))),
\quad
k=0,\ldots,T-1.
\]
Thus every length-$T$ output trajectory of the immersed system is generated by the ambient system from the initial condition 
$\tau(x)$. Therefore,
\begin{equation}
\mathcal Y_{T}\subseteq \overline{\mathcal Y}_{T},
\end{equation}
which is an output inclusion property.

We now specialize our study to the case where the ambient system is \emph{finite-dimensional and linear}. In particular, we show that nonlinear systems with a finite-dimensional observation space can be represented by a finite-dimensional linear system built from the future output functions of the nonlinear system.

Consider a nonlinear system $(f,h,\mathcal M)$ with scalar output, namely
$h:\mathcal M\to\mathbb R$. For each $k\geq 0$, define the future output map
\[
\psi_k:\mathcal M\to\mathbb R,
\quad
\psi_k(x):=h(f^k(x)).
\]
The observation space of $(f,h,\mathcal M)$ is the real vector space of real-valued functions on $\mathcal M$ that are generated by these future output maps:
\[
\mathcal O
:=
\operatorname{span}\{\psi_k:k\geq 0\}.
\]
Equivalently,
\[
\mathcal O
=
\left\{
\phi(x) =
\sum_{k=0}^{s}\alpha_k h(f^k(x)),
\ 
s\in\mathbb N,\ 
\alpha_k\in\mathbb R
\right\}.
\]

\begin{definition}[Finite-dimensional observation space]
The system $(f,h,\mathcal M)$ has a finite-dimensional observation space if there exist an integer $q\geq 1$ and real numbers $a_0,a_1,\ldots,a_{q-1}$ such that
\begin{equation} \label{eq:finite-dim-obsv-space}
h(f^q(x))
+
\sum_{k=0}^{q-1} a_k h(f^k(x))
=0
\end{equation}  
for all $x\in\mathcal M$. 
We take $q$ to be the minimal integer for which such a relation holds.
\end{definition}

The minimality requirement in Definition~4 avoids introducing spurious
modes. Indeed, if
\[
m(\sigma)=\sigma^q+\sum_{k=0}^{q-1}a_k\sigma^k
\]
is the polynomial associated with the minimal relation
\[
h(f^{q}(x))+\sum_{k=0}^{q-1}a_k h(f^{k}(x))=0,
\quad
x\in \mathcal M,
\]
then any polynomial multiple $\widetilde m(\sigma)=p(\sigma)m(\sigma)$
also gives a valid, but nonminimal, relation. Its roots contain the roots
of $m$ together with the roots of $p$. Hence a nonminimal choice may add
spectral values that are not intrinsic output modes, making the condition
$\spec_h(f)\subseteq\spec_{\mathcal C}(\mathcal A)$ considered in next Proposition \ref{prop:nonlinear-output-containment}
more restrictive.

If the system has a finite-dimensional observation space, we define
\begin{equation} \label{eq:nonl-spec}
\spec_{h}(f)
:=
\left\{
\sigma\in\mathbb C:
\sigma^q + \sum_{k=0}^{q-1} a_k \sigma^{k} =0 \right\}.
\end{equation}

Let $(\mathcal A,\mathcal C)$ be the linear template system associated with the library used for prediction. In the aggregate case of Section~\ref{subsec:mix-agg}, $(\mathcal A,\mathcal C)=(A_{\mathrm{mix}},C_{\mathrm{mix}})$; in the interconnected case, $(\mathcal A,\mathcal C)$ is obtained from the corresponding interconnection construction.  \\
We define
\begin{equation}\label{HY}
H_Y:=\mathcal{O}_T(\mathcal A,\mathcal C)\mathcal X,
\end{equation}
where $\mathcal{O}_T(\mathcal A,\mathcal C)\in\mathbb R^{pT\times n}$ and $\mathcal X$ is the matrix whose columns are the initial states used to generate the sample trajectories in $H_Y$. Recall the following standard relation.

\begin{lemma}\label{lem.rich.data}
If $\mathcal X$ has full row rank, then the trajectory library satisfies 
$
\im(H_Y) =\{\mathcal O_T(\mathcal A,\mathcal C,x):x\in\mathbb R^n\}.
$
\end{lemma}

We then have the following result. 

\begin{proposition}[Exact prediction for nonlinear systems]\label{prop:nonlinear-output-containment}
Assume that:
\begin{enumerate}
\item The system $(f,h,\mathcal M)$ has a finite-dimensional observation 
space and $\spec_{h}(f)$ is made of distinct elements.
\item The linear system $(\mathcal A,\mathcal C)$ is diagonalizable.
\item Both systems have scalar outputs, i.e., $p=1$.
\item The matrix $\mathcal X$ in \eqref{HY} has full row rank.
\end{enumerate}
Let $\mathcal Y_{T}$ denote the set of length-$T$ output trajectories of the nonlinear system $(f,h,\mathcal M)$ defined as in \eqref{set.output.trajcs.finite.length.nl}, and let $H_Y$ be the library in \eqref{HY}. If
\[
\spec_{h}(f)
\subseteq
\spec_{\mathcal C}(\mathcal A),
\]
then
\[
\mathcal Y_{T}\subseteq\im(H_Y).
\]
Accordingly, let $r\geq s_{\mathrm{lib}}$ where  
$s_{\mathrm{lib}}$
is the observability index of $(\mathcal A,\mathcal C)$. 
Then, for every length-$T$ trajectory
\[
Y^\star=
\begin{bmatrix}
Y_{\mathrm{past}}\\
Y_{\mathrm{fut}}
\end{bmatrix}
\in\mathcal Y_{T},
\]
any solution $g$ of
\[
H_pg=Y_{\mathrm{past}}
\]
satisfies
\[
H_fg=Y_{\mathrm{fut}}.
\]
Hence the predicted future 
$
\widehat Y_{\mathrm{fut}}:=H_fg
$
is exact.
\end{proposition}

\begin{proof}
Define
$
\tau:\mathcal M \to\mathbb R^q
$
by
\[
\tau(x):=
\begin{bmatrix}
h(x)\\
h(f(x))\\
\vdots\\
h(f^{q-1}(x))
\end{bmatrix}.
\]
By assumption, condition \eqref{eq:finite-dim-obsv-space} holds 
for all $x\in\mathcal M$. Hence $\tau$ satisfies the immersion identities \eqref{imm.cond.nl} with the finite-dimensional linear ambient system
$(\overline A,\overline C)$  where
\[
\overline C=
\begin{bmatrix}
1&0&\cdots&0
\end{bmatrix} \in \mathbb{R}^{1 \times q}
\]
and
\[
\overline A=
\begin{bmatrix}
0&1&0&\cdots&0\\
0&0&1&\cdots&0\\
\vdots&\vdots&\vdots&\ddots&\vdots\\
0&0&0&\cdots&1\\
-a_0&-a_1&-a_2&\cdots&-a_{q-1}
\end{bmatrix} \in \mathbb{R}^{q \times q}.
\]
Indeed, 
$
\overline C\tau(x)=h(x)
$,
and
\[
\tau(f(x))
=
\begin{bmatrix}
h(f(x))\\
h(f^2(x))\\
\vdots\\
h(f^q(x))
\end{bmatrix}
=
\overline A\tau(x).
\]
Hence, system $(f,h,\mathcal M)$ is immersed into the linear system 
$(\overline A,\overline C,\mathbb R^q)$ according to Definition \ref{def:immersion}. If
\[
\overline{\mathcal Y}_T
:=
\{\mathcal O_T(\overline A,\overline C,z):z\in\mathbb R^q\},
\]
then the immersion property gives
\[
\mathcal Y_{T}\subseteq \overline{\mathcal Y}_T.
\]

We now compare the linear system $(\overline A,\overline C)$ with the template system $(\mathcal A,\mathcal C)$. 
The pair $(\overline A,\overline C)$ is observable. Moreover, $\overline A$ is in companion form, and its characteristic polynomial is
\[
\lambda^q+a_{q-1}\lambda^{q-1}+\cdots+a_1\lambda+a_0.
\]
Hence, its eigenvalues are precisely the elements of $\spec_{h}(f)$ in \eqref{eq:nonl-spec}. Since these elements are distinct by assumption, $\overline A$ is diagonalizable. Observability implies that all the eigenvalues of $\overline A$ are $\overline C$-visible, and thus
\[
\spec_{\overline C}(\overline A)
=
\spec(\overline A)
=
\spec_{h}(f).
\]
Thus, the condition
$
\spec_{h}(f)
\subseteq
\spec_{\mathcal C}(\mathcal A)
$
is equivalent to
\[
\spec_{\overline C}(\overline A)
\subseteq
\spec_{\mathcal C}(\mathcal A).
\]
Since both systems have scalar outputs, Lemma \ref{lem:an.spectr.cond} implies that
\[
\overline C\bigl(\ker(\mu I-\overline A)\bigr)
\subseteq
\mathcal C\bigl(\ker(\mu I-\mathcal A)\bigr)
\]
for every $\mu\in\spec(\overline A)$. Therefore, by Lemma \ref{lem:spect.cond.for.output.inclusion}, every output trajectory of $(\overline A,\overline C)$ is also an output trajectory of $(\mathcal A,\mathcal C)$. Thus
$
\overline{\mathcal Y}_T\subseteq \{\mathcal O_T(\mathcal A,\mathcal C,x):x\in\mathbb R^n\}.
$
Combining this inclusion with $\mathcal Y_{T}\subseteq\overline{\mathcal Y}_T$, we obtain
\[
\mathcal Y_{T}\subseteq \{\mathcal O_T(\mathcal A,\mathcal C,x):x\in\mathbb R^n\}.
\]
By assumption, $\mathcal X$ has full row rank. Then Lemma \ref{lem.rich.data} implies that
$
\{\mathcal O_T(\mathcal A,\mathcal C,x):x\in\mathbb R^n\}=\im(H_Y).
$
Therefore,
\[
\mathcal Y_{T}\subseteq\im(H_Y).
\]

It remains to prove that the prediction mechanism is correct. 
From the inclusion just proved, $Y^\star\in\im(H_Y)$. Hence there exists $g^\star$ such that
$
Y^\star=H_Yg^\star.
$
Equivalently, 
\[
Y_{\mathrm{past}}=H_pg^\star,
\quad
Y_{\mathrm{fut}}=H_fg^\star.
\]
Since $Y^\star\in\im(H_Y)$ is an admissible $T$-long trajectory of the template system $(\mathcal A,\mathcal C)$, Lemma \ref{lem:key} applies and the condition $r\geq s_{\mathrm{lib}}$ ensures that $\widehat Y_{\mathrm{fut}}:=H_fg$ gives the correct prediction for any $g$ satisfying $H_pg=Y_{\mathrm{past}}$.
\end{proof}

\begin{example}[Nonlinear prediction from linear templates]
The example is taken from \cite{Surana2016KoopmanObservers}.
Consider the nonlinear system with state and output maps
\begin{equation}
\label{exmpl:nonl}
f(x) =
\begin{bmatrix}
\rho x_1\\
\mu x_2+(\rho^2-\mu)c x_1^2
\end{bmatrix}, \quad 
h(x)=x_1^2+x_2 .
\end{equation}
where $c,\rho,\mu$ are constant parameters. 

We first verify that the system has a finite-dimensional observation space. Define
\[
z_1(x):=x_2-cx_1^2,
\quad
z_2(x):=x_1^2.
\]
Then
\[
z_1(f(x))=\mu z_1(x),
\quad
z_2(f(x))=\rho^2 z_2(x).
\]
Iterating these identities gives
\[
z_1(f^k(x))=\mu^k z_1(x),
\quad
z_2(f^k(x))=\rho^{2k} z_2(x).
\]
Moreover,
\[
h(x)=z_1(x)+(c+1)z_2(x).
\]
Consequently, for every $k\geq 0$,
\[
h(f^k(x))
=
\mu^k(x_2-cx_1^2)
+
(c+1)\rho^{2k}x_1^2 .
\]
Thus, the output functions are generated by the two modes $\mu$ and $\rho^2$.  
In particular, if $\mu\neq\rho^2$ and $c+1\neq 0$, \footnote{The case when $\mu=\rho^2$ or $c+1= 0$ would give $q=1$. In this example we examine the case $q=2$.} the observation space has order $q=2$
and the output functions satisfy
\[
h(f^2(x))
-(\mu+\rho^2)h(f(x))
+\mu\rho^2 h(x)=0.
\]
Hence, 
\[
\spec_{h}(f)=\{\sigma\in \mathbb{C}\colon \sigma^2 -(\mu+\rho^2) \sigma
+\mu\rho^2 =0\},
\]
which is made of distinct elements.

We construct a library by aggregating $30$ sample trajectories generated by three second-order template systems ($10$ trajectories per system).
The systems are given by
\[
A_1=\begin{bmatrix} \rho & 0 \\ 0 & 0.15 \end{bmatrix}, \,
A_2=\begin{bmatrix} \tilde \mu & 0 \\ 0 & -0.25 \end{bmatrix}, \,
A_3=\begin{bmatrix} \tilde \rho^2 & 0 \\ 0 & 0.35 \end{bmatrix},
\]
with 
\[
\widetilde\mu=\mu+\delta,
\quad
\widetilde{\rho}=\rho+\delta,
\]
where $\delta\geq 0$ is a scalar that 
accounts for a mismatched library. 
The output matrix is assumed to be the same for all the template systems: $C=\begin{bmatrix} 1 & 1 \end{bmatrix}$. 

When $\delta=0$, the library contains the visible eigenvalues of the ambient system, namely $\mu$ and $\rho^2$, together with additional modes that are not needed to represent the nonlinear output. Therefore, the condition $\spec_{h}(f) \subseteq
\spec_{\mathcal C}(\mathcal A)$ is satisfied, and we can apply Proposition~\ref{prop:nonlinear-output-containment}.
For a trajectory length $T=10$ and a past window length $r=6$ ($r\geq s_{\mathrm{lib}}=6$), Proposition~\ref{prop:nonlinear-output-containment} implies that the future output window is recovered exactly from the template library.

Next, we consider a mismatched library. To that purpose, we recall a few facts from the proof of  Proposition~\ref{prop:nonlinear-output-containment}. We define the immersion map
\[
\tau(x):=
\begin{bmatrix}
h(x)\\
h(f(x))
\end{bmatrix}.
\]
The recurrence above gives
\[
\tau(f(x))
=
\begin{bmatrix}
h(f(x))\\
h(f^2(x))
\end{bmatrix}
=
\overline A\tau(x),
\]
where
\[
\overline A=
\begin{bmatrix}
0&1\\
-\mu\rho^2&\mu+\rho^2
\end{bmatrix},
\quad
\overline C=
\begin{bmatrix}
1&0
\end{bmatrix}.
\]
Moreover,
\[
h(x)=\overline C\tau(x).
\]
Hence every output trajectory of the nonlinear system is also an output trajectory of the linear ambient system $(\overline A,\overline C)$. The eigenvalues of $\overline A$ are precisely $\mu$ and $\rho^2$. Consequently, any error bound obtained by applying
Corollary~\ref{cor:out-of-lib-error} to the ambient system
$(\overline A,\overline C)$ also applies to the nonlinear trajectories,
because these trajectories form a subset of the output trajectories of
the ambient system.

We take $\rho=0.8$, $\mu=0.5$, and $c=0.7$. For each value of $\delta$, the library is rebuilt by simulating the same template systems with the perturbed eigenvalues. The sample initial conditions are kept fixed, so the variation in the prediction error is due only to the spectral mismatch and not to a change in the sampled data.
For each $\delta$, the nonlinear system is still immersed into the same ambient system $(\overline A,\overline C)$, but this ambient system is no longer exactly represented by the template library. In this case, Corollary~\ref{cor:out-of-lib-error} can be applied by comparing the ambient system $(\overline A,\overline C)$ with the aggregate system built from the template systems $(A_i,C)$. The observability index $\overline s$ of $(\overline A,\overline C)$ is $2$, hence $r\ge \overline s$, as requested by Corollary~\ref{cor:out-of-lib-error}. 
For each $\delta$, we test $N_r=200$ initial states uniformly distributed within $[-3,3]^2$. For the $i$-th trajectory, let $Y_{\mathrm{past}}^{(i)}$ and $Y_{\mathrm{fut}}^{(i)}$ denote the true past and future output windows, and let $\widehat Y_{\mathrm{fut}}^{(i)}$ be the future predicted from the mismatched linear template library. The prediction error is
\[
e_{\mathrm{pred}}^{(i)}(\delta):=\left\|
Y_{\mathrm{fut}}^{(i)} - \widehat Y_{\mathrm{fut}}^{(i)}
\right\|_2 .
\]

Figure~\ref{fig:nonlinear-immersion-example} shows one representative nonlinear trajectory and the corresponding prediction obtained from the mismatched library obtained with $\delta=0.05$. Table~\ref{tab:nonlinear-mismatch-bound} reports, for increasing values of $\delta$, the average and worst-case prediction errors over the $N_r$ tested initial conditions.

\end{example}

\begin{figure}[t]
\centering
\includegraphics[width=0.85\linewidth]{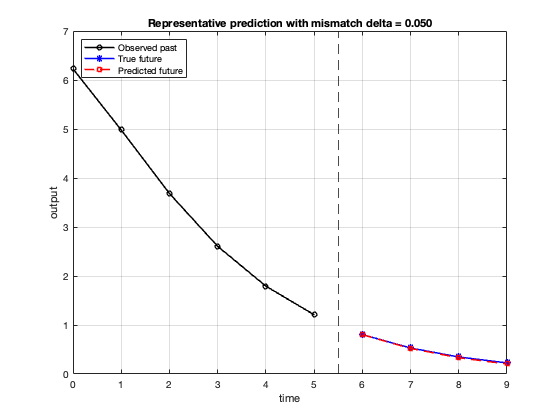}
\caption{Prediction of a nonlinear output trajectory from a mismatched linear template library.}
\label{fig:nonlinear-immersion-example}
\end{figure}

\begin{table}[t]
\centering
\caption{Prediction error for increasing spectral mismatch.}
\label{tab:nonlinear-mismatch-bound}
\begin{tabular}{c c c c c c}
\hline
 & $\delta=0$ & $\delta=0.025$ & $\delta=0.050$ & $\delta=0.075$ & $\delta=0.100$ \\
\hline \hline
$\overline e_{\mathrm{pred}}$ 
& $0$ & $5.6\cdot10^{-3}$ & $1.0\cdot10^{-2}$ & $1.4\cdot10^{-2}$ & $1.7\cdot10^{-2}$ \\ \hline
$e_{\mathrm{pred}}^{\max}$ 
& $0$ & $1.7\cdot10^{-2}$ & $3.2\cdot10^{-2}$ & $4.5\cdot10^{-2}$ & $5.5\cdot10^{-2}$  \\
\hline
\end{tabular}
\end{table}

\section{Conclusions}

This paper developed a framework for trajectory prediction from stored
templates. We showed that libraries of output trajectories can act as
prediction machines by generating behavioral spaces from which future
outputs are inferred from observed past windows. For linear systems, we
characterized exact prediction through continuation maps, behavioral
containment, and spectral inclusion conditions, and we quantified the
effect of noise and out-of-library mismatch through explicit error
bounds. We also showed that template libraries can generalize
compositionally through interconnection constraints, producing emergent
modes not present in the original atomic libraries, and that the same
prediction principle extends to nonlinear systems whose output behaviors
are contained in, or immersed into, finite-dimensional linear behaviors.

Several directions remain open. A first one is the extension of the
theory to broader classes of nonlinear systems, beyond exact output
containment or finite-dimensional linear immersions. A second direction is
the design of template libraries with prescribed prediction guarantees.
When exact prediction is impossible, Proposition~\ref{prop:necessity}
shows that spectral inclusion is a fundamental obstruction; in that
regime, one may instead seek libraries that cover a desired family of
target behaviors up to a prescribed error threshold. This leads naturally
to covering and library-design problems: how many templates are needed,
where they should be placed, and how they should be selected so as to
guarantee a desired bound on the prediction error. Other important directions are  the derivation of  analogous results for systems with inputs, the development  of on-line recursive forms for the machines and the  design of templates for control purposes.

\section*{Appendix A. Proof of Fact \ref{fact:rank.stability}}
For compactness set $\mathcal{O}_k:=\mathcal{O}_k(A,C)$. 
By definition of observability index, for every \(k \ge s\) we have
\[
\im\!\bigl(\mathcal O_k^\top\bigr)
=
\im\!\bigl(\mathcal O_s^\top\bigr).
\]
In particular, since $\ell+1 \geq r \geq s$, it follows that
\[
\im\!\bigl(\mathcal O_{\ell+1}^\top\bigr)
=
\im\!\bigl(\mathcal O_s^\top\bigr)
=
\im\!\bigl(\mathcal O_r^\top\bigr).
\]
Since $CA^\ell$ is the last block row of $\mathcal O_{\ell+1}$, every row of $CA^\ell$ belongs to
$
\im\!\bigl(\mathcal O_{\ell+1}^\top\bigr)
=
\im\!\bigl(\mathcal O_r^\top\bigr).
$
Therefore, every row of \(CA^\ell\) can be written as a linear combination of the rows of $\mathcal O_r$. 
Collecting these coefficients row by row, there exists a matrix $Q_{\ell}$ such that
$
CA^\ell = Q_{\ell} \mathcal O_r.
$
This concludes the proof.

Note that $Q_{\ell}$ generally depends on $r$ (equivalently, on $\mathcal O_r$).
We omit the dependence on $r$ in the notation $Q_\ell$, since $r$ is fixed throughout the paper.

\section*{Appendix B. Least squares solutions}

\begin{theorem}\label{thm.least.squares}
The following hold:
\begin{enumerate}
\item The vector $\hat g\in \R^N$ solves $\min_{g\in \R^N} \frac{1}{2}\|Z_{\mathrm{past}} -H_p g\|_2^2$ if and only if $H_p^\top H_p \hat g =H_p^\top Z_{\mathrm{past}}$. 
\item 
A solution to $H_p^\top H_p \hat g =H_p^\top Z_{\mathrm{past}}$ always exists and is given by $\hat g = H_p^\dag Z_{\mathrm{past}} + (I_N-H_p^\dag H_p)\gamma$, where $\gamma$ is any vector in $\R^{N}$ and $(\cdot)^\dagger$ is the Moore--Penrose
pseudoinverse 
\end{enumerate}
\end{theorem}

\bibliographystyle{IEEEtran}

\bibliography{refs}

 \end{document}